\documentclass{amsart}
\usepackage{amsbsy,amssymb,amscd,amsfonts,latexsym,amstext,delarray, amsmath,graphicx,color,caption,enumerate,subcaption,wrapfig,mathtools,adjustbox}
\usepackage{tikz,tikz-cd}
\usetikzlibrary{matrix,calc,arrows,decorations.markings,patterns,calc,decorations.pathmorphing}
\usepackage{xcolor}

\newtheorem{theorem}{Theorem}
\newtheorem{definition}[theorem]{Definition}
\newtheorem{proposition}[theorem]{Proposition}
\newtheorem{lemma}[theorem]{Lemma}
\newtheorem{corollary}[theorem]{Corollary}

\newtheorem{remark}[theorem]{Remark}

\newcommand{\Tau}{\tau}

\newcommand{\cM}{\mathcal M}
\newcommand{\cP}{\mathcal P}
\newcommand{\cB}{\mathcal B}
\newcommand{\cA}{\mathcal A}
\newcommand{\cJ}{\mathcal J}
\newcommand{\cZ}{\mathcal Z}
\newcommand{\cO}{\mathcal O}
\newcommand{\CC}{\mathbb C}
\newcommand{\ZZ}{\mathbb Z}

\newcommand{\pa}{\partial}

\newcommand{\Spec}{\mathrm{Spec}\,}
\newcommand{\git}{\mathbin{
  \mathchoice{/\mkern-6mu/}
    {/\mkern-6mu/}
    {/\mkern-5mu/}
    {/\mkern-5mu/}}}

\tikzset{
->-/.style={decoration={
  markings,
  mark=at position .5 with {\arrow{>}}},postaction={decorate}},
->--/.style={decoration={
  markings,
  mark=at position .3 with {\arrow{>}}},postaction={decorate}},
-->-/.style={decoration={
  markings,
  mark=at position .7 with {\arrow{>}}},postaction={decorate}},
-<-/.style={decoration={
  markings,
  mark=at position .5 with {\arrow{<}}},postaction={decorate}},
-<--/.style={decoration={
  markings,
  mark=at position .3 with {\arrow{<}}},postaction={decorate}},
--<-/.style={decoration={
  markings,
  mark=at position .7 with {\arrow{<}}},postaction={decorate}},
subset/.style={
    draw=none,
    every to/.append style={
      edge node={node [sloped, allow upside down, auto=false]{$\subset$}}}
  },
equal/.style={
    draw=none,
    every to/.append style={
      edge node={node [sloped, allow upside down, auto=false]{$=$}}}
  }
assign/.style={
    draw=none,
    every to/.append style={
      edge node={node [sloped, allow upside down, auto=false]{$=$}}}
  }
}
\tikzset{
  symbol/.style={
    draw=none,
    every to/.append style={
      edge node={node [sloped, allow upside down, auto=false]{$#1$}}}
  }
}

\title{Superintegrable Systems on Moduli Spaces of Flat Connections}

\author{S.~Arthamonov}
\address{S.A.: Department of Mathematics, University of California, Berkeley,
CA 94720, USA \& ITEP, Moscow, Russia}
\email{artamonov@berkeley.edu}

\author{N.~Reshetikhin}
\address{N.R.: Department of Mathematics, University of California, Berkeley,
CA 94720, USA \& Physics Department, St. Petersburg University, Russia \&KdV Institute for Mathematics, University of Amsterdam,
Science Park 904, 1098 XH Amsterdam, The Netherlands.}
\email{reshetik@math.berkeley.edu}

\begin{document}

\begin{abstract}
The main result of this paper is the construction of a family of superintegrable Hamiltonian systems on
moduli spaces of flat connections on a principle $G$-bundle on
a surface. The moduli space is a Poisson variety with Atiyah-Bott Poisson structure. Among particular cases of such systems are spin
generalizations of Ruijsenaars-Schneider models.
\end{abstract}

\maketitle

\section*{Introduction}

Let $G$ be a simple linear algebraic group over $\mathbb C$. In this paper we construct superintegrable Hamiltonian systems on moduli spaces of flat $G$-connections over any oriented surface with nonempty boundary. Hamiltonians of such systems are traces of holonomies along non intersecting non self-intersecting curves. The construction naturally works in the same way for various real forms of $G$, for example for compact simple Lie groups or for split real forms.

The fact that Hamiltonian integrable systems appear in the
context of gauge theories is not new, with Hitchin systems \cite{H}
being one of the well known examples. Calogero-Moser type systems were put in the
context of gauge theory earlier 1990's,  see  \cite{GN1}, \cite{GN2} and references therein.
See \cite{FK}\cite{AO} for more recent devolopments and references.

The main result of this paper involves the notion of a superintegrable Hamiltonian system
which we review in section \ref{SI}, and of graph functions which we review in section \ref{MGF}.
Denote the moduli space of flat $G$-connections on
$\Sigma$ as $\mathcal M^G_\Sigma$. Our main result is the following theorem.

\begin{theorem}\label{theorem:SuperIntSystemsCyclesIntro}
Let $\Sigma$ be an oriented surface of genus $g$ with $b>0$ boundary components and $G$ be a simple linear algebraic group over $\mathbb C$. For each disjoint union of pairwise nonhomotopic simple closed curves $C=C_1\sqcup\dots\sqcup C_k$ in $\Sigma$, none of which is homopotic to a boundary component, the following inclusions of Poisson algebras define an affine superintegrable system
\begin{align*}
    Z_{\partial\Sigma}\subset B_C\subset J_{\Sigma\backslash C}\subset \mathcal O[\mathcal M^G_\Sigma].
\end{align*}
Here $Z_{\partial\Sigma}$ is the algebra of $G$-invariant functions on holonomies around boundary components, $B_C$ is a subalgebra of the coordinate ring $O[\mathcal M^G_\Sigma]$ generated by graph functions $F_{\Gamma,V,c}$ with $\Gamma$ being contractible to $C\sqcup\partial\Sigma$. Similarly, $J_{\Sigma\backslash C}$ is a subalgebra generated by graph functions $F_{\Gamma,V,c}$ with $\Gamma\subset (\Sigma\backslash C)$.
\end{theorem}

We also introduce the notion of a {\it refinement} of one superintegrable systems
by the other. The refinement defines a partial order on the space of superintegrable
systems.

For a torus with one puncture these systems are spin Ruijsenaars systems (also known as relativistic spin Calogero-Moser systems) \cite{FGNR}\cite{Reshetikhin'2016}.
For a torus with two punctures they are relativistic deformations
of two sided spin Calogero-Moser systems \cite{R2} corresponding to
symmetric pair $G\subset G\times G$ where $G$ is embedded diagonally. A particular case of our systems on a torus with $n$ punctures and rank one conjugation orbits of $SL_n$ is closely related to the systems studied in \cite{CF}. In this case the systems are actually Liouville integrable and they admit a superintegrable refinement. The maximal integrability in this case is achieved by a superintegrable system
with $r=\mathrm{rank}\,G$ independent Poisson commuting integrals. The minimal integrability is
a Liouville integrability with $b\,\mathrm{rank}\,G$ Poisson commuting integrals. In our framework
this correspond to a different choice of cycles $C$.

The mapping class group acts on the moduli space of flat connections by Poisson automorphisms. If $C$ and $C'$ are collections of cycles on surface $\Sigma$ which belong to the same orbit of the Mapping Class Group action, the corresponding superintegrable systems are isomorphic. An example of such isomorphism in the genus one case for involution exchanging the equator and meridian of the torus gives rise to a self-duality of the Ruijsenaars system \cite{FGNR,FockRosly'1999,FeherKlimcik'2011}. For a spin version of this duality see \cite{Reshetikhin'2016}.

The plan of this paper is as follows.
In the first section we recall the definition of superintegrable (degenerately integrable systems).
There we also define the notion of an affine superintegrable system in the algebro-geometrical setting.
The second section  is an overview of basic
notions about moduli spaces of flat connections on a surface. In this section we also recall the definition of
graph functions and the description of Poisson brackets between two such functions. In the third section we describe
the main result, the construction of a family of Hamiltonian systems defined by a choice of non intersecting, non selfinersection cycles and prove their superintegrability. At the end of this section we introduce
the notion of a a partial order on such systems.
In section 4 we explain how solutions to equations of motion of these superintegrable systems
can be solved using the projection method.  Section five has some  genus one examples.
In the conclusion we define a conjecturally superintegrable system on the space of chord diagrams,
discuss the case of non-generic conjugation orbits and some further directions.

{\bf Acknowledgements.} N.R. is grateful to J.~Stokman for numerous discussions. S.A. is grateful to V.~Roubtsov and L.~Feher for useful remarks. The work of N.R. was partially supported by grants NSF DMS-1601947 and RSF-18-11-00-297. The work of S.A. was partially supported by RSF-18-01-00926 and 19-51-50008-YaF.

\section{Superintegrable Systems}\label{SI}

\subsection{Superintegrable systems} For an overview of superintegrable systems see for example \cite{SInt}\cite{Reshetikhin'2016}.
Here we recall briefly basic definitions and introduce the notion of an affine
superintegrable system.

The notion of a superintegrable Hamiltonian systems were introduced in \cite{Nekhoroshev'1972} (where he called them degenerate integrable systems) as a generalization of the Liouville integrable systems. Superintegrable system on a $2n$ dimensional smooth symplectic manifold has $m>n$ independent\footnote{On a smooth manifold $\mathcal M$ (complex or real) we call functions $f_1,\dots,f_k\in C^\infty(\mathcal M)$ \textit{independent} if the corresponding differentials $\mathrm df_1,\dots,\mathrm df_k\in \Gamma(T^*\mathcal M)$ are linearly independent at every point in $\mathcal M$.} first integrals, however only $k$ of them are in involution and $r+k=2n$. A particular case when $m=k=n$ corresponds to Liouville integrable systems.  First examples of such systems
appeared earlier \cite{Pa}\cite{F}\cite{FMSW}. A family of examples of superintegrable systems associated to Lie groups was given  in \cite{MishchenkoFomenko'1978}. More recent examples include charateristic systems on simple Poisson-Lie groups \cite{R3}, spin Calogero-Moser and Ruijsenaars systems \cite{Reshetikhin'2003} and their relativistic counterparts \cite{Reshetikhin'2016}.

Geometrically, a superintegrable system of rank $k$ on a symplectic manifold $\cM_{2n}$ consists of a Poisson manifold $\cP_{2n-k}$,
a manifold $\cB_k$\footnote{Typically $\cB_k$ is not smooth but stratified by smooth strata
with one smooth higher dimensional stratum. In is frequently an orbifold.} (considered as a Poisson manifold with trivial symplectic structure) and Poisson projections

\begin{align}
\cM_{2n}\stackrel{p_1}{\longrightarrow}\cP_{2n-k}\stackrel{p_2}{\longrightarrow}\cB_k
\label{eq:PoissonProjections}
\end{align}

Denote by $A=C^\infty(\cM_{2n})$ the algebra of smooth functions on $\cM_{2n}$. It is equipped with a Poisson bracket $\{,\}$ determined by the symplectic structure. In terms of functions on $\cM_{2n}$ superintegrable system is determined by two
Poisson subalgebras of $A$:
\begin{align}
B\subset J\subset A
\label{eq:SuperintegrableSystemDef}
\end{align}
where $J=C^\infty(\cP_{2n-k})$ and $B=C^\infty(\cB_k)$. The subalgebra
$B$ is a Poisson commutative subalgebra and $J$ is the Poisson centralizer of $B$, i.e. $J$ is the maximal subalgebra
in $A$ such every element from $J$ Poisson commutes with every element from $B$.

Elements of subalgebra $B$ define derivations of $A$
\begin{align}
D_b: A\rightarrow A,\qquad D_b(a)=\{b,a\}, \ \ D_b(a_1a_2)=a_1D_b(a_2)+D_b(a_1)a_2, \ \ b\in B, \ \ a,a_1,a_2\in A.
\label{eq:HamiltonFlowDef}
\end{align}
Because $B$ is Poisson commutative, these derivations commute
\begin{align*}
\qquad D_{b_1}D_{b_2}-D_{b_2}D_{b_1}=0
\end{align*}
Derivation $D_b$ is the Lie derivative along the Hamiltonian vector field on $\cM_{2n}$ generated by $b$.
We refer to elements of $B$ as \textit{Hamiltonians} of a superintegrable system (\ref{eq:SuperintegrableSystemDef}).
On the other hand, elements of subalgebra $J$ are precisely the first integrals of the Hamilton flow
generated by $b$, i.e. they are constant on flow lines of the Hamiltonian vector field generated by $b$.

Fix $b_1,\dots,b_k\in B$, a choice of independent Hamiltonians. Corresponding vector fields $D_{b_1},\dots,D_{b_k}\in T\mathcal M_{2n}$ are independent at every point $x\in\mathcal M_{2n}$ and hence form a basis on the $k$-dimensional subspace of $T_x\mathcal M_{2n}$ defined by the level set of $J$. Algebraically, these vector fields $D_{b_1},\dots,D_{b_k}\in Der_J(A)$ define derivations of $A$ relative to $J$. Balance of dimensions in (\ref{eq:PoissonProjections}) is equivalent to the fact that $Der_J(A)$ is generated as an $A$-module by Hamiltonian vector fields of the form $D_b$ for $b\in B$.

\subsection{Affine superintegrable systems} In many interesting examples superintegrable systems appear in families on affine Poisson varieties, we refer to such families as \textit{affine superintegrable systems}.

\begin{definition}
Let $\cA$ be an affine variety\footnote{More generally, one can take $\mathcal A=\Spec A$ to be an integral affine scheme.} over $\mathbb C$ with a ring of regular functions $ A=\mathcal O[\cA]$ equipped with a Poisson bracket $\{,\}:A\otimes A\rightarrow A$ of maximal rank $2n$. We say that a chain of inclusions of finitely generated Poisson subalgebras
\begin{align}
Z\subset B\subset J\subset A
\label{eq:ZBJAinclusion}
\end{align}
is an \textbf{affine superintegrable system} if
\begin{align*}
\{Z,A\}=0=\{B,J\}
\end{align*}
and
\begin{align}
\dim\mathcal B+\dim\mathcal J-2\dim\mathcal Z=2n=\dim\mathcal A-\dim\mathcal Z,
\label{eq:SuperIntDim}
\end{align}
where $\mathcal J=\Spec J,\;\mathcal B=\Spec B,\;\mathcal Z=\Spec Z$ stand for the spectrum of prime ideals of $J,B,Z$ respectively.
\label{def:SuperIntDim}
\end{definition}

Inclusions (\ref{eq:ZBJAinclusion}) are equivalent to the chain of dominant maps of affine schemes which preserve the Poisson bracket
\begin{align*}
\cA\stackrel{p_1}{\longrightarrow}\cJ \stackrel{p_2}{\longrightarrow} \cB\stackrel{p_3}{\longrightarrow}\cZ.
\end{align*}
Definition \ref{def:SuperIntDim} implies that general fibers of $p_2$ are symplectic leaves of $\mathcal J$\footnote{More precisely, $p_2^{-1}(b),\;b\in\mathcal B$ is an algebraic subset of $\mathcal J$, it contains an irreducible component of maximal dimension which is an affine Poisson variety. For a general $b$, this variety is equipped with a Poisson bracket of maximal rank equal to the dimension.}, while general fibers of the composition map $p$
\begin{equation*}
\begin{tikzcd}
\cA\arrow[r, "p_1"]\arrow[rrr,bend right=17,"p"']&\mathcal J\arrow[r,"p_2"]&\mathcal B\arrow[r,"p_3"]&\mathcal Z
\end{tikzcd}
\end{equation*}
are symplectic leaves of $\mathcal A$.

As a corollary, open subset $\mathcal M_{2n}(z)$ of irreducible component of a general fiber $p^{-1}(z),\;z\in\mathcal Z$ can be viewed as a phase space of a superintegrable system. $\mathcal M_{2n}(z)$ comes equipped with two poisson projections of the form (\ref{eq:PoissonProjections}). Thus, as we already mentioned, affine superintegrable systems should be regarded as families of superintegrable systems.

\subsection{Refinement of a superintegrable system}

Here we will introduce the notion of a {\it refinement} of a superintegrable system. It gives
a partial order on all superintegrable systems on a given Poisson algebra $A$.
A minimal object in this partial order is known as a maximally superintegrable
system. Maximal superintegrability means that all invariant tori are one dimensional \cite{SInt}.

Let
\begin{equation}\label{2sys}
    Z\subset B_1\subset J_1\subset A,\qquad Z\subset B_2\subset J_2\subset A
\end{equation}
be two superintegrable systems on a Poisson algebra $A$ with the Poisson center $Z$.

\begin{definition}\label{refnmnt}
The first superintegrable system is a refinement of the second if the following
chain of Poisson inclusions hold:
\[
Z\subset B_1\subset B_2\subset J_2\subset J_1\subset A
\]
\end{definition}

Clearly a refinement defines a partial order on the set of integrable systems.

Another important relation is an equivalence of integrable systems

\begin{definition}
Two integrable systems (\ref{2sys}) are {\it equivalent} if there is a Poisson automorphism $\varphi:A\to A$
such that $J_2=\varphi(J_1)$ and $B_1=\varphi(B_2)$.
\end{definition}

\section{Moduli Spaces of Flat Connections}
\label{sec:ModuliSpaceOfFlatConnections}

\subsection{Character Variety of the Fundamental Group}

Let $\Sigma_{g,b}$ be an oriented surface of genus $g$ with $b$ boundary components. We will consider only
surfaces with $b>0$. The fundamental group $\pi_1(\Sigma,p)$ with the base point $p$ of such surface is generated by elements $x_1,\cdots, x_g, y_1,\cdots, y_g,z_1,\cdots z_{b}$ with one defining relation:
\begin{align*}
\pi_1(\Sigma_{g,b})=\left\langle x_1,\dots,x_g,y_1,\dots,y_g,z_1,\dots,z_b\,\big|\,[x_1,y_1]\dots[x_g,y_g]z_1\dots z_b=\mathrm{Id}_{p} \right\rangle.
\end{align*}
Here $x_i$ and $y_i$ are fundamental cycles for $\Sigma$ and $z_i$ is the cycle for $i$-th boundary component.\footnote{Throughout the text we always assume that the composition of paths is read from right to left. For example, $[x_1,y_1]z_1z_2$ corresponds to the contractible path along the boundary of a disc in $\Sigma_{1,2}$ on Figure \ref{fig:ChoiceOfArcsG1B2}.}
\begin{figure}
\begin{tikzpicture}[scale=1.3]
  \draw[thick, -->-] (1.5,0.2) to[out=155,in=0] (0,0.9) to[out=180,in=90] (-1,0.1) to[out=270,in=180] (0,-0.7) to[out=0,in=225] (1.5,0.2);
  \draw[thick,--<-] (1.5,0.2) to[out=100,in=-20] (0.9,0.95) to[out=160,in=220,looseness=0.6] (0,1.5);
  \draw[thick,dashed] (0,1.5) to[out=0,in=0,looseness=0.5] (0,0.2);
  \draw[thick] (0,0.2) to[out=180,in=200] (0.2,0.4) to[out=20,in=180] (1.5,0.2);
  \draw[thick,->-] (1.5,0.2) to[out=260,in=90] (1.3,-0.5) to[out=270,in=180] (1.6,-0.8) to[out=0,in=270] (1.9,-0.5) to[out=90,in=280] (1.5,0.2);
  \draw[thick,->-] (1.5,0.2) to[out=300,in=90] (2.2,-0.5) to[out=270,in=180] (2.5,-0.8) to[out=0,in=270] (2.8,-0.5) to[out=90,in=320] (1.5,0.2);
  \fill (1.5,0.2) circle (0.06);
  \fill[gray] (1.5,0.2) circle (0.04);
  \draw[ultra thick] (0,1.5) to[out=0,in=180] (2,0.8);
  \draw[ultra thick] (0,1.5) arc (90:270:1.5);
  \draw[ultra thick] (0,-1.5) to[out=0,in=180] (2,-1) to[out=0,in=270] (3.7,-0.2) to[out=90,in=0] (2,0.8);
  \draw[ultra thick] (-0.5,0) to[out=50,in=130] (0.5,0);
  \draw[ultra thick] (-0.7,0.2) to[out=-60,in=240] (0.7,0.2);
  \fill [pattern= north east lines] (1.6,-0.5) circle (0.15);
  \draw [ultra thick] (1.6,-0.5) circle (0.15);
  \fill [pattern= north east lines] (2.5,-0.5) circle (0.15);
  \draw [ultra thick] (2.5,-0.5) circle (0.15);
  \draw (0.85,1.65) node[left] {$y_1$};
  \draw (-0.2,-0.9) node[below] {$x_1$};
  \draw (1.6,-1.2) node[below] {$z_2$};
  \draw (2.7,-1.2) node[below] {$z_1$};
  \draw (1.6,0.35) node[above] {$p$};
\end{tikzpicture}
\caption{One choice of arcs corresponding to generators of $\pi_1(\Sigma_{1,2},p)$.}
\label{fig:ChoiceOfArcsG1B2}
\end{figure}
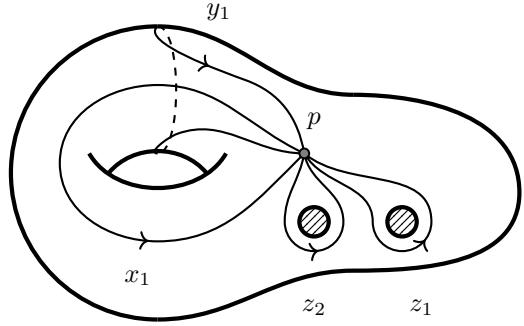

The \textit{character variety} of $\pi_1(\Sigma_{g,b})$ in a simple complex linear algebraic group $G$ is the categorical quotient of the space of group homomorphisms from $\pi_1(\Sigma_{g,b})$ to $G$ with respect to conjugations by $G$:\footnote{Note that despite its name, ${\mathcal M}_\Sigma^G$ is not a variety but rather an affine scheme. Throughout the text we will always think of $\mathcal M_\Sigma^G$ in terms of its coordinate ring, which is precisely the $G$-invariant subring of the coordinate ring of representation variety.}
\[
\cM^G_\Sigma=(\pi_1(\Sigma_{g,b})\to G)\git G\quad:=\quad\Spec\left(\mathcal O[\mathrm{Hom}(\pi_1(\Sigma_{g,b}),G)]^G\right)
\]
For $b>0$ we have an obvious isomorphism:
\[
\cM^G_\Sigma\simeq G^{\times (2g+b-1)}\git G
\]

We can choose this isomorphism as the mapping which assigns
to a group homomorphism $\rho: \pi_1(\Sigma_{g,b})\to G$ an element
$(\rho(x_1),\dots, \rho(x_g), \rho(y_1), \dots, \rho(y_g), \rho(z_1),\dots, \rho(z_{b-1}))\in G^{\times (2g+b-1)}$ and then projects it to the conjugacy class in $G^{\times (2g+b-1)}\git G$.

\subsection{Moduli space of surface graph connections}\label{MGF}
Denote by $V(\Gamma)$ and $E(\Gamma)$ the set of vertices and edges of an oriented graph $\Gamma$. Graph connection is an assignment of a parallel transport $g_e\in G$ to each oriented edge $e\in E(\Gamma)$.

The gauge group $G(V)\simeq G^{V(\Gamma)}=\{ V(\Gamma)\to G, \ \ v\mapsto h_v\}$
acts on a graph connection $g$ as
\[
g_e\mapsto h_{t(e)} g_e h_{s(e)}^{-1}.
\]
where $t(e)$ is the target vertex for an oriented edge $e$ and $s(e)$ is its source vertex.

Assume that  $\Gamma\subset\Sigma$ is an oriented embedded graph such that $\Sigma\backslash\Gamma$ is a disjoint collection of disks and $b$ annuli corresponding to the boundary components of $\Sigma=\Sigma_{g,b}$. We will call such graphs {\it simple}. Hereinafter we allow graphs to have multiple edges and loops.

Let $F(\Gamma)$ be the set of contractible faces, i.e. the set of disks in $\Sigma_{g,b}\backslash\Gamma$.
For a disc $D$ with boundary $\pa D\subset \Gamma$ define
\[
g_{\partial D}=g_{e_k}^{\epsilon_k}\dots g_{e_1}^{\epsilon_k}
\]
as the holonomy of graph connection along the boundary of a disk. This holonomy is the product of elements $g_{e_j}^{\epsilon_j}\in G$ associated to each edge $e_j\in\partial D$ in a cyclic order induced by the orientation of the surface\footnote{The total order is irrelevant when we pass to the moduli space of flat connections.}.
The relative orientation $\epsilon_j$ is $+1$ when the
orientations of $e_j$ and $\partial D$ coincide and $-1$ otherwise.

A graph connection over $\Gamma\subset \Sigma$ is called {\it flat} if $g_{\partial D}=1$ for all $D\in F(\Gamma)$. Define the space of {\it flat graph $G$-connections} $\cA_{(\Sigma, \Gamma)}$ on $\Gamma\subset \Sigma$ as the space of such flat connections.

When $G$ is a linear algebraic group, the space of flat graph connections is an algebraic subset of $G^{|E(\Gamma))|}$ equipped with a regular action of $G(V)$. The {\it moduli space of flat graph connections}  is the categorical quotient
\begin{align*}
\mathcal M_{(\Sigma,\Gamma)}^G=\mathcal A_{(\Sigma,\Gamma)}\git G(V)\quad:=\quad\Spec\left(\mathcal O[\mathcal A_{(\Sigma,\Gamma)}]^{G(V)}\right).
\end{align*}

Let $e\in E(\Gamma)$ be an edge which connects two distinct vertices, denote as $\Gamma_e$ an embedded graph obtained from $\Gamma$ by contracting an edge $e$. Similarly, for any contractible face $D\in F(\Gamma)$ which has an edge $e\in\partial D$ that appears only once in $\partial D$ we can define an ebedded graph $\Gamma^e=\Gamma\backslash e$. Recall that $\Gamma$ was simple,
i.e. its complement in $\Sigma$ is the union of disks and annuli. It is clear that $(\Sigma\backslash\Gamma^e)$ is also simple. The flatness condition $g_{\partial D}=1$ guarantees the isomorphism of schemes.

\begin{theorem} The following schemes are isomorphic:

\begin{equation}\label{moduli-isom}
\cM^G_{(\Sigma, \Gamma)}\simeq \cM^G_{(\Sigma, \Gamma_e)}\simeq \cM^G_{(\Sigma, \Gamma^e)}
\end{equation}

\end{theorem}

\begin{proof} Resolving one of the relations $g_{\partial D}=1$ we obtain an isomorphism of $G(V)$-modules $\mathcal O[\mathcal A_{(\Sigma,\Gamma)}]\simeq\mathcal O[\mathcal A_{(\Sigma,\Gamma^e)}]$. Taking the $G(V)$-invariant part we conclude
\begin{align*}
    \mathcal O[\mathcal A_{(\Sigma,\Gamma)}]^{G(V)}\simeq\mathcal O[\mathcal A_{(\Sigma,\Gamma^e)}]^{G(V)}.
\end{align*}
and therefore $\cM^G_{(\Sigma, \Gamma)}\simeq\cM^G_{(\Sigma, \Gamma^e)}$.

Now let us prove the first isomorphism in (\ref{moduli-isom}).
Let $e_0$ be an edge of $\Gamma$ with adjacent vertices $v=s(e_0)$ to $w=t(e_0)$. Here $s(e),t(e)$ stand for the source and target of $e$ respectively. Contracting of $e_0$ maps vertices $v$ and $w$ to one vertex $w_0$ of
the new graph $\Gamma_{e_0}$ which is the contraction of $\Gamma$. Denote this mapping by $\pi_{e_0}$.

The contraction of $\Gamma$ to $\Gamma_{e_0}$ defines a functor between the fundamental grouppoids
which we will denote by the same letter $\pi_{e_0}: \Pi(\Gamma)\to \Pi(\Gamma_{e_0})$. It is clear that
this is a projection. Choose a section $\phi: \Pi(\Gamma_{e_0})\to \Pi(\Gamma)$ of $\pi_{e_0}$, so that $\pi_{e_0}\phi=id$. It is
clear that $\phi$ does not have to be unique, but it always exists.
This is illustrated on Figure \ref{fig:ContractionEdgeGraphConnection}.
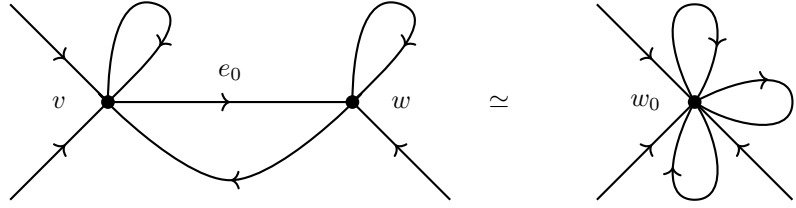
\begin{figure}
\begin{tikzpicture}[scale=1.3]
\draw[thick,-<-] (0,0) to (-1,1);
\draw[thick,-<-] (0,0) to (-1,-1);
\draw[thick,-->-] (0,0) to[out=90,in=160] (0.5,1) to[out=-20,in=45] (0,0);
\draw[thick,->-] (0,0) to (2.5,0);
\draw[thick,->-] (2.5,0) to[in=0] (1.25,-0.8) to[out=180] (0,0);
\draw[thick,-->-] (2.5,0) to[out=90,in=160] (3,1) to[out=-20,in=45] (2.5,0);
\draw[thick,-<-] (2.5,0) to (3.5,-1);
\fill (0,0) circle (0.07);
\fill (2.5,0) circle (0.07);
\draw (-0.5,0) node {$v$};
\draw (3,0) node {$w$};
\draw (1.25,0.3) node {$e_0$};
\draw (4,0) node {$\simeq$};
\draw[thick,-<-] (6,0) to (5,1);
\draw[thick,-<-] (6,0) to (5,-1);
\draw[thick,-<-] (6,0) to (7,-1);
\draw[thick,-->-] (6,0) to[out=-55,in=0] (6,-1) to[out=180,in=-120] (6,0);
\draw[thick,-<--] (6,0) to[out=55,in=0] (6,1) to[out=180,in=120] (6,0);
\draw[thick,--<-] (6,0) to[out=55-90,in=0-90] (7,0) to[out=180-90,in=120-90] (6,0);
\fill (6,0) circle (0.07);
\draw (5.5,0) node {$w_0$};
\end{tikzpicture}
\caption{Contraction of an edge}
\label{fig:ContractionEdgeGraphConnection}
\end{figure}

The functors $\pi_{e_0}$ and $\phi$ define the projection $p_{e_0}: \cA_{(\Sigma, \Gamma)}\to \cA_{(\Sigma, \Gamma_{e_0})}$ and its section $f: \cA_{(\Sigma, \Gamma_{e_0})}\to \cA_{(\Sigma, \Gamma)}$.
Taking quotients with respect to the gauge group we obtain the projection $[p_{e_0}]: \cM_{(\Sigma, \Gamma)}\to \cM_{(\Sigma, \Gamma_{e_0})}$ and its section $[f]: \cM_{(\Sigma, \Gamma_{e_0})}\to \cM_{(\Sigma, \Gamma)}$.
Comparing dimensions we conclude that these mappings are isomorphisms. The section $f$ depends on the choice of
$\phi$. It is easy to see that the corresponding mapping $[f]$ between moduli spaces does not.

\end{proof}

\begin{corollary} For any oriented embedded graph $\Gamma\in\Sigma_{g,b}$ such that $\Sigma_{g,b}\backslash\Gamma$ is a disjoint union of discs and $b$ annuli, the moduli space of graph connections is isomorphic to the $G$-character variety
    \begin{align*}
        \mathcal M_{\Sigma}^G\simeq\mathcal M_{(\Sigma,\Gamma)}^G.
    \end{align*}
\end{corollary}

\begin{corollary}
The dimensions of moduli spaces of graph connections on surface $\Sigma_{g,b}$ with $b>0$ are:
\begin{align}
\dim\left(\cM^G_{(\Sigma, \Gamma)}\right)=\left\{\begin{array}{cc}
0&g=0,\;b=1,\\
rk(G)&g=0,\;b=2,\\
(2g+b-2)\dim G&\textrm{otherwise}.
\end{array}\right.
\label{eq:dimrepscheme}
\end{align}
\end{corollary}

\subsection{Colored Immersed Graphs}
Let $\Gamma$ be an oriented graph.
\begin{definition}
An edge coloring of $\Gamma$ is an assignment of a finite-dimensional $G$-representation $V_\epsilon$ to each edge $e\in\Gamma$ of the oriented graph. We will denote edge colored graph as $(\Gamma, V)$.
\end{definition}

We say that an oriented graph $\Gamma$ is \textit{immersed} in a surface $\Sigma_{g,b}$ if the inclusion mapping $I:\Gamma\rightarrow\Sigma_{g,b}$ is locally an embedding\footnote{Immersed graphs may intersect on a surface and intersections does not have to be transversal.}. If the intersection of $\Gamma$ on $\Sigma_{g,b}$  are transversal, we will call it {\it transversally immersed}.

Orientation of the surface induces a cyclic order on edges adjacent to a given vertex. For each vertex $v\in V(\Gamma)\subset \Sigma$ let us choose a total ordering $e_1,\dots, e_n$ of the adjacent edges which agrees with the cyclic ordering induced by the orientation of $\Sigma$. We will refer to this as {\it edge ordering} of $\Gamma\subset \Sigma$. For a finite dimensional vector space $V$
define $V^+=V$ and $V^-=V^*$. For each vertex of edge ordered graph define the space
\begin{equation}\label{V-space}
V(v)=V_{e_1}^{\epsilon_1}\otimes \dots \otimes V_{e_n}^{\epsilon_n}
\end{equation}
Here $\epsilon_i=+1$ is oriented outward from $v$, i.e. $v=s(e_i)$ and $\epsilon=-1$ otherwise,
i.e. when $v=t(e)$\footnote{Note that because the category of $G$-modules is symmetric, there exists a canonical isomorphism between $V(v)$ for different choices of edge ordering. Still, it will be convenient for us to assume that we have chosen a total order of edges adjoint to the given vertex.}.

\begin{definition}
A vertex coloring of $(\Gamma,V)$ is an assignment of a vector $c_v\in V(v)^G$ to each vertex $v$ of an edge ordered oriented graph $\Gamma$.
\label{def:VertexColoring}
\end{definition}
Here $V(v)^G$ is the space of $G$-invariant vectors in a representation $V(v)$.

We define \textit{colored graph} $\mathbf\Gamma=(\Gamma,V,c)$ as an edge ordered oriented immersed graph $\Gamma$ equipped with an edge coloring $V$ and a vertex coloring $c$.

\subsection{Invariant Functions on Connections associated to Immersed Graphs}

\subsubsection{} For each colored graph $(\Gamma,V,c)$ embedded into $\Sigma$ one can associate a function $F_{\Gamma,V,c}$ on  connections on the principal $G$-bindle over $\Sigma_{g,b}$ [reference] as follows.

For a graph $\Gamma$ with an edge coloring $V$ define the following $G$-module
\begin{align}
{\mathbf V}(\Gamma)=\bigotimes_{v\in V(\Gamma)} V(v)
\label{eq:VSpaceDef}
\end{align}
Here we assume a choice of a linear ordering on vertices. The space $V(v)$ are defined in (\ref{V-space}).
Note that $V(v)$ can also be written as
\[
V(v)=\bigotimes_{e\in S(v)} V(e)^{\epsilon(e,v)}
\]
Here $S(v)$ is the star of vertex $v$\footnote{In our case this is the set of edges adjacent to $v$.} and $\epsilon(e,v)=+$ if $e$ is outgoing and $\epsilon(e,v)=-$
if it is incoming.

By changing the order of factors in tensor product we obtain a natural isomorphism of $G$-modules
\begin{equation}\label{reenum-V}
{\mathbf V}(\Gamma)\simeq \bigotimes_{e\in E(\Gamma)} (V_e\otimes V_e^*)
\end{equation}

The space ${\mathbf V}(\Gamma)$ comes equipped with a symmetric bilinear form $\langle,\rangle$ defined on pure tensors $\mathbf u,\mathbf w\in\mathbf V$ by the following rule
\begin{equation}
\langle \mathbf u,\mathbf w\rangle=\prod_{e\in E(\Gamma)}u_{e,t(e)}(w_{e,s(e)})\,w_{e,t(e)}(u_{e,s(e)}),\label{eq:BilinearFormOnV}
\end{equation}
where
\[
\mathbf u=\bigotimes_{v\in V(\Gamma)} \left(\bigotimes_{e\in S(v)} u_{e,v} \right),
\]
\[
\mathbf w=\bigotimes_{v\in V(\Gamma)} \left(\bigotimes_{e\in S(v)} w_{e,v)}\right).
\]
Here $u_{e,s(e)}, w_{e,s(e)}\in V_e$ and $u_{e,t(e)}, w_{e,t(e)}\in V_e^*$

\subsubsection{} It is clear that a principal $G$-bundle $E_\Sigma$  over $\Sigma$ defines a principal $G$-bundle $E_\Gamma$
over vertices of $\Gamma\subset \Sigma$. A connection $A$ on $E_\Sigma$ defines a graph connection
on $E_\Gamma$ with the parallel transport along the edge $e$ being the parallel transport along $e\subset \Sigma$ with respect to connection $A$. The evaluation of the parallel transport along $e$ in the representation $V_e$
we will denote by $P^A_e: V_e\to V_e$.

Define the vector $\mathbf \pi^A\in\mathbf V$ in the $G$-module (\ref{eq:VSpaceDef})
as
\begin{align}
\pi^A=\bigotimes_{e\in E(\Gamma)} P^A_e
\label{eq:ParallelTransortVector}
\end{align}
Here we assume a choice of total ordering on the set of edges and the identification with the
tensor product with the reordered tensor product as in (\ref{reenum-V}). The holonomy map $P_e$ we regard as a vector in $V_e\otimes V_e^*$.

From the coloring of vertices we  we get another vector $\mathbf c^\Gamma\in {\mathbf V}(\Gamma)$
\begin{align*}
\mathbf c^\Gamma=\bigotimes_{v\in V(\Gamma)}c^\Gamma_v.
\end{align*}
Here $c_v^\Gamma$ is the coloring of the vertex $v$ i.e. a vector in $V(v)^G\subset V(v)$.

\begin{definition} Define graph functions using the symmetric bilinear form (\ref{eq:BilinearFormOnV}) as
\begin{align}
F_{\Gamma,V,c}(A)=\langle\pi^A,\mathbf c^\Gamma\rangle.
\label{eq:FDefinition}
\end{align}
\end{definition}

\begin{proposition}
Graph functions have the following properties:
\[
F_{\Gamma,V,c}(A^g)=F_{\Gamma,V,c}(A)
\]
where $A^g$ is a flat connection $A$ after gauge transformation $g$. If the
connection is flat, the function
$F_{\Gamma,V,c}(A)$ depends only on the isotopy class $[\Gamma]$ of $\Gamma$.
\end{proposition}

\begin{proof}
By Definition \ref{def:VertexColoring}, $c_p^\Gamma\in V(p)^G$ is an invariant vector for each vertex $p$ of the graph. Hence, each function $F_{\Gamma,V,c}$ is invariant under gauge transformations of $A$. Moreover, when $A$ is flat connection, $F_{\Gamma_1,V,c}(A)=F_{\Gamma_2,V,c}(A)$ whenever an immersed graph $\Gamma_1$ can be deformed to $\Gamma_2$ by the regular homotopy (i.e. homotopy through graph immersions).

As a result, for each triple $([\Gamma],V,c)$ where $[\Gamma]$ is the  isotopy class of an immersed oriented ordered graph $\Gamma$,  $V$ is an edge coloring and $c$ is a vertex coloring  we have a function $F_{[\mathbf\Gamma], V,c}=F_{\Gamma,V,c}$ on the moduli space of flat connections on $\Sigma_{g,b}$. As a consequence, $F$ is a function of the moduli space of flat connections.

\end{proof}

\begin{theorem}\cite{AndersenMattesReshetikhin'1996}
Functions $F_{[\Gamma],V,c}$ span the coordinate ring (\ref{moduli-isom}) of the moduli space of $G$-represen\-tations of the fundamental group $\pi_1(\Sigma_{g,b},p)$ .
\label{prop:InvariantSpanningSet}
\end{theorem}

\begin{proof}
By an algebraic analogue of Peter-Weyl theorem, the coordinate ring $R=\mathcal O[G^N]$ of $G$-representations of the fundamental group can be decomposed as a $G^N\times\left(G^N\right)^{op}$-module
    \begin{align}
        \mathcal O[G^N]
        \simeq\bigoplus_{\lambda_1,\dots,\lambda_N\in\Lambda}\left(\bigotimes_{j=1}^N\quad W_{\lambda_j}\otimes W_{\lambda_j}^*\right),
        \label{eq:AlgebraicPeter-Weyl}
    \end{align}
    where $W_\lambda$ stands for a finite-dimensional irreducible representation with highest weight $\lambda$.

    In particular, (\ref{eq:AlgebraicPeter-Weyl}) is an isomorphism of $G$-modules, where $G$ acts by conjugation. Now let $Q\in\mathcal O[G^N]^G$ be a $G$-invariant polynomial, by (\ref{eq:AlgebraicPeter-Weyl}) it can be decomposed as
    \begin{align*}
        Q=\sum_{j=1}^mQ_j,\qquad Q_j\in \left(W_{\lambda_{j,1}}\otimes W^*_{\lambda_{j,1}}\otimes\dots\otimes W_{\lambda_{j,N}}\otimes W^*_{\lambda_{j,N}}\right)^G
    \end{align*}
    where each $Q_j$ is a graph function associated to the ribbon graph with a single vertex and $N=2g+b-1$ morphisms corresponding to free generators of $\pi_1(\Sigma_{g,b},p)$. The edge coloring of edges corresponding to generators of the fundamental group is given by $\lambda_{j,1},\dots,\lambda_{j,N}$, while $Q_j$ is an invariant vector which defines the coloring of the single vertex $p$.
\end{proof}

\subsection{Equivalence of Graph Functions}
In this subsection we list elementary operations on colored immersed graphs which leave the associated graph function (\ref{eq:FDefinition}) invariant. These operations include collapsing a pair of neighbouring vertices (Figure \ref{fig:Collapsing}) and thinning a pair of neighbouring edges (Figure \ref{fig:Thinning}). In each case shown on Figure \ref{fig:EquivalenceImmersedGraphs} we refer to the left graph as $\Gamma_1$ and to the right graph as $\Gamma_2$.

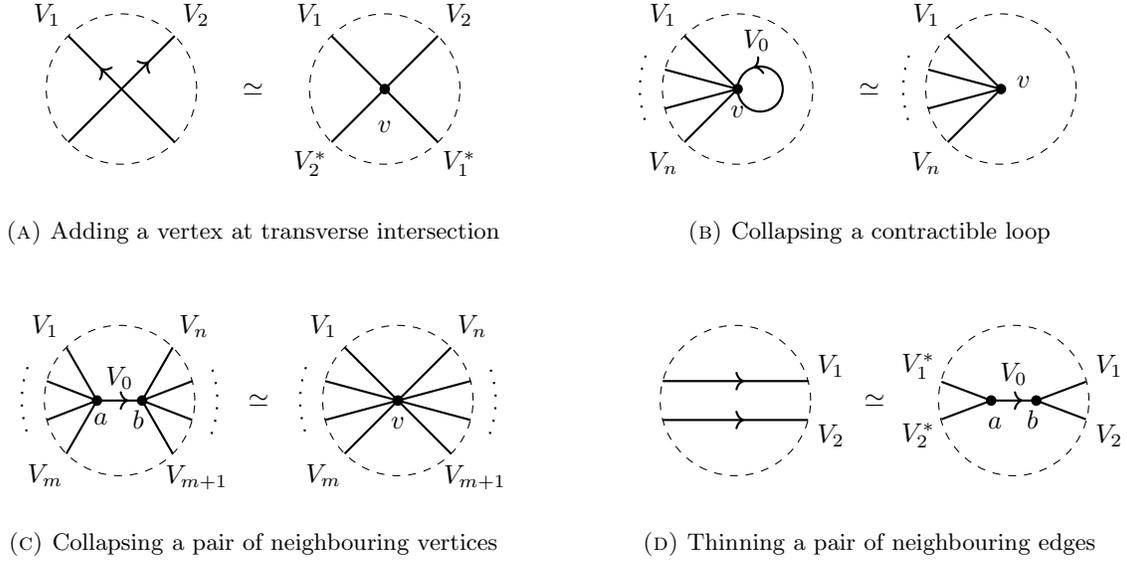
\begin{figure}
\begin{subfigure}[b]{0.45\linewidth}
\centering
\begin{tikzpicture}[scale=1]
\draw[white] (-1.7,-1.5) rectangle (5.2,1.5);
\draw (0,0) [dashed] circle (1);
\draw[thick,-<--] (45:1) to (225:1);
\draw[thick,-->-] (-45:1) to (-225:1);
\draw (135:1.4) node {$V_1$};
\draw (45:1.4) node {$V_2$};
\draw (1.75,0) node {$\simeq$};
\draw (3.5,0) [dashed] circle (1);
\draw[thick] (3.5,0) to +(45:1);
\draw[thick] (3.5,0) to +(-45:1);
\draw[thick] (3.5,0) to +(225:1);
\draw[thick] (3.5,0) to +(-225:1);
\fill (3.5,0) circle (0.07);
\draw (3.5,0)+(135:1.4) node {$V_1$};
\draw (3.5,0)+(45:1.4) node {$V_2$};
\draw (3.5,0)+(-135:1.4) node {$V_2^*$};
\draw (3.5,0)+(-45:1.4) node {$V_1^*$};
\draw (3.5,-0.5) node {$v$};
\end{tikzpicture}
\caption{Adding a vertex at transverse intersection}
\label{fig:AddingVertexTransverse}
\vspace*{0.5cm}
\end{subfigure}
\qquad
\begin{subfigure}[b]{0.45\linewidth}
\centering
\begin{tikzpicture}[scale=1]
\draw[white] (-1.7,-1.5) rectangle (5.2,1.5);
\draw (0,0) [dashed] circle (1);
\draw[thick] (0,0) to +(135:1);
\draw[thick] (0,0) to +(165:1);
\draw[thick] (0,0) to +(195:1);
\draw[thick] (0,0) to +(225:1);
\draw[thick,->--] (0.3,0) circle (0.3);
\fill (0,0) circle (0.07);
\draw (0.25,0.65) node {$V_0$};
\draw (0,-0.3) node {$v$};
\draw (135:1.4) node {$V_1$};
\draw[thick, loosely dotted] (160:1.3) arc (160:200:1.3);
\draw (225:1.4) node {$V_n$};
\draw (1.75,0) node {$\simeq$};
\draw (3.5,0) [dashed] circle (1);
\draw[thick] (3.5,0) to +(135:1);
\draw[thick] (3.5,0) to +(165:1);
\draw[thick] (3.5,0) to +(195:1);
\draw[thick] (3.5,0) to +(225:1);
\fill (3.5,0) circle (0.07);
\draw (3.5,0)+(135:1.4) node {$V_1$};
\draw[thick, loosely dotted] (3.5,0)+(160:1.3) arc (160:200:1.3);
\draw (3.5,0)+(225:1.4) node {$V_n$};
\draw (3.8,-0.1) node[above] {$v$};
\end{tikzpicture}
\caption{Collapsing a contractible loop}
\label{fig:CollapsingLoop}
\vspace*{0.5cm}
\end{subfigure}
\qquad
\begin{subfigure}{0.45\linewidth}
\centering
\begin{tikzpicture}
\draw[white] (-1.7,-1.5) rectangle (5.2,1.5);
\draw (0,0) [dashed] circle (1);
\draw[thick,-->-] (-0.3,0) to (0.3,0);
\draw[thick] (-0.3,0) to (135:1);
\draw[thick] (-0.3,0) to (165:1);
\draw[thick] (-0.3,0) to (195:1);
\draw[thick] (-0.3,0) to (225:1);
\draw[thick] (0.3,0) to (-45:1);
\draw[thick] (0.3,0) to (-15:1);
\draw[thick] (0.3,0) to (15:1);
\draw[thick] (0.3,0) to (45:1);
\fill (-0.3,0) circle (0.07);
\fill (0.3,0) circle (0.07);
\draw (135:1.4) node {$V_1$};
\draw[thick, loosely dotted] (160:1.3) arc (160:200:1.3);
\draw (225:1.4) node {$V_m$};
\draw (-45:1.45) node {$V_{m+1}$};
\draw[thick, loosely dotted] (-20:1.3) arc (-20:20:1.3);
\draw (45:1.4) node {$V_n$};
\draw (0,0.3) node {$V_0$};
\draw (-0.25,-0.45) node[above] {$a$};
\draw (0.25,-0.5) node[above] {$b$};
\draw (1.85,0) node {$\simeq$};
\draw (3.7,0) [dashed] circle (1);
\draw[thick] (3.7,0) to +(135:1);
\draw[thick] (3.7,0) to +(165:1);
\draw[thick] (3.7,0) to +(195:1);
\draw[thick] (3.7,0) to +(225:1);
\draw[thick] (3.7,0) to +(-45:1);
\draw[thick] (3.7,0) to +(-15:1);
\draw[thick] (3.7,0) to +(15:1);
\draw[thick] (3.7,0) to +(45:1);
\fill (3.7,0) circle (0.07);
\draw (3.7,0)+(135:1.4) node {$V_1$};
\draw[thick, loosely dotted] (3.7,0)+(160:1.3) arc (160:200:1.3);
\draw (3.7,0)+(225:1.4) node {$V_m$};
\draw (3.7,0)+(-45:1.45) node {$V_{m+1}$};
\draw[thick, loosely dotted] (3.7,0)+(-20:1.3) arc (-20:20:1.3);
\draw (3.7,0)+(45:1.4) node {$V_n$};
\draw (3.7,-0.5) node[above] {$v$};
\end{tikzpicture}
\caption{Collapsing a pair of neighbouring vertices}
\label{fig:Collapsing}
\end{subfigure}
\qquad
\begin{subfigure}{0.45\linewidth}
\centering
\begin{tikzpicture}
\draw[white] (-1.7,-1.5) rectangle (5.2,1.5);
\draw (0,0) [dashed] circle (1);
\draw[thick,-<-] (15:1) to (165:1);
\draw[thick,-<-] (-15:1) to (195:1);
\draw (20:1.35) node {$V_1$};
\draw (-20:1.35) node {$V_2$};
\draw (1.85,0) node {$\simeq$};
\draw (3.7,0) [dashed] circle (1);
\draw[thick,-->-] (3.7-0.3,0) to (3.7+0.3,0);
\draw[thick] (3.7-0.3,0) to (3.7-0.96,0.25);
\draw[thick] (3.7-0.3,0) to (3.7-0.96,-0.25);
\draw[thick] (3.7+0.3,0) to (3.7+0.96,0.25);
\draw[thick] (3.7+0.3,0) to (3.7+0.96,-0.25);
\fill (3.7-0.3,0) circle (0.07);
\fill (3.7+0.3,0) circle (0.07);
\draw (3.7,0.35) node {$V_0$};
\draw (3.7,0)+(20:1.35) node {$V_1$};
\draw (3.7,0)+(-20:1.35) node {$V_2$};
\draw (3.7,0)+(160:1.35) node {$V_1^*$};
\draw (3.7,0)+(200:1.35) node {$V_2^*$};
\draw (3.7-0.25,-0.5) node[above] {$a$};
\draw (3.7+0.25,-0.5) node[above] {$b$};
\end{tikzpicture}
\caption{Thinning a pair of neighbouring edges}
\label{fig:Thinning}
\end{subfigure}
\caption{Graphs defining equivalent functions (Edges are oriented outwards unless specified)}
\label{fig:EquivalenceImmersedGraphs}
\end{figure}
\begin{lemma}\cite{FockRosly'1999},\cite{KashaevReshetikhinWebsterYakimov}
Let $\Gamma_1,\Gamma_2$ be a pair of colored immersed oriented graphs which coincide away from a disc shown on any of Figures \ref{fig:AddingVertexTransverse}-\ref{fig:Thinning}. Then there exists a vertex coloring of $\Gamma_2$, such that the two colored immersed graph give rise to the same invariant function on graph connections:
\begin{align*}
    F_{[\mathbf{\Gamma_1}]}=F_{[\mathbf{\Gamma_2}]}.
\end{align*}
\label{lemm:GraphEquivalence}
\end{lemma}
\begin{proof}
Consider a flat connection $A$ on a principle $G$-bundle over $\Sigma_{g,b}$ and let
\begin{align}
    \pi^A_{\Gamma_1}=\bigotimes_{e\in E(\Gamma_1)}P_e\in \mathbf V(\Gamma_1),\qquad \pi^A_{\Gamma_2}=\bigotimes_{e\in E(\Gamma_2)}Q_e\in\mathbf V(\Gamma_2).
    \label{eq:HolonomyVectors}
\end{align}
be a pair of vectors in and  as in (\ref{eq:ParallelTransortVector}) defined by the parallel transport along edges of the two immersed graphs. In each of the four parts we are going to prove that
    \begin{align}
        F_{[\mathbf\Gamma_1]}(A)=\langle \pi_{\Gamma_1}^A,\mathbf c^{\Gamma_1}\rangle=\langle \pi_{\Gamma_1}^A,\mathbf c^{\Gamma_2}\rangle=F_{[\mathbf\Gamma_2]}(A)
    \label{eq:EquivalenceOfGraphConnectionsPairing}
    \end{align}
for an appropriate choice of vertex coloring.

\begin{enumerate}[(A)]
    \item Let
    \begin{align*}
        \mathrm{Id}_{V_1}=\sum_i v_i\otimes v_i^*,\qquad \mathrm{Id}_{V_2}=\sum_j w_i\otimes w_j^*
    \end{align*}
    where $\{v_i\},\{v_i^*\}$ and $\{w_i\},\{w_i^*\}$ be any two pairs of dual bases on $V_1, V_1^*$ and $V_2,V_2^*$ respectively, i.e. $\langle v_i,v_j^*\rangle=\delta_{i,j}$ and $\langle w_i,w_j^*\rangle=\delta_{i,j}$. Assign a vertex coloring of $\Gamma_2$ at $v$ as
    \begin{align*}
        c^{\Gamma_2}_v=\sum_{i,j}v_i\otimes w_j^*\otimes v_i^*\otimes w_j\quad\simeq\quad\mathrm{Id}_{V_1}\otimes\mathrm{Id}_{V_2}
    \end{align*}

    \item Let $\Gamma_1$ and $\Gamma_2$ be a pair of colored oriented immersed graphs which coincide away from a disc shown on Figure \ref{fig:CollapsingLoop}. Suppose that the vertex coloring of $\Gamma_1$ at $v$ is given by
    \begin{align*}
        c^{\Gamma_1}_v=\sum_j v^{0}_j\otimes w^{(0)}_j\otimes v^{(1)}_j\otimes\dots\otimes v^{(1)}_j.
    \end{align*}
    Define $c^{\Gamma_2}_v$ as follows
    \begin{align*}
        c^{\Gamma_1}_v=\sum_j \langle v^{0}_j, w^{(0)}_j\rangle\; v^{(1)}_j\otimes\dots\otimes v^{(1)}_j
    \end{align*}
    Because the holonomy of flat connection along contractible loop $e_0$ is trivial, we get the following contribution to holonomy vector (\ref{eq:HolonomyVectors})
    \begin{align}
        P_{e_0}=\mathrm{Id}_{V_0}=\sum_{i=1}^M{u_i\otimes u_i^*}
        \label{eq:HolonomyTrivial}
    \end{align}
    where $\{u_i\}$ and $\{u_i^*\}$ is an arbitrary pair of dual bases in $V_0$ and $V_0^*$, i.e. $\langle u_i,u_l^*\rangle=\delta_{i,l}$. Then (\ref{eq:EquivalenceOfGraphConnectionsPairing}) follows immediately by
    \begin{align}
        \sum_{j=1}^M\langle w_k^{(0)},u_j\rangle\langle u_j^*,v_j^{(0)}\rangle=\langle w_k^{(0)},v_j^{(0)}\rangle.
        \label{eq:HolonomyTrivialSum}
    \end{align}

    \item Assume further that the vertex coloring of vertices $a$ and $b$ is given by a pair of $G$-invariant vectors
    \begin{align*}
        c^{\Gamma_1}_a=&\sum_jv_j^{(0)}\otimes v_j^{(1)}\otimes\dots\otimes v_j^{(m)}\quad\in\quad (V_0\otimes\dots\otimes V_m)^G,\\[5pt]
        c^{\Gamma_1}_b=&\sum_kw_k^{(0)}\otimes w_k^{(m+1)}\otimes\dots\otimes w_k^{(n)}\quad\in\quad (V_0^*\otimes V_{m+1}\otimes\dots\otimes V_n)^G.
    \end{align*}
    Let $c_v$ be a $G$-invariant vector defined as follows:
    \begin{align*}
        c_v^{\Gamma_2}=\sum_{j,k} \left\langle w_k^{(0)},v_j^{(0)}\right\rangle\; v_j^{(1)}\otimes\dots\otimes v_j^{(m)}\otimes w_k^{(m+1)}\otimes\dots\otimes w_k^{(n)}\quad\in\quad (V_1\otimes\dots\otimes V_n)^G
    \end{align*}
    Note that because both $F_{[\mathbf\Gamma_1]}$ and $F_{[\mathbf\Gamma_2]}$ are functions on gauge equivalence classes of flat connections, without loss of generality we can assume that holonomy of $A$ in a contractible disc on Figure \ref{fig:Collapsing} is trivial. Using the same reasoning as in (\ref{eq:HolonomyTrivial}) and (\ref{eq:HolonomyTrivialSum}) we conclude that (\ref{eq:EquivalenceOfGraphConnectionsPairing}) holds.
    \item For the remaining part, again without loss of generality we can assume that connection is trivial inside contractible disk shown on Figure \ref{fig:Thinning}. For $\varphi:V_0\xrightarrow{\sim} V_1\otimes V_2$, assign $G$-invariant vectors at both vertices induced by the isomorphism $\varphi$ and its inverse, the statement (\ref{eq:EquivalenceOfGraphConnectionsPairing}) then follows by tautology.
\end{enumerate}
\end{proof}

\begin{corollary}\label{cor-isom}
   Let $\Gamma\in\Sigma$ be an oriented immersed graph with a single vertex $p$ and edges corresponding to free generators of $\pi_1(\Sigma,p)$. For any oriented immersed graph $(\Gamma_1,W,d)$ we can present the associated graph function as
    \begin{align*}
        F_{[\Gamma_1,W,d]}= F_{[\Gamma, V, c]}
    \end{align*}
    for an appropriate choice of coloring of simple oriented surface graph $\Gamma\subset \Sigma$.
\end{corollary}
\begin{proof}
    Because $F_{[\Gamma_1,W,d]}$ depends only on homotopy (through immersed graphs) equivalence type of $\Gamma_1$, w.l.o.g. we can assume that $\Gamma\cup\Gamma_1$ has only double transverse intersection points. So by Lemma \ref{lemm:GraphEquivalence} we can turn $\Gamma\cup\Gamma_1$ first into an embedded graph, then collapse all vertices of $\Gamma_1$ into vertices of $\Gamma_2$ and, finally, by a sequence of moves (\ref{fig:CollapsingLoop}), (\ref{fig:Collapsing}), (\ref{fig:Thinning}) eliminate all vertices of $\Gamma_1$.
\end{proof}

Replace holonomies along edges by groups elements $g_e$ given by a flat
graph connection $g: E(\Gamma)\to G$ we identify the graph function $F_{[\Gamma, V, c]}$
with a function on the space flat graph connections for $\Gamma\subset \Sigma$. Taking into account isomorphisms (\ref{moduli-isom}) we can identify invariant graph functions with functions
on the representation variety $\mathcal M_{\Sigma_{g,b}}^G$.
Thus the corollary \ref{cor-isom} establishes an isomorphism between the coordinate ring $\mathcal O[\mathcal M_{(\Sigma_{g,b},\Gamma)}]^G$ of moduli space of graph connections and the coordinte ring $\mathcal O[\mathcal M_{\Sigma_{g,b}^G}]^G$.

\subsection{Poisson structure on the moduli space of flat connections} The coordinate ring $\mathcal O[\mathcal M_\Sigma^G]$ of the $G$-charater variety of $\pi_1(\Sigma_{g,b})$ comes equipped with a Poisson bracket equivariant under the action of the Mapping Class Group $\mathop{Mod}(\Sigma_{g,b})$ \cite{AtiyahBott'1983,Goldman'1986,BiswasGuruprasad'1993}. Following \cite{FockRosly'1993,AndersenMattesReshetikhin'1996,FockRosly'1999} we will define this bracket in terms of graph functions (\ref{eq:FDefinition}).

Let $([\Gamma],V,c)$ and $([\Gamma'],V',c')$ be a pair of isotopy equivalence classes of colored immersed graphs. Pick representatives $\Gamma$ and $\Gamma'$ which have only transverse double intersection points at edges.

\begin{definition}
Let $p\in\Gamma\cap\Gamma'$ be a transverse double intersection point of two colored immersed graphs $\mathbf\Gamma=(\Gamma,V,c)$ and $\mathbf\Gamma'=(\Gamma',V',c')$. Define (isotopy equivalence class of) colored immersed graph $[\mathbf\Gamma\star_p\mathbf\Gamma']=([\Gamma''],V'',c'')$ as follows:

Consider a small neighborhood of $p$ and denote by $e\in\Gamma$ and $e'\in\Gamma'$ the corresponding edges of graphs containing the intersection point $p\in e\cap e'$.
\begin{figure}
\centering
\begin{subfigure}{0.3\linewidth}
\centering
\begin{tikzpicture}[scale=1.2]
\draw [thick] (-0.707,-0.707) to (0,0);
\draw [thick,-->-] (0,0) to (0.707,0.707);
\draw [thick] (0.707,-0.707) to (0,0);
\draw [thick,-->-] (0,0) to (-0.707,0.707);
\draw (0,0) [dashed] circle (1);
\draw (1,1.1) node {$e\in\Gamma$};
\draw (-1,1.1) node {$e'\in\Gamma'$};
\end{tikzpicture}
\caption{Neighborhood of a simple intersection point $p\in\Gamma\cap\Gamma'$}
\end{subfigure}
\qquad
\begin{subfigure}{0.3\linewidth}
\centering
\begin{tikzpicture}[scale=1.2]
\draw [thick,->--] (-0.707,-0.707) to (0,0);
\draw [thick,-->-] (0,0) to (0.707,0.707);
\draw [thick,->--] (0.707,-0.707) to (0,0);
\draw [thick,-->-] (0,0) to (-0.707,0.707);
\draw [thick,->-] (-0.707/2,-0.707/2) to (0.707/2,-0.707/2);
\fill (-0.707/2,-0.707/2) circle (0.07);
\fill (0.707/2,-0.707/2) circle (0.07);
\draw (0,0) [dashed] circle (1);
\draw (-1,-1) node {$e_1$};
\draw (1,1) node {$e_2$};
\draw (1,-1) node {$e_1'$};
\draw (-1,1) node {$e_2'$};
\draw (0,-0.7) node {$e''$};
\draw (-0.5,-0.1) node {$a$};
\draw (0.5,-0.1) node {$b$};
\end{tikzpicture}
\caption{The resulting immersed graph $\mathbf\Gamma\star_p\mathbf\Gamma'$}
\label{fig:BracketImmersedGraph}
\end{subfigure}
\caption{Double intersection point at edges and associated graph $(\Gamma,\Gamma')_p$}
\label{fig:DoubleIntersectionPoint}
\end{figure}
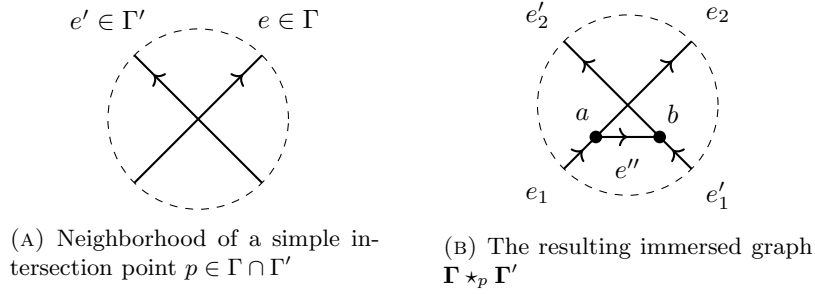
Take union of immersed graphs $\Gamma\cup\Gamma'$ and add two extra vertices $a,b$ which divide edges $e,e'$ into $e_1,e_2$ and $e_1',e_2'$ respectively as shown on Figure \ref{fig:BracketImmersedGraph}. Then add an edge $e''$ connecting $a$ and $b$ colored by adjoint representation $V_{e''}''=\mathfrak g$ of $G$.\footnote{Note that the adjoint representation is self-dual, so the orientation of $e''$ on Fig. \ref{fig:BracketImmersedGraph} can be chosen arbitrarily.} Define vertex coloring
\begin{equation}
\begin{aligned}
c''_a=&(\mathrm{ad}_{V_e}\otimes Id_{\mathfrak g})(I)\quad\in\quad V_e\otimes V_e^*\otimes \mathfrak g,\\
c''_b=&(\mathrm{ad}_{V_{e'}}\otimes Id_{\mathfrak g})(I)\quad\in\quad V_{e'}\otimes V_{e'}^*\otimes\mathfrak g,
\end{aligned}
\label{eq:NewVertexColoring}
\end{equation}
where $I=\sum_i t^i\otimes t_i\;\in\;\mathfrak g\otimes\mathfrak g$ is a quadratic Casimir. Finally, assume that edge and vertex coloring of the remaining graph is inherited from $\mathbf\Gamma$ and $\mathbf\Gamma'$.
\end{definition}

\begin{theorem}[\cite{FockRosly'1993,AndersenMattesReshetikhin'1996}]
The Poisson Bracket between $F_{[\mathbf\Gamma]}$ and $F_{[\mathbf\Gamma']}$ in terms of invariant graph functions reads
\begin{align}
\left\{F_{[\mathbf\Gamma]},F_{[\mathbf\Gamma']}\right\}=\sum_{p\in\Gamma\cap\Gamma'}\epsilon(p)F_{[\mathbf\Gamma\star_p\mathbf\Gamma']},
\label{eq:PoissonBracketGraphFunctions}
\end{align}
where the sum is taken over the intersection points $p\in\Gamma\cap\Gamma'$ and $\epsilon(p)=\pm1$ stands for the sing of an intersection point.
\end{theorem}

\section{Main Construction}
\label{sec:MainConstruction}

The following theorem describes superintegrable systems on moduli spaces flat connections.
Poisson commuting Hamiltonians of these systems are invariant functions of holonomies along
a systems of nonintersecting simple curves on the surface.

Denote by $Z_{\pa \Sigma}$ the Poisson center of the algebra of functions on the moduli space
of flat connection $A=\cO(\cM_\Sigma^G)$. The subalgebra $Z_{\pa \Sigma}$ is spanned by
$G$-invariant functions on holonomies around boundary components.

Let $C=C_1\sqcup\dots\sqcup C_k$ in $\Sigma$ be a disjoint union of pairwise nonhomotopic simple closed curves
none of which is homopotic to a boundary component. Define  $B_C\subset A$ as the subalgebra generated by graph functions $F_{\Gamma,V,c}$ with $\Gamma\subset C$. Similarly, define $J_{\Sigma\backslash C}$ as the subalgebra generated by graph functions
$F_{\Gamma,V,c}$ with $\Gamma\subset (\Sigma\backslash C)$.

\begin{theorem}
The system of subalgebras below defines an affine superintegrable system
\begin{align*}
    Z_{\partial\Sigma}\subset B_C\subset J_{\Sigma\backslash C}\subset A=\mathcal O[\mathcal M^G_\Sigma].
\end{align*}
\label{theorem:SuperIntSystemsCycles}
\end{theorem}

In the rest of this section we will prove this theorem. In order to prove it we should prove the balance of dimensions.

\subsection{Trivial superintegrable system on a surface}

Let $\Sigma$ be an oriented surface of genus $g$ with $b>0$ boundary components
and let $\mathcal M_\Sigma^G$ be the moduli space of $G$-representations of the fundamental group $\pi_1(\Sigma)$. Assume that we have at least one boundary component. Denote by $Z\subset A=\mathcal O[\mathcal M_\Sigma^G]$ a subalgebra generated by  graph functions  $F_{[\mathbf\Gamma]}$, where $\mathbf\Gamma$ is a colored graph with a single vertex and a single edge homotopic to one of the boundary components of $\Sigma$.

\begin{lemma} The subalgebra  $Z\subset \mathcal O[\mathcal M_\Sigma^G]$ belongs to the Poisson center of $\mathcal O[\mathcal M_\Sigma^G]$.
\label{lemm:BoundaryPoissonCenter}
\end{lemma}

\begin{proof}
Let $[\mathbf\Gamma_1],[\mathbf\Gamma_2]$ be a pair of isotopy equivalence classes of colored immersed graphs such that $[\Gamma_1]$ is homotopic to one of the boundary components. Then there exist members $\mathbf\Gamma_1$ and $\mathbf\Gamma_2$ with a trivial intersection. From (\ref{eq:PoissonBracketGraphFunctions}) we get
\begin{align}
\{F_{[\mathbf\Gamma_1]},F_{[\mathbf\Gamma_2]}\}=0.
\label{eq:BoundaryCommutes}
\end{align}
By Theorem \ref{prop:InvariantSpanningSet}, invariant graph functions of the form $F_{[\mathbf{\Gamma_2}]}$ span the entire coordinate ring, hence we conclude that $F_{[\mathbf{\Gamma_1}]}$ belongs to the Poisson center of $\mathcal O[\mathcal M_\Sigma^G]$.
\end{proof}

Note that $Z_0$ is finitely generated, indeed it is generated by $b$ copies of $\mathcal O[G]^G$, one for each boundary component. Every $\mathcal O[G]^G$ is in turn finitely generated. It has the following Krull dimension:
\begin{align}
\dim Z_0=\left\{\begin{array}{ll}
\mathop{rk}(G),&g=0,\,b=2,\\[5pt]
b\mathop{rk}(G),&\mathrm{otherwise}.
\end{array}\right.
\label{eq:dimcas}
\end{align}

\begin{proposition}
Inclusion $ Z\subset Z\subset A\subset A$ defines an affine superintegrable system.
\end{proposition}
\begin{proof}
First, note that $Z$ is finitely generated, indeed it is generated by $b$ copies of $\mathcal O[G]^G$, one for each boundary component. Every $\mathcal O[G]^G$ is in turn finitely generated.

By Lemma \ref{lemm:BoundaryPoissonCenter} we know that $Z$ is a subalgebra of the Poisson center of $A$. On the other hand, by results of \cite{GuruprasadHuebschmannJeffreyWeinstein'1997,BiswasGuruprasad'1993} we know that the general fiber of $\mathcal M_\Sigma^G\rightarrow\Spec Z$ is equipped with Poisson bracket of rank $\dim\mathcal M_\Sigma^G-\dim\Spec Z$.
\end{proof}

\subsection{Superintegrable system associated to non-separating curve}
Let $\Sigma$ be an oriented surface of genus $g$ with $b\geq 1$ boundary components. There exists no nonseparating curves on genus zero surface, so throught this section we can assume without loss of generality that $g>0$. Let $C=C_1\sqcup C_2\sqcup \dots\sqcup C_r$ be a nonseparating curve which is a disjoint union of $r,\;1\leqslant r\leqslant g$ simple closed curves on $\Sigma$ (see Figure \ref{fig:NonseparatingCurve}).
\begin{figure}
\begin{tikzpicture}[scale=1.5]
    \newcommand{\surfacecycle}[5][0]{
        \draw[thick, red] (#2,#3) to[out=180+#1,in=180+#1,looseness=0.5] (#4,#5);
        \draw[thick, dashed, red] (#2,#3) to[out=#1,in=#1,looseness=0.5] (#4,#5);
    }
    \surfacecycle{0}{3}{0}{2.27};
    \surfacecycle[225]{-2.93}{2.33}{-2.6}{1.8};
    \draw[thick,style={decoration={
  markings,
  mark=at position .32 with {\arrow{>}}},postaction={decorate}}] (0,0.5) to[out=80,in=240] (0.5,2.1) to[out=60,in=0] (0,2.6) to[out=180,in=120] (-0.5,2.1) to[out=300,in=100] (0,0.5);
    \draw[thick,->--] (0,0.5) to[out=160,in=270] (-1.9,1.4) to[out=90,in=30] (-2.75,2.1) to[out=210,in=140] (-2.5,1.1) to[out=320,in=170] (-1.3,0.5) to[out=350,in=190] (0,0.5);
    \fill (0,0.5) circle (0.06);
    \fill[gray] (0,0.5) circle (0.04);
    \draw[ultra thick] (-4.8,0) to[out=0,in=270] (-2.6,0.5) to[out=90,in=290] (-3.3,1.5) to[out=110,in=200] (-2.6,2.5) to[out=20,in=110] (-1.5,1.3) to[out=290,in=180] (-1,1) to[out=0,in=270] (-0.5,1.5) to[out=90,in=270] (-0.9,2.3) to[out=90,in=180] (0,3) to[out=0, in=90] (0.9,2.3) to[out=270,in=90] (0.5,1.5) to[out=270,in=180] (1,1) to[out=0,in=250] (1.5,1.3) to[out=70,in=160] (2.6,2.5) to[out=340,in=70] (3.3,1.5) to[out=250,in=90] (2.6,0.5) to[out=270,in=180] (4.8,0);
    \newcommand{\surfacehole}[3][0]{
        \draw[ultra thick, rotate around={#1:(#2,#3)}] (#2-0.2,#3+0.2) to[out=40,in=140] (#2+0.2,#3+0.2);
        \draw[ultra thick, rotate around={#1:(#2,#3)}] (#2-0.3,#3+0.25) to[out=320,in=220] (#2+0.3,#3+0.25);
        };
    \surfacehole{0}{2};
    \surfacehole[-45]{-2.6}{1.5};
    \surfacehole[45]{2.6}{1.5};
    \draw (0,0.4) node[below] {$p$};
    \draw (-3.1,2.5) node[red, above] {$C_1$};
    \draw (0, 3.1) node[red, above] {$C_2$};
    \draw (-1.63,1.2) node {$x_1$};
    \draw (0.68,2.15) node {$x_2$};
\end{tikzpicture}
\caption{An example of a nonseparating Curve}
\label{fig:NonseparatingCurve}
\end{figure}
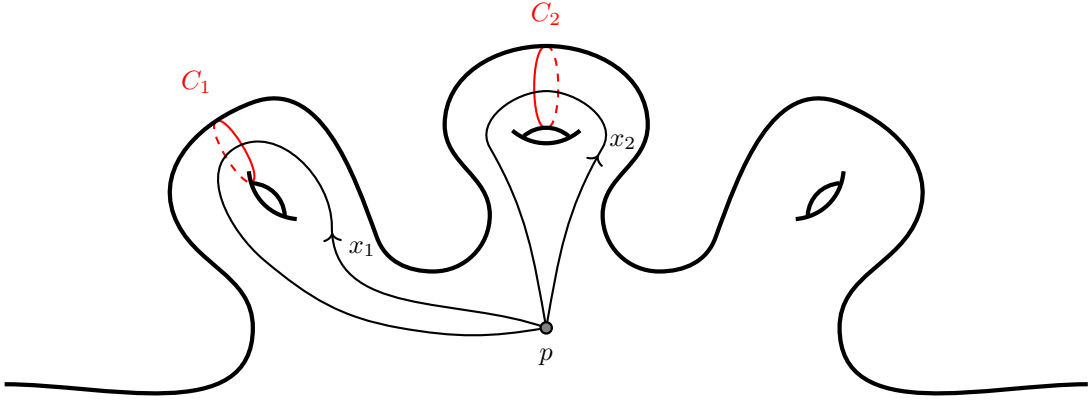

The fundamental group $\pi_1(\Sigma,p)$ is a free group with $2g+b-1$ generators, it admits the following presentation
\begin{align}
    \pi_1(\Sigma,p)=&\left\langle x_1,y_1,\dots,x_g,y_g,z_1,\dots,z_{b-1}\right\rangle,
\label{eq:GeneratorsNonseparating}
\end{align}
where the first $2r$ generators $x_i,\;y_i,\;1\leqslant i\leqslant r$ are chosen for handles containing $C_1,\dots,C_r$ as shown on Figure \ref{fig:NonSeparatingCycle}. This group contains a subgroup which consists of paths that do not cross $C$
\begin{equation}
\begin{tikzcd}
\pi_1(\Sigma,p\,|\,C)\ar[r,subset]&\pi_1(\Sigma,p)
\end{tikzcd}
\label{eq:subgroupC}
\end{equation}
\begin{align*}
    \pi_1(\Sigma,p\,|\,C)=\left\langle y_1,\dots, y_k, x_1y_1x_1^{-1},\dots, x_ry_rx_r^{-1},x_{r+1},y_{r+1},\dots,x_g,y_g,z_1,\dots,z_{b-1}\right\rangle.
\end{align*}
\begin{figure}
\begin{tikzpicture}[scale=1.5]
  \draw[thick, red] (-1.5,0) to[out=90,in=90,looseness=0.9] (-0.5,0);
  \draw[thick, dashed, red] (-1.5,0) to[out=270,in=270,looseness=0.9] (-0.5,0);
  \draw[thick, ->--] (1.5,0.2) to[out=155,in=0] (0,0.9) to[out=180,in=90] (-1,0.1) to[out=270,in=180] (0,-0.7) to[out=0,in=225] (1.5,0.2);
  \draw[thick,-<-] (1.5,0.2) to[out=100,in=-20] (0.9,0.95) to[out=160,in=220,looseness=0.6] (0,1.5);
  \draw[thick,dashed] (0,1.5) to[out=0,in=0,looseness=0.5] (0,0.21);
  \draw[thick] (0,0.21) to[out=180,in=200] (0.2,0.4) to[out=20,in=180] (1.5,0.2);
  \fill (1.5,0.2) circle (0.06);
  \fill[gray] (1.5,0.2) circle (0.04);
  \draw[ultra thick] (0,1.5) to[out=0,in=180] (2,0.8);
  \draw[ultra thick] (0,1.5) arc (90:270:1.5);
  \draw[ultra thick] (0,-1.5) to[out=0,in=180] (2,-0.8);
  \draw[ultra thick] (-0.5,0) to[out=50,in=130] (0.5,0);
  \draw[ultra thick] (-0.7,0.2) to[out=-60,in=240] (0.7,0.2);
  \draw (1.5,1.3) node[left] {$y_1$};
  \draw (-0.3,1.1) node {$x_1$};
  \draw (-1.9,0) node[red] {$C_1$};
\end{tikzpicture}
\caption{Connected Component of Nonseparating Curve}
\label{fig:NonSeparatingCycle}
\end{figure}
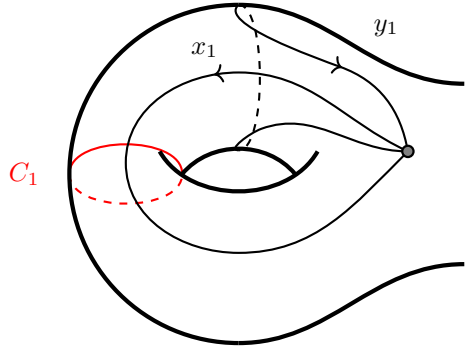

A $G$-representation of $\pi_1(\Sigma,p)$ is determined by $2g+b-1$ elements of $G$
\begin{equation}
\begin{aligned}
    x_i\mapsto X_i,\quad y_i\mapsto Y_i,\quad z_j\mapsto Z_j,\qquad 1\leqslant i\leqslant g,\quad 1\leqslant j\leqslant b-1.
\end{aligned}
\label{eq:NonseparatingCycleGRepresentation}
\end{equation}
As a corollary, $\mathrm{Hom}(\pi_1(\Sigma,p),G)$ is an affine variety with the coordinate ring $\mathcal O[\mathrm{Hom}(\pi_1(\Sigma,p),G)]\simeq\mathcal O[G]^{\otimes(2g+b-1)}$ isomorphic to the coordinate ring of $2g+b-1$ copies of $G$.

Inclusion (\ref{eq:subgroupC}) defines an inclusion of coordinate rings
\begin{equation}
\begin{tikzcd}
    \mathcal O[\mathrm{Hom}(\pi_1(\Sigma,p\,|\,C),G)]\ar[r,subset]&\mathcal O[\mathrm{Hom}(\pi_1(\Sigma,p),G)]\\
    \mathcal O[\mathrm{Hom}(\pi_1(\Sigma,p\,|\,C),G)]^G\ar[u,subset]\ar[r,subset]&\mathcal O[\mathrm{Hom}(\pi_1(\Sigma,p),G)]^G\ar[u,subset]
\end{tikzcd}
\label{eq:PoissonInclusionNonSeparating}
\end{equation}

At the same time, $\pi_1(\Sigma,p\,|\,C)$ contains $b$ infinite cyclic subgroups
generated by boundary cycles.
Let $Z_{\partial\Sigma}$ be a subalgebra of $\mathcal O[\mathcal M_\Sigma^G]$ generated by traces of powers of holonomies about the boundary components, or, equivalently, generated by $\mathcal O[\mathrm{Hom}(\langle z_i\rangle,G)]^G$ for $1\leqslant i\leqslant b$.

Similarly, let $B_C$ be a subalgebra of $\mathcal O[\mathcal M_\Sigma^G]$ generated by $Z_{\partial\Sigma}$ together with traces of holonomies about the cycles $C_i$, in other words, by $\mathcal O[\mathrm{Hom}(\langle y_i\rangle,G)]^G$ for $1\leqslant i\leqslant r$.

\begin{proposition}
The following chain of subalgebras
\begin{align*}
    Z_{\partial\Sigma}\;\subset\;B_C\;\subset\;J_{\Sigma\backslash C}\;\subset\;\mathcal O[\mathcal M_\Sigma^G],
\end{align*}
defines an affine superintegrable system. Here $J=\mathcal O[\mathrm{Hom}(\pi_1(\Sigma,p\,|\,C),G)]^G$ stands for the coordinate ring of the $G$-character variety of $\pi_1(\Sigma,p\,|\,C)$.
\label{prop:NonseparatingCurve}
\end{proposition}
\begin{proof}
By construction, all algebras involved are finitely generated and from (\ref{eq:PoissonBracketGraphFunctions}) it follows immediately that $\{Z,A\}=\{B,J\}=0$. Hence $Z$ is a subalgebra of the Poisson center of $\mathcal O[\mathcal M_\Sigma^G]$ and $B$ is Poisson commutative and centralized by $J$.
We have
\begin{align}
    \dim\Spec B_C-\dim\Spec Z_{\partial\Sigma}=r\,\mathrm{rank}(G).
    \label{eq:NonSepDim1}
\end{align}
On the other hand, since any representation of $\pi_1(\Sigma,p\,|\,C)$ is determined by a $2g+b-1$ elements of $G$
\begin{align*}
    Y_i,\quad \widetilde Y_j,\quad X_k,\quad Z_l,\qquad 1\leqslant i\leqslant g,\quad 1\leqslant j\leqslant r,\quad r+1\leqslant k\leqslant g,\quad 1\leqslant l\leqslant b-1,
\end{align*}
such that $Y_j\sim\widetilde Y_j$ belong to the same conjugacy class for all $j, 1\leqslant j\leqslant r$, we obtain
\begin{align}
    \dim \mathcal M_\Sigma^G-\dim\Spec J_{\Sigma\backslash C}=r\,\mathrm{rank}(G).
    \label{eq:NonSepDim2}
\end{align}
Combining (\ref{eq:NonSepDim1}) with (\ref{eq:NonSepDim2}) we conclude that balance of dimensions (\ref{eq:SuperIntDim}) holds.
\end{proof}

\subsection{Product of Surface Superintegrable systems associated with gluing}

\subsubsection{Choice of presentation of the fundamental groupoid.}
Let $\Sigma$ be an oriented surface with $b>0$ boundary components and $C=C_1\sqcup\dots\sqcup C_r$ be a disjoint union of $r$ simple closed curves in $\Sigma$ which cuts the surface into two parts $\Sigma=\Sigma_L\#_C\Sigma_R$. As in the statement of Theorem \ref{theorem:SuperIntSystemsCyclesIntro} we assume that none of the simple curves is homotopic to a boundary and all $C_1,\dots,C_r$ are pairwise nonhomotopic. As a corollary, neither $\Sigma_L$, nor $\Sigma_R$ can be a cylinder. Further, because surface $\Sigma$ has nontrivial boundary $\partial\Sigma\neq\emptyset$, at least one of the surfaces $\Sigma_L,\Sigma_R$ has a boundary component not in $C$. Without loss of generality we can assume that $\partial\Sigma_L\supsetneq C$.
\begin{figure}
\begin{tikzpicture}[scale=1.5]
  \draw[thick, red] (0,-0.8) to[out=180,in=180,looseness=0.5] (0,0.8);
  \draw[thick, dashed, red] (0,-0.8) to[out=0,in=0,looseness=0.5] (0,0.8);
  \draw[thick,->-] (-1.1,0.3) to[out=-20,in=200] (1.1,0.3);
  \draw[thick,->-] (1.1,0.3) to[out=135,in=220] (0.3,0.835);
  \draw[thick, dashed] (0.3,0.835) to[out=0,in=0,looseness=0.5] (0.3,-0.835);
  \draw[thick] (0.3,-0.835) to[out=160,in=250,looseness=0.7] (1.1,0.3);
  \draw[thick] (-1.1,0.3) to[out=0,in=160,looseness=0.8] (-0.6,0.92);
  \draw[thick, dashed] (-0.6,0.92) to[out=0,in=0,looseness=0.5] (-0.6,-0.92);
  \draw[thick,->-] (-0.6,-0.92) to[out=190,in=-40,looseness=0.8] (-1.1,0.3);
  \fill (1.1,0.3) circle (0.06);
  \fill[gray] (1.1,0.3) circle (0.04);
  \fill (-1.1,0.3) circle (0.06);
  \fill[gray] (-1.1,0.3) circle (0.04);
  \draw[ultra thick] (-4,0.8) to[out=0,in=180] (-2,1.5) to[out=0,in=180] (0,0.8) to[out=0,in=180] (2,1.5) to[out=0,in=180] (4,0.8);
  \draw[ultra thick] (-4,-0.8) to[out=0,in=180] (-2,-1.5) to[out=0,in=180] (0,-0.8) to[out=0,in=180] (2,-1.5) to[out=0,in=180] (4,-0.8);
  \draw[ultra thick] (-2.2,0) to[out=40,in=140] (-1.8,0);
  \draw[ultra thick] (-2.3,0.05) to[out=320,in=220] (-1.7,0.05);
  \draw[ultra thick] (1.8,0) to[out=40,in=140] (2.2,0);
  \draw[ultra thick] (1.7,0.05) to[out=320,in=220] (2.3,0.05);
  \draw (0,0.2) node[above] {$e_1$};
  \draw (1.05,0.75) node[left] {$b_1$};
  \draw (-1.38,-0.45) node[right] {$a_1$};
  \draw (-1.25,0.3) node[above] {$p$};
  \draw (1.25,0.3) node[below] {$q$};
  \draw (0,-1) node[below, red] {$C_1$};
\end{tikzpicture}
\caption{Connected component of a separating curve}
\label{fig:SeparatingCycle}
\end{figure}
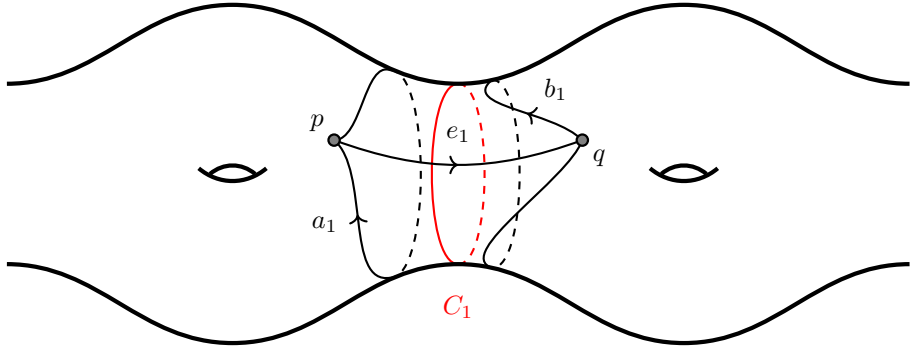

Because $\Sigma_L$ and $\Sigma_R$ both have at least one boundary component, their fundamental groups are free, say with $m$ and $n$ generators respectively. It will be convenient for us to choose the following presentation
\begin{align}
    \pi_1(\Sigma_L,p)=\langle a_1,\dots,a_m\rangle,\qquad
    \pi_1(\Sigma_R,q)=\langle b_1,\dots,b_n,b_{n+1}\,|\, b_1=F(b_2,\dots,b_{n+1})\rangle.
\label{eq:LRFundamentalGroupsDef}
\end{align}
Here $a_1,\dots,a_m$ are free generators of $\pi_1(\Sigma_L)$. Because $\partial\Sigma_L\supsetneq C$, we can choose the first $r$ generators $a_j,\;1\leqslant j\leqslant r$ to be homotopic to $C_j$ as shown on Figure \ref{fig:SeparatingCycle}. At the same time, for $\Sigma_R$ we will choose the presentation as follows: Let $b_2,\dots,b_{n+1}$ correspond to closed paths from $q$ to $q$ which freely generate $\pi_1(\Sigma_R,q)$. Without loss of generality we can assume that $b_2,\dots,b_r$ are homotopic to $C_2,\dots, C_r$. Now, let $b_1$ be the closed path homotopic to $C_1$, it can be expressed as a word in other generators $b_1=F(b_2,\dots,b_{n+1})$.\footnote{Adding extra generator $b_1$ for $\pi_1(\Sigma_R,q)$ satisfying relation $b_1=F(b_2,\dots,b_{n+1})$ allows us to treat both cases $\partial\Sigma_R=C$ and $\partial\Sigma_R\supsetneq C$ simultaneously without any changes in notation.}

Denote by $\pi_1(\Sigma,p,q)$ the full subcategory of the fundamental grouppoid $\pi_1(\Sigma)$ with two objects $p,q\in\Sigma$. As a grouppoid $\pi_1(\Sigma,p,q)$ is generated by the following set of morphisms
\begin{align}
    \pi_1(\Sigma,p,q)=\left\langle a_1,\dots,a_m,e_1,\dots,e_r,b_1,\dots,b_n,b_{n+1}\;\Big|\;
    \begin{array}{c}
    a_1=e_1^{-1}b_1 e_1,\dots,a_r=e_r^{-1}b_re_r,\\
    b_1=F(b_2,\dots,b_{n+1})
    \end{array}\right\rangle,
    \label{eq:CpqDef}
\end{align}
subject to $r+1$ relations.

Note that $\pi_1(\Sigma,p,q)$ is freely generated by $a_{r+1},\dots,a_m,e_1,\dots,e_r,b_2,\dots,b_{n+1}$. Presentation (\ref{eq:CpqDef}), however, is chosen to best describe inclusions of coordinate rings of $G$-representations discussed in the following subsection.

\subsubsection{Coordinate rings of $G$-representations and relations between them.} Any $G$-representation of $\pi_1(\Sigma,p,q)$ is defined by $m+n+r+1$ elements subject to $r+1$ relations:
\begin{equation}
\begin{aligned}
   &A_1,\dots,A_m,E_1,\dots, E_r,B_1,\dots,B_n,B_{n+1}\quad\in\quad G,\\
   &B_jE_j=E_jA_j,\quad 1\leqslant j\leqslant r,\qquad B_1=F(B_2,\dots,B_{n+1}).
\end{aligned}
\label{eq:CpqRepRelations}
\end{equation}
This makes the set $\mathrm{Hom}(\pi_1(\Sigma,p,q),G)$ of all $G$-representations of $\pi_1(\Sigma,p,q)$ into an algebraic subset of $G^{\times(m+n+r+1)}$. Moreover, $\mathrm{Hom}(\pi_1(\Sigma,p,q),G)$ comes equipped with a $G\times G$ action:
\begin{equation}
\begin{aligned}
    A_j\mapsto&\; gA_jg^{-1},\qquad 1\leqslant j\leqslant m,\\
    E_k\mapsto&\; hE_kg^{-1},\qquad 1\leqslant k\leqslant r,\qquad\qquad\qquad (g,h)\in G\times G,\\
    B_l\mapsto&\;hB_lh^{-1},\qquad 1\leqslant l\leqslant n+1,
\end{aligned}
\label{eq:GGActionSeparatingCycle}
\end{equation}

Of course, any $G$-representation of $\pi_1(\Sigma,p,q)$ also defines $G$-representations of $\pi_1(\Sigma_L,p)$ and $\pi_1(\Sigma_R,q)$. Explicitly, for our choice of generators we obtain a diagram of $G\times G$-equivariant projections
\begin{subequations}
\begin{equation}
\begin{tikzcd}
\mathrm{Hom}(\pi_1(\Sigma,p,q),G)\ar[r]\ar[d]&\mathrm{Hom}(\pi_1(\Sigma_R,q),G)\ar[d]\\
\mathrm{Hom}(\pi_1(\Sigma_L,p),G)\ar[r]&\mathrm{Hom}(\pi_1(C),G)\git G
\end{tikzcd}
\end{equation}
\begin{equation}
\begin{tikzcd}
(A_1,\dots,A_m,E_1,\dots,E_r,B_1,\dots, B_{n+1})\arrow[r,mapsto]\arrow[d,mapsto]&(B_1,\dots,B_{n+1})\arrow[d,mapsto]\\
(A_1,\dots,A_m)\arrow[r,mapsto]&(\left[A_1\right]=[B_1],\dots, \left[A_r\right]=[B_r])
\end{tikzcd}
\end{equation}
\label{eq:RepProjections}
\end{subequations}
where $[A_j]=[B_j]$ stands for the common conjugacy class of $A_j$ and $B_j$.

Let $I$ be a $G\times G$-invariant ideal in $\mathcal O[G^{\times(m+n+r+1)}]$ defined by equations (\ref{eq:CpqRepRelations}), then
\begin{align}
\mathcal O[\mathrm{Hom}(\pi_1(\Sigma,p,q),G)]\simeq\mathcal O[G^{\times(m+n+r+1)}]/I.
\label{eq:CpqFullCoordinateRing}
\end{align}

In terms of the coordinate ring (\ref{eq:CpqFullCoordinateRing}) of $G$-representations of $\pi_1(\Sigma,p,q)$ diagram (\ref{eq:RepProjections}) corresponds to the following diagram of inclusions of subalgebras
\begin{equation}
\begin{tikzcd}
\mathcal O[\mathrm{Hom}(\pi_1(\Sigma,p,q),G)]&\mathcal O[\mathrm{Hom}(\pi_1(\Sigma_R,q),G)]\arrow[l,subset]\\
\mathcal O[\mathrm{Hom}(\pi_1(\Sigma_L,p),G)]\arrow[u,subset]
&H=\mathcal O[\mathrm{Hom}(\pi_1(C),G)]^G\arrow[u,subset]\arrow[l,subset]
\end{tikzcd}.
\label{eq:GGEquivariantAlgebraInclusionsSeparatingCycle}
\end{equation}
Here $H=\mathcal O[\mathrm{Hom}(\pi_1(C),G)]^G$ is the coordinate ring of a $G$-character variety of a disjoint union of $r$ circles corresponding to connected components of $C=\partial\Sigma_L\cap\partial\Sigma_R=C_1\sqcup\dots\sqcup C_r$. Hence $H\simeq\left(\mathcal O[G]^G\right)^{\otimes r}$ is generated by invariant functions on common conjugacy classes $[A_i]=[B_i],\;1\leqslant i\leqslant r$ of matrices $A$ and $B$.
\begin{lemma}
\begin{align}
H=\mathcal O[\mathrm{Hom}(\pi_1(\Sigma_L,p))]\cap\mathcal O[\mathrm{Hom}(\pi_1(\Sigma_R,q))]
\label{eq:HInteresection}
\end{align}
\label{lemm:HIntersection}
\end{lemma}
\begin{proof}
From (\ref{eq:GGEquivariantAlgebraInclusionsSeparatingCycle}) we already know that $H\subset \mathcal O[\mathrm{Hom}(\pi_1(\Sigma_L,p))]\cap\mathcal O[\mathrm{Hom}(\pi_1(\Sigma_R,q))]$, so we have to prove only the opposite inclusion. To this end, first recall that all algebras in (\ref{eq:GGEquivariantAlgebraInclusionsSeparatingCycle}) are $G\times G$-modules w.r.t. the action (\ref{eq:GGActionSeparatingCycle}). Moreover, $1\times G$ is acting trivially on $\mathcal O[\mathrm{Hom}(\pi_1(\Sigma_L,p))]$, while $G\times1$ is acting trivially on $\mathcal O[\mathrm{Hom}(\pi_1(\Sigma_R,q))]$. As a corollary, their intersection lies in the $G$-invariant part of both subalgebras $\mathcal O[\mathrm{Hom}(\pi_1(\Sigma_L,p))]\cap\mathcal O[\mathrm{Hom}(\pi_1(\Sigma_R,q))]\subset \mathcal O[\mathcal M_{\Sigma_L}^G]\cap\mathcal O[\mathcal M_{\Sigma_R}^G]$. By Theorem \ref{prop:InvariantSpanningSet} combined with Corollary \ref{cor-isom} we conclude that the right hand side of (\ref{eq:HInteresection}) is spanned by invariant graph functions on $C$, hence contained in $H$.
\end{proof}
Now consider the following tensor product over $H$
\begin{align}
    J_0=\mathcal O[\mathrm{Hom}(\pi_1(\Sigma_L,p),G)]\underset{H}\otimes\mathcal O[\mathrm{Hom}(\pi_1(\Sigma_R,q),G)].
\label{eq:J0def}
\end{align}
By Lemma \ref{lemm:HIntersection}, $J\subset \mathcal O[\mathrm{Hom}(\pi_1(\Sigma,p,q),G)]$, this gives us another subalgebra of $\mathcal O[\mathrm{Hom}(\pi_1(\Sigma,p,q),G)]$ which fits in the middle of diagram (\ref{eq:GGEquivariantAlgebraInclusionsSeparatingCycle}).

By construction, all rings involved in (\ref{eq:GGEquivariantAlgebraInclusionsSeparatingCycle}) and (\ref{eq:J0def}) are finitely generated integral domains. Hence, taking the spectrum of all rings in (\ref{eq:GGEquivariantAlgebraInclusionsSeparatingCycle}) and (\ref{eq:J0def}) we obtain a commutative diagram of $G\times G$-equivariant dominant maps of affine schemes
\begin{equation}
\begin{tikzcd}
\mathrm{Hom}(\pi_1(\Sigma,p,q),G)\arrow[rrd,bend left=10]\arrow[ddr,bend right]\arrow[dr,dashed,"\gamma_0"]\\[30pt]
&\mathrm{Hom}(\pi_1(\Sigma_L,p),G)\underset{\Spec H}\times\mathrm{Hom}(\pi_1(\Sigma_R,q),G)\arrow[r]\arrow[d] \arrow[rd,dashed,"\beta_0"] &\mathrm{Hom}(\pi_1(\Sigma_R,q),G)\arrow[d,"\beta_R"]\\
&\mathrm{Hom}(\pi_1(\Sigma_L,p),G)\arrow[r,"\beta_L"']&\Spec H
\end{tikzcd}
\label{eq:SeparatingCycleMainDiag}
\end{equation}
where $Spec(H)\simeq(G\git \mathrm{Ad}\,G)^r$ and
\begin{align*}
\mathrm{Hom}(\pi_1(\Sigma_L,p),G)\underset{\Spec H}\times\mathrm{Hom}(\pi_1(\Sigma_R,q),G)=\Spec J_0
\end{align*}
is the fiber product of maps $\beta_L$ and $\beta_R$.

\subsubsection{Product of Superintegrable systems.}
Taking the GIT quotient by $G\times G$ action of diagram (\ref{eq:SeparatingCycleMainDiag}) we obtain from the main diagonal the following chain of dominant maps of affine schemes
\begin{align}
\mathrm{Hom}(\pi_1(\Sigma,p,q),G)\git (G\times G)\xrightarrow{\;\gamma\;}\left(\mathrm{Hom}(\pi_1(\Sigma_L,p),G)\underset{\Spec H}\times\mathrm{Hom}(\pi_1(\Sigma_R,q),G)\right)\git (G\times G)\xrightarrow{\;\beta\;} \Spec H
\label{eq:PoissonProjectionSeparatingCycle}
\end{align}
Note that by (\ref{moduli-isom}), we have natural isomorphism
\begin{align*}
\mathrm{Hom}(\pi_1(\Sigma,p,q),G)\git (G\times G)\;\simeq\;\mathcal M_\Sigma^G
\end{align*}

\begin{remark}
Key idea of the remaining part of this subsection is that projections (\ref{eq:PoissonProjectionSeparatingCycle}) define an affine superintegrable system on $\mathcal O[\mathcal M_\Sigma^G]$ associated with the separating collection of cycles $C=\partial\Sigma_L\cap\partial\Sigma_R=C_1\sqcup\dots\sqcup C_r$.

Moreover, we can actually prove a much more general statement, formulated in Proposition \ref{prop:SeparatingCurve}, that if each of the surfaces $\Sigma_L,\Sigma_R$ was already equipped with an affine superintegrable system on $\mathcal O[\mathcal M_{\Sigma_L}]$ and $\mathcal O[\mathcal M_{\Sigma_R}]$ respectively, then gluing along $C$ provides a way to construct a new superintegrable system on $\Sigma$ which refines (\ref{eq:PoissonProjectionSeparatingCycle}).
\end{remark}

\begin{lemma}
Equation (\ref{eq:PoissonProjectionSeparatingCycle}) defines a chain of Poisson projections.
\end{lemma}
\begin{proof}
Note that the neighbourhood of $C$ is a disjoint union of $r$ annuli. Recall that $H\simeq\left(\mathcal O[G]^G\right)^{\otimes r}$ is the tensor product of $r$ copies of the coordinate ring of the moduli space of $G$-representations of an annulus. It is naturally equipped with the trivial Poisson structure and with two natural injective Poisson homomorphisms
\begin{equation}
\begin{tikzcd}
H\ar[r,hookrightarrow,"\beta_L^*"]&\mathcal O[\mathcal M_{\Sigma_L}^G]
\end{tikzcd}
\qquad\qquad
\begin{tikzcd}
H\ar[r,hookrightarrow,"\beta_R^*"]&\mathcal O[\mathcal M_{\Sigma_R}^G]
\end{tikzcd}
\end{equation}
to the Poisson center of the coordinate rings of $G$-character varieties of $\pi_1(\Sigma_L,p)$ and $\pi_1(\Sigma_R,q)$ respectively. As a result,
\begin{equation*}
\begin{tikzcd}
H\ar[r,hookrightarrow,"\quad\beta^*"]&J=\mathcal O[\mathcal M_{\Sigma_L}^G]\underset{H}\otimes\mathcal O[\mathcal M_{\Sigma_R}^G]
\end{tikzcd}
\end{equation*}
is a Poisson embedding into the Poisson center of $J$. To finalize the proof, note that
\begin{align*}
J\xhookrightarrow{\quad\gamma^*\quad} \mathrm{Hom}(\pi_1(\Sigma,p,q),G)^{G\times G}\simeq \mathcal O[\mathcal M_\Sigma^G]
\end{align*}
is a subalgebra of the full coordinate ring of the moduli space of $G$-representations of $\pi_1(\Sigma,p,q)$.
\end{proof}

\begin{proposition}
Suppose that we have a pair of affine superintegrable systems on $\Sigma_L$ and $\Sigma_R$
defined by some choice of nonseparating cycles $C_L\subset \Sigma_L$ and $C_R\subset \Sigma_R$:
\begin{align}
Z_L\subset B_L\subset J_L\subset A_L=\mathcal O[\mathcal M_{\Sigma_L}^G],\qquad\qquad Z_R\subset B_R\subset J_R\subset A_R=\mathcal O[\mathcal M_{\Sigma_R}^G],
\label{eq:LRSystems}
\end{align}
where $Z_L$ and $Z_R$ denote Poisson centers of $A_L$ and $A_R$ respectively. Then
\begin{align}
Z\;\subset\; B\;\subset\; J\;\subset\;A=\mathcal O[\mathcal M_\Sigma^G],
\label{eq:GluedSuperInt}
\end{align}
defines an affine superintegrable system on $\Sigma$. Here $Z$ is the Poisson center of $\mathcal O[\mathcal M_\Sigma^G]$,
\begin{align}
B=B_L\underset{H}\otimes B_R,\qquad\qquad J=J_L\underset{H}\otimes J_R,
\label{eq:SeparatingPropBJDef}
\end{align}
and $H=\mathcal O[\mathcal M_C^G]$ is the coordinate ring of the $G$-character variety of $\pi_1(C)\simeq Free_r$.
\label{prop:SeparatingCurve}
\end{proposition}
\begin{proof}
First, note that
\begin{align*}
\left\{\mathcal O[\mathcal M_{\Sigma_L}^G],\mathcal O[\mathcal M_{\Sigma_R}^G]\right\}=0,
\end{align*}
where $\{,\}$ stands for the Poisson bracket on $\mathcal M_\Sigma^G$. Indeed, by Proposition \ref{prop:InvariantSpanningSet} subalgebras $\mathcal O[\mathcal M_{\Sigma_L}^G]$ and $\mathcal O[\mathcal M_{\Sigma_R}^G]$ are spanned by invariant graph functions which have no intersections. Combining it with the fact that $\{B_L,J_L\}=0=\{B_R,J_R\}$ we conclude that $\{B,J\}=0$, and, in particular, that $B$ has a trivial Poisson bracket. So the only thing we have to prove is that Krull dimensions of the algebras satisfy (\ref{eq:SuperIntDim}).

Krull dimensions of algebras $B$ and $J$ are determined as follows
\footnote{These two formulas is an algebraic version of a statement that for a pair of projections of irreducible affine varieties $X\rightarrow Z$ and $Y\rightarrow Z$ over $\mathbb C$, the dimension of the corresponding fiber product is $\dim(X\times_ZY)=\dim X+\dim Y-\dim Z$.

Note that in order for (\ref{eq:dimBJSeparating}) to be valid it is crucial that all algebras involved are finitely generated over $\mathbb C$ and have no zero divisors.
Identities (\ref{eq:dimBJSeparating}) then follow from the fact that transcendence degree of a field extension is additive w.r.t. the composition of field extensions.

}
\begin{align}
    \dim B=\dim B_L+\dim B_R-\dim H,\qquad \dim J=\dim J_L+\dim J_R-\dim H.
\label{eq:dimBJSeparating}
\end{align}


Now denote by $g_L$ and $b_L$ the genus and number of boundary components of $\Sigma_L$. Similarly, let $g_R$ and $b_R$ denote the genus and number of boundary components of $\Sigma_R$. The resulting surface $\Sigma$ then has a genus $g=g_L+g_R+(r-1)$ and number of boundary components $b=b_L+b_R-2r$. Note that we have assume that neither $\Sigma_L$ nor $\Sigma_R$ is a cylinder (i.e. $(g_L,b_L)\neq(0,2)\neq(g_R,b_R)$). From (\ref{eq:dimBJSeparating}) we have
\begin{align*}
\dim B+\dim J=& \dim B_L+\dim B_R+\dim J_L+\dim J_R-2\dim H
\end{align*}
On the other hand, because both systems (\ref{eq:LRSystems}) are superintegrable, by (\ref{eq:SuperIntDim}) we conclude that
\begin{align*}
\dim B+\dim J=&\dim A_L+\dim Z_L+\dim A_R+\dim Z_R-2\dim H.
\end{align*}
Finally, using explicit formula (\ref{eq:dimrepscheme}) for Krull dimensions of $A_L,A_R$ and $(\ref{eq:dimcas})$ for $Z_L,Z_R$ we obtain
\begin{align*}
\dim B+\dim J=&((2g_L+b_L-2)+(2g_R+b_R-2))\dim G+(b_L+b_R-2r)\text{rank}(G)\\
=\quad\,&(2g+b-2)\dim G+b\text{rank}(G)\\
=&\dim A+\dim Z.
\end{align*}
Hence, condition (\ref{eq:SuperIntDim}) on Krull dimensions of subalgebras (\ref{eq:GluedSuperInt}) is satisfied and the resulting system is superintegrable.
\end{proof}

\subsection{The proof of Theorem \ref{theorem:SuperIntSystemsCycles}}

Now we can prove our main result.

\begin{proof}
Let $\Sigma\backslash C=\Sigma_1\sqcup\dots\sqcup\Sigma_m$, where $\Sigma_1,\dots,\Sigma_m$ stand for connected components. We will prove the statement of the theorem by induction in the number $m$ of connected components.

The base case $m=1$ follows by Proposition \ref{prop:NonseparatingCurve}.

Now assume that the statement of the theorem holds for all $(\Sigma',C')$ s.t. $\Sigma'\backslash C'$ has at most $m$ connected components. Then for each pair $(\Sigma,C)$ s.t. $\Sigma\backslash C$ has $m+1$ connected components we can present $\Sigma=\Sigma_L\sqcup\Sigma_R$, where both $\Sigma_L\backslash C$ and $\Sigma_R\backslash C$ have at most $m$ connected components. Define
\begin{align*}
    C_L=C\backslash\partial\Sigma_R,\qquad C_R=C\backslash\partial\Sigma_L,\qquad C_{LR}=\partial\Sigma_L\cap\partial\Sigma_R.
\end{align*}
Note that $C=C_L\sqcup C_{LR}\sqcup C_R$. By the inductive assumption we have two superintegrable systems
\begin{align*}
    Z_{\partial\Sigma_L}\subset B_{C_L}\subset J_{\Sigma\backslash C_L}\subset \mathcal O[\mathcal M^G_{\Sigma_L}],\qquad\qquad Z_{\partial\Sigma_R}\subset B_{C_R}\subset J_{\Sigma\backslash C_R}\subset \mathcal O[\mathcal M^G_{\Sigma_R}].
\end{align*}
Now, applying Proposition \ref{prop:SeparatingCurve} we finalize the proof.
\end{proof}


The following proposition is an immediate consequence.

\begin{proposition}
Let $\Sigma$ be an oriented closed surface with $b>0$ boundary components and $C=C_1\sqcup\dots\sqcup C_k$ and $C'=C_1'\sqcup\dots\sqcup C_m'$ be two collections of simple closed curves none of which is homotopic to a boundary component. Denote by $\mathcal I_C$ and $\mathcal I_{C'}$ the two superintegrable systems associated to $C$ and $C'$ as in Theorem \ref{theorem:SuperIntSystemsCycles}.
\begin{itemize}
    \item If $C\subset C'$ is a subcollection of $C'$, the integrable system $\mathcal I_C$ is a refinement of $\mathcal I_{C'}$.
    \item If an element of the Mapping Class Group of surface $\Sigma$ bring the collection $C$ to $C'$, it induces an equivalence of corresponding suerintegrable systems.
    \end{itemize}
\end{proposition}

\begin{remark}
It is natural to expect that for an affine superintegrable system
\begin{align*}
Z_{\partial\Sigma}\subset B_C\subset J_{\partial\Sigma\backslash C}\subset \mathcal O[\mathcal M_\Sigma^G],
\end{align*}
subalgebra $Z_{\partial\Sigma}$ is the full Poisson center of $\mathcal O[\mathcal M_\Sigma^G]$, and this is true in
the compact case.  Note that this implies that $J_{\Sigma\backslash C}$ is the full Poisson centralizer of $B_C$.
\end{remark}

\subsection{Beyond generic orbits}

The main result of the paper, formulated in Theorem \ref{theorem:SuperIntSystemsCyclesIntro}, states that
a collection of simple closed curves $C=C_1\sqcup\dots\sqcup C_r$ on an oriented surface $\Sigma$ defines a family of superintegrable systems. This family is parameterized by generic values of Casimir functions, in other words, by
fixing holonomies around boundary components at regular conjugation orbits. We have left details of
the case of nongeneric values of Casimir functions outside of the scope of the current paper.


We expect to prove the following stronger statement in a sequel publication. Consider a superintegrable system from
Theorem \ref{theorem:SuperIntSystemsCyclesIntro} and let
\begin{equation*}
\begin{tikzcd}
\cA\arrow[r, "p_1"]\arrow[rrr,bend right=17,"p"']&\mathcal J\arrow[r,"p_2"]&\mathcal B\arrow[r,"p_3"]&\mathcal Z
\end{tikzcd}
\end{equation*}
be an associated chain of dominant maps of affine schemes preserving the Poisson bracket. Then for all $z\in\Spec\mathcal Z$, every irreducible component of the fiber $p^{-1}(z)$ contains an open subset $\mathcal M_{2n}$ such that
\begin{equation*}
\begin{tikzcd}
\mathcal M_{2n}\arrow[r,"\pi_1"]&\mathcal P_{2n-k}\arrow[r,"\pi_2"]&\mathcal B_{k}
\end{tikzcd}
\end{equation*}
is a superintegrable system. Here $\pi_1,\pi_2$ are maps induced by $p_1$ and $p_2$ respectively.

Example of systems with lower rank orbits were studied in \cite{CF}. In that case $G=SL_m$,
the surface $\Sigma$ is a torus with $n$ punctures, the conjugacy class of the monodromy around each
puncture is fixed and is assumed to be of rank 1. It is easy to see that the dimension of the moduli space
in this case is $2(m-1)n$. The choice of $C=C_1,\dots C_n$ being separating cycles
cutting the torus into $n$ cylinders with one puncture in each cylinder produces a Liouville
integrable system with $(m-1)n$ Poisson commuting integrals. For any $C_i$, the choice $C=C_i$
gives a superintegrable refinement of this Liouville integrable system.

\section{Hamilton flows on the full coordinate ring}
\label{sec:HamiltonFlows}

The Poisson bracket (\ref{eq:PoissonBracketGraphFunctions}) on the coordinate ring $\mathcal O[\mathrm{Hom}(\pi_1(\Sigma,p_1,\dots,p_m))]^G$ of the moduli space $\mathcal M_{\Sigma}^G$ can be obtained from the bivector field on $R=\mathcal O[\mathrm{Hom}(\pi_1(\Sigma,p_1,\dots,p_m))]$ using two essentially different methods. One way is the ``$r$-matrix approach'' developed in \cite{FockRosly'1993}, the other way is the ``Quasi Poisson'' bracket approach developed in \cite{AlekseevMalkinMeinrenken'1998, AlekseevKosmann-Schwarzbach'2000, AlekseevKosmann-SchwarzbachMeinrenken'2002}.

In the $r$-matrix approach one defines Poisson bracket on the full coordinate ring
\begin{align}
\{,\}_{r-matrix}:R\otimes R\rightarrow R,\qquad R=\mathcal O[\mathrm{Hom}(\pi_1(\Sigma,p_1,\dots,p_m),G)],
\label{eq:RMatrixBracketGeneral}
\end{align}
which coincides with the Goldman bracket on the $G$-invariant part $R^G\simeq\mathcal O[\mathcal M_\Sigma^G]$.

In the Quasi Hamiltonian approach the bivectorfield on $\mathrm{Hom}(\pi_1(\Sigma,p_1,\dots,p_m))$ is not Poisson. The bracket $\{,\}_{quasi}: R\otimes R\rightarrow R$ satisfies the Jacobi identity only up to a fixed trivector field vanishing on $R^G$. As in the $r$-matrix approach, quasi Poisson bracket induces the Goldman bracket on $R^G$
\begin{align*}
\{,\}_{r-matrix}\Big|_{R^G\otimes R^G}=\{,\}_{quasi}\Big|_{R^G\otimes R^G}=\{,\}_{Goldman}:\quad R^G\otimes R^G\rightarrow R^G.
\end{align*}

Here we will focus on the $r$-matrix approach.\footnote{For the analogue of the current section in terms of quasi Poisson bracket see \cite{ArthamonovRoubtsov}.} The bracket (\ref{eq:RMatrixBracketGeneral}) makes $R$ into a left Lie module over $R^G$. In other words, the $G$-invariant part acts on the full coordinate ring by derivations
\begin{align}
\{,\}:\quad R^G\otimes R\rightarrow R,
\label{eq:PoissonModuleAction}
\end{align}
satisfying
\begin{align*}
&\{h,ab\}=\{h,a\}b+a\{h,b\},\qquad\textrm{for all}\quad h\in R^G,\quad a,b\in R,\\
&\{h_1,\{h_2,a\}\}-\{h_2,\{h_1,a\}\}=\{\{h_1,h_2\}_{Goldman},a\},\qquad\textrm{for all}\quad h_1,h_2\in R^G,\quad a\in R.
\end{align*}

Action (\ref{eq:PoissonModuleAction}) allows one to integrate Hamilton flows described in Theorem \ref{theorem:SuperIntSystemsCycles}. Moreover, with the appropriate choice of generators of the fundamental group, this action acquires a remarkably simple form presented in Theorem \ref{theorem:HamiltonFlows} below. This can be viewed as an analogue of (degenerate) separation of variables in the context of character varieties.

In subsection \ref{sec:ChoiceOfGenerators} for a given collection of simple closed curves as in Theorem \ref{theorem:SuperIntSystemsCycles} we define a particular choice of marked points and generators of $\pi_1(\Sigma,p_1,\dots,p_m)$. Next, in subsection \ref{sec:HamiltonFlowsRMatrix} we calculate action (\ref{eq:PoissonModuleAction}) via the r-matrix.

\subsection{Choice of generators of $\pi_1$}
\label{sec:ChoiceOfGenerators}

Let $\Sigma$ be, as before, an oriented surface with nonempty boundary $\partial\Sigma=D_1\sqcup\dots\sqcup D_b$ and let $C=C_1\sqcup \dots\sqcup C_k$ be a disjoint collection of simple closed curves as in Theorem \ref{theorem:SuperIntSystemsCycles}. Namely, we require that all $C_j$ are pairwise nonhomotopic and none of them are homotopic to the boundary. Let $\Sigma\backslash C=\Sigma_1\sqcup\dots\sqcup\Sigma_m$, where $\Sigma_1,\dots,\Sigma_m$ stand for connected components. For the purpose of this section we further assume that $\partial\Sigma_j\cap\partial \Sigma\neq\emptyset$ for all $j\in\{1,\dots,m\}.$ This assumption can be made with no loss generality when we describe Hamilton flows, because initial conditions with trivial monodromy about the boundary component effectively correspond to absence of the boundary component.

For any connected component $\Sigma_j$ we choose a single marked point $p_j\in \partial\Sigma_j\cap\partial\Sigma$ on the common boundary with $\Sigma$. Each surface $\Sigma_j$ is an oriented surface of genus $g_j$ with $b_j$ boundary components. Its boundary can be presented as a disjoint union
\begin{align*}
\partial\Sigma_j=D_{\beta_{j,1}}\sqcup\dots\sqcup D_{\beta_{j,r_j}}\sqcup  C_{\alpha_{j,r_j+1}}\sqcup\dots\sqcup C_{\alpha_{j,b_j}}
\end{align*}
of $r_j$ connected components $D_{\beta_{j,1}},\dots,D_{\beta_{j,r_j}}\subset \partial\Sigma$ of the boundary of the full surface $\Sigma$, as well as $b_j-r_j$ connected components $C_{\alpha_{j,r_j+1}},\dots, C_{\alpha_{j,b_j}}\subset C$ of a curve $C$. Without loss of generality we can assume that $p_j\in D_{\beta_{j,1}}$.

Denote by
\begin{align}
H_j=(C\cap\Sigma_j)\backslash\partial\Sigma_j=C_{\gamma_{j,1}}\sqcup\dots\sqcup C_{\gamma_{j,s_j}},
\label{eq:NonSeparatingCurvesSingleComponent}
\end{align}
the subcollection of simple closed curves in $C$ which belong to $\Sigma_j$, but do not belong to $\partial\Sigma_j$. Such subcollection is nonseparating by definition.

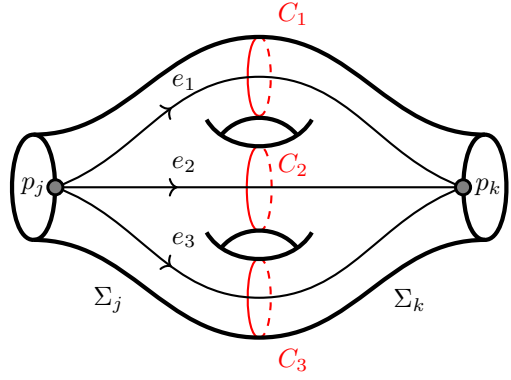
\begin{figure}
\begin{tikzpicture}
\draw[thick,red,looseness=0.5] (0,0.95) to[out=180,in=180] (0,2);
\draw[thick,red,looseness=0.5,dashed] (0,0.95) to[out=0,in=0] (0,2);
\draw[thick,red,looseness=0.5] (0,-0.95) to[out=180,in=180] (0,-2);
\draw[thick,red,looseness=0.5,dashed] (0,-0.95) to[out=0,in=0] (0,-2);
\draw[thick,red,looseness=0.5] (0,-0.58) to[out=180,in=180] (0,0.55);
\draw[thick,red,looseness=0.5,dashed] (0,-0.58) to[out=0,in=0] (0,0.55);
\draw[thick,->--] (-2.71,0) to[out=20,in=180] (0,1.47) to[out=0,in=160] (2.71,0);
\draw[thick,->--] (-2.71,0) to[out=-20,in=180] (0,-1.47) to[out=0,in=200] (2.71,0);
\draw[thick,->--] (-2.71,0) to (2.71,0);
\draw[ultra thick,looseness=0.7] (-3,-0.7) to[out=0,in=0] (-3,0.7) to[out=180,in=180] (-3,-0.7);
\draw[ultra thick,looseness=0.7] (3,-0.7) to[out=0,in=0] (3,0.7) to[out=180,in=180] (3,-0.7);
\draw[ultra thick] (-3,0.7) to[out=0,in=180] (0,2) to[out=0,in=180] (3,0.7);
\draw[ultra thick] (-3,-0.7) to[out=0,in=180] (0,-2) to[out=0,in=180] (3,-0.7);
\draw[ultra thick] (-0.5,0.7) to[out=50,in=130] (0.5,0.7);
\draw[ultra thick] (-0.7,0.9) to[out=-60,in=240] (0.7,0.9);
\draw[ultra thick] (-0.5,-0.8) to[out=50,in=130] (0.5,-0.8);
\draw[ultra thick] (-0.7,-0.6) to[out=-60,in=240] (0.7,-0.6);
\fill (-2.71,0) circle (0.12);
\fill[gray] (-2.71,0) circle (0.08);
\fill (2.71,0) circle (0.12);
\fill[gray] (2.71,0) circle (0.08);
\draw[red] (0.45,2.3) node {$C_1$};
\draw[red] (0.45,0.3) node {$C_2$};
\draw[red] (0.45,-2.3) node {$C_3$};
\draw (-1,1.4) node {$e_1$};
\draw (-1,0.3) node {$e_2$};
\draw (-1,-0.7) node {$e_3$};
\draw (-3,0) node {$p_j$};
\draw (3.05,0) node {$p_k$};
\draw (-2,-1.5) node {$\Sigma_j$};
\draw (2,-1.5) node {$\Sigma_k$};
\end{tikzpicture}
\caption{Adjacent pair $\Sigma_j\cup\Sigma_k$}
\label{fig:AdjacentPair}
\end{figure}

We will choose a set of generators of $\pi_1(\Sigma,p_1,\dots,p_n)$ as follows:
\begin{enumerate}
\item First, for each simple closed curve $C_i\in \partial\Sigma_j\cap \partial\Sigma_k$ which separates a pair of distinct connected components $\Sigma_j\neq\Sigma_k$ of $\Sigma\backslash C$ we choose a simple arc $e_i\in (\Sigma_j\backslash H_j)\cup(\Sigma_k\backslash H_k)$ from $p_j$ to $p_k$ passing once through $C_j$ as shown on Figure \ref{fig:AdjacentPair}.

    Note that without loss of generality we can assume that different $e_j$ intersect only at the endpoints. This can be achieved by sliding all intersection points between $e_i$ and $e_{i'}$ in $\Sigma_j$ to $p_j$ followed by left multiplication by an element of $\pi_1(\Sigma_j,p_j)$. Similarly, one can get rid of the intersection points in $\Sigma_k$.
\item For each arc $e_i$ chosen on the previous step we pick a simple closed curve $a_i$ which starts and ends at $p_j$ and homotopic to $C_i$ as shown on Figure \ref{fig:SeparatingCycle}.

    Again, we can assume that none of the $a_i,e_l$ intersect away from the endpoints. This can be achieved by choosing $a_i$ sufficiently close to $e_i\cup C_i$. Also, each $a_i$ chosen on this step cuts $\Sigma_j$ into two connected components, one of them is a cylinder given by the neighbourhood of $C_i$, the other is equivalent to the original surface. In what follows we assume that all further curves never enter into the cylinder.
\item For each nonseparating simple closed curve $C_{\gamma_{j,k}}$ from (\ref{eq:NonSeparatingCurvesSingleComponent}) we choose a pair of simple closed curves $x_{\gamma_{j,k}}$ and $y_{\gamma_{j,k}}$, where $y_{\gamma_{j,k}}$ is homotopic to $C_{\gamma_{j,k}}$, while $x_{\gamma_{j,k}}$ passes once through $C_{\gamma_{j,k}}$ as shown on Figure \ref{fig:NonSeparatingCycle}.
\item At the last step we consider $\Sigma_j\backslash C$. This surface has genus $g_j-s_j$ and $b_j+2s_j$ boundary components. Our choice of arcs already generates the paths around $r_j+2s_j$ boundary components of $\Sigma_j\backslash C$. We can choose the remaining $2g_j+b_j-r_j$ free generators of $\pi_1(\Sigma_j\backslash C)$ arbitrarily.
\end{enumerate}

Collection of marked points and free generators of $\pi_1(\Sigma,p_1,\dots,p_m)$ chosen above defines an ordered oriented ribbon graph with vertices $p_1,\dots,p_m$ and edges corresponding to generators. The total order of half-edges adjacent to the given vertex is defined by an orientation of the surface. An example of the total order of half-edges is given on Figure \ref{fig:OrderHalfEdges}.
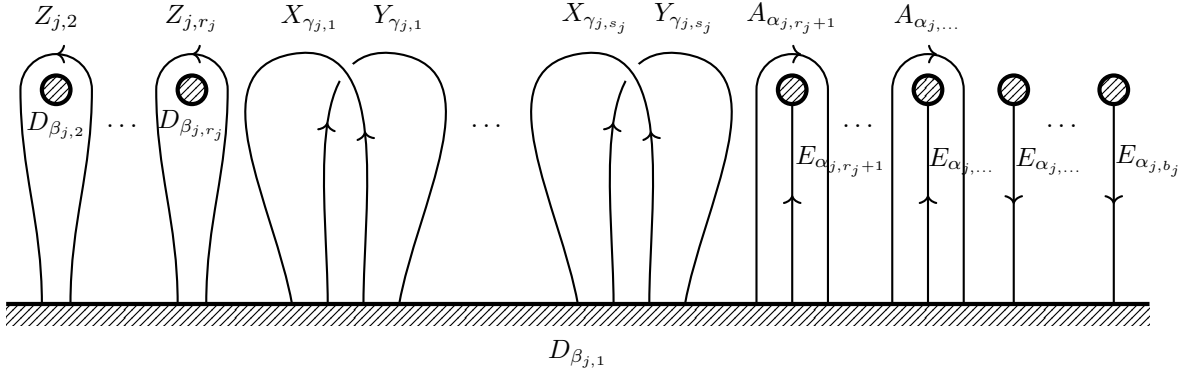
\begin{figure}
\begin{tikzpicture}[scale=0.95]
\draw[thick,->-] (-7.1,0) to[out=90,in=270] (-6.8,3) to[out=90,in=0] (-7.3,3.5) to[out=180,in=90] (-7.8,3) to[out=270,in=90] (-7.5,0);
\draw[thick,->-] (1.4-7.1+0.5,0) to[out=90,in=270] (1.4-6.8+0.5,3) to[out=90,in=0] (1.4-7.3+0.5,3.5) to[out=180,in=90] (1.4-7.8+0.5,3) to[out=270,in=90] (1.4-7.5+0.5,0);
\draw[thick,--<-] (-3+0.5,0) to[out=80,in=-10] (-3+0.5,3.5) to[out=170,in=90] (-4+0.5,0);
\fill[white] (-3.75+0.5,3.25) circle (0.15);
\draw[thick,->--] (-3.5+0.5,0) to[out=90,in=10] (-4.5+0.5,3.5) to[out=190,in=100] (-4.5+0.5,0);
\draw[thick,--<-] (4-3+0.5,0) to[out=80,in=-10] (4-3+0.5,3.5) to[out=170,in=90] (4-4+0.5,0);
\fill[white] (4-3.75+0.5,3.25) circle (0.15);
\draw[thick,->--] (4-3.5+0.5,0) to[out=90,in=10] (4-4.5+0.5,3.5) to[out=190,in=100] (4-4.5+0.5,0);
\draw[thick,-<-] (3+0.5,0) to[out=90,in=270] (3+0.5,3) to[out=90,in=0] (2.5+0.5,3.5) to[out=180,in=90] (2+0.5,3) to[out=270,in=90] (2+0.5,0);
\draw[thick,-<-] (2.5+0.5,2.8) to (2.5+0.5,0);
\draw[thick,-<-] (4.4+1,0) to[out=90,in=270] (4.4+1,3) to[out=90,in=0] (3.9+1,3.5) to[out=180,in=90] (3.4+1,3) to[out=270,in=90] (3.4+1,0);
\draw[thick,-<-] (3.9+1,2.8) to (3.9+1,0);
\draw[thick,->-] (6.1,2.8) to (6.1,0);
\draw[thick,->-] (7.5,2.8) to (7.5,0);
\draw[fill,pattern=north east lines,draw=none] (-8,-0.3) rectangle (8,0);
\draw[ultra thick] (-8,0) to (8,0);
\draw[fill,pattern=north east lines,ultra thick] (-7.3,3) circle (0.2);
\draw[fill,pattern=north east lines,ultra thick] (-5.9+0.5,3) circle (0.2);
\draw[fill,pattern=north east lines,ultra thick] (2.5+0.5,3) circle (0.2);
\draw[fill,pattern=north east lines,ultra thick] (3.9+1,3) circle (0.2);
\draw[fill,pattern=north east lines,ultra thick] (6.1,3) circle (0.2);
\draw[fill,pattern=north east lines,ultra thick] (7.5,3) circle (0.2);
\draw (-6.35,2.5) node {\dots};
\draw (-1.25,2.5) node {\dots};
\draw (3.95,2.5) node {\dots};
\draw (6.8,2.5) node {\dots};
\draw (0,-0.7) node {$D_{\beta_{j,1}}$};
\draw (-7.3,2.5) node {$D_{\beta_{j,2}}$};
\draw (-5.4,2.5) node {$D_{\beta_{j,r_j}}$};
\draw (-7.3,4) node {$Z_{j,2}$};
\draw (-5.4,4) node {$Z_{j,r_j}$};
\draw (-3.75,4) node {$X_{\gamma_{j,1}}$};
\draw (-2.5,4) node {$Y_{\gamma_{j,1}}$};
\draw (4-3.75,4) node {$X_{\gamma_{j,s_j}}$};
\draw (4-2.5,4) node {$Y_{\gamma_{j,s_j}}$};
\draw (3,4) node {$A_{\alpha_{j,r_j+1}}$};
\draw (4.9,4) node {$A_{\alpha_{j,\dots}}$};
\draw (3.68,2) node {$E_{\alpha_{j,r_j+1}}$};
\draw (5.38,2) node {$E_{\alpha_{j,\dots}}$};
\draw (6.6,2) node {$E_{\alpha_{j,\dots}}$};
\draw (8,2) node {$E_{\alpha_{j,b_j}}$};
\end{tikzpicture}
\caption{An example of a total order of half-edges adjacent to $p_j\in\partial\Sigma_j$ )left to right).}
\label{fig:OrderHalfEdges}
\end{figure}

Data of a ribbon graph can be encoded in $m$ ordered sets of half-edges adjacent to $p_j$
\begin{align}
  S_j=\{h_{j,1},\dots, h_{j,n_j}\},\qquad 1\leqslant j\leqslant m.
\label{eq:OrderedSetOfHalfEdgesDef}
\end{align}
Our ribbon graph is oriented, hence it will be convenient for us to label half-edges by generators of $\pi_1(\Sigma,p_1,\dots,p_m)$ and its inverses. We associate a generator for an outgoing edge and inverse of a generator for an ingoing edge. For example, ordered set of half-edges adjacent to $p_j$ on Figure \ref{fig:OrderHalfEdges} reads
\begin{align*}
S_j=\left\{z_{j,2}^{-1},z_{j,2},\dots, 
x_{\gamma_{j,1}}^{-1},y_{\gamma_{j,1}},x_{\gamma_{j,1}}, y_{\gamma_{j,1}}^{-1},\dots,
a_{\alpha_{j,r_j+1}},e_{\alpha_{j,r_j+1}},a_{\alpha_{j,r_j+1}}^{-1},\dots,
e_{\alpha_j,b_j}^{-1}\right\}.
\end{align*}

\subsection{The $r$-matrix approach}
\label{sec:HamiltonFlowsRMatrix}

Let $\mathfrak g$ be a Lie algebra of $G$ and  $r\in \mathfrak g\otimes\mathfrak g$  be a classical r-matrix, i.e. a a solution to the Yang-Baxter equation
\begin{align*}
[r_{12},r_{13}]+[r_{12},r_{23}]+[r_{13},r_{23}]=0.
\end{align*}

We will be interested in solutions corresponding to factorizable Lie bialgebras, i.e. classical $r$-matrices with the following symmetric part (with respect to the exchange of tensor components)
given by the Killing form:
\begin{align}
I=r+\sigma(r)=\sum_{J=1}^{\dim\mathfrak g}e_J\otimes e_J,
\label{eq:QuadraticCasimir}
\end{align}
Here $\{e_J\}$ is a basis in $\mathfrak g$ which is orthonormal with respect to the Killing form
and $\sigma(x\otimes y)=y\otimes x$. Such solutions are classified in \cite{BelavinDrinfeld'1982}.

The choice of free generators of $\pi_1(\Sigma,p_1,\dots,p_m)$ fixes a natural isomorphism $\mathrm{Hom}(\pi_1(\Sigma,p_1,\dots,p_m),G)\simeq G^{N}$ where $N=2g+b+m-2$.
For an ordered ribbon graph  V.~Fock and A.~Rosly  \cite{FockRosly'1993} introduced
a Poisson structure on this space with the Poisson tensor defined in terms of factorizable
classical $r$-matrices. Below we present the corresponding Poisson brackets between matrix element functions.

For two matrix valued functions $M_{ij}$ and $N_{ab}$ on a Poisson manifold
we will use the following notation
\[
\{M \otimes N\}_{(ia)(jb)}=\{M_{ij},N_{ab}\}
\]
for the matrix of Poisson brackets. Here the convention is $(M\otimes N)_{(ia)(jb)}=M_{ij}N_{ab}$.

\subsubsection*{Distinct arcs}
Basic building block of brackets corresponds to the case shown on Figure \ref{fig:BracketLOO} when the two oriented arcs share the source but have distinct targets different from their common source.
\begin{figure}
\begin{subfigure}{0.21\linewidth}
\begin{tikzpicture}
\draw[thick,-->-] (0,0.43) to[out=105,in=-45] (-1.2,2.2);
\draw[thick,-->-] (0,0.43) to[out=75,in=225] (1.2,2.2);
\draw[fill,pattern=north east lines,draw=none] (-1.5,0) to[out=30,in=150] (1.5,0) to (1.5,-0.3) to[out=150,in=30] (-1.5,-0.3) to (-1.5,0);
\draw[ultra thick] (-1.5,0) to[out=30,in=150] (1.5,0);
\fill (0,0.43) circle (0.12);
\fill[gray] (0,0.43) circle (0.08);
\draw (-1,1.1) node {$M_A$};
\draw (1,1.1) node {$M_B$};
\draw (0,-0.15) node {$p_j$};
\end{tikzpicture}
\caption{Outgoing-Outgoing}
\label{fig:BracketLOO}
\end{subfigure}
\qquad
\begin{subfigure}{0.21\linewidth}
\begin{tikzpicture}
\draw[thick,--<-] (0,0.43) to[out=105,in=-45] (-1.2,2.2);
\draw[thick,-->-] (0,0.43) to[out=75,in=225] (1.2,2.2);
\draw[fill,pattern=north east lines,draw=none] (-1.5,0) to[out=30,in=150] (1.5,0) to (1.5,-0.3) to[out=150,in=30] (-1.5,-0.3) to (-1.5,0);
\draw[ultra thick] (-1.5,0) to[out=30,in=150] (1.5,0);
\fill (0,0.43) circle (0.12);
\fill[gray] (0,0.43) circle (0.08);
\draw (-1,1.1) node {$N_A$};
\draw (1,1.1) node {$M_B$};
\draw (0,-0.15) node {$p_j$};
\end{tikzpicture}
\caption{Ingoing-Outgoing}
\label{fig:BracketLIO}
\end{subfigure}
\qquad
\begin{subfigure}{0.21\linewidth}
\begin{tikzpicture}
\draw[thick,-->-] (0,0.43) to[out=105,in=-45] (-1.2,2.2);
\draw[thick,--<-] (0,0.43) to[out=75,in=225] (1.2,2.2);
\draw[fill,pattern=north east lines,draw=none] (-1.5,0) to[out=30,in=150] (1.5,0) to (1.5,-0.3) to[out=150,in=30] (-1.5,-0.3) to (-1.5,0);
\draw[ultra thick] (-1.5,0) to[out=30,in=150] (1.5,0);
\fill (0,0.43) circle (0.12);
\fill[gray] (0,0.43) circle (0.08);
\draw (-1,1.1) node {$M_A$};
\draw (1,1.1) node {$N_B$};
\draw (0,-0.15) node {$p_j$};
\end{tikzpicture}
\caption{Outgoing-Ingoing}
\label{fig:BracketLOI}
\end{subfigure}
\qquad
\begin{subfigure}{0.21\linewidth}
\begin{tikzpicture}
\draw[thick,--<-] (0,0.43) to[out=105,in=-45] (-1.2,2.2);
\draw[thick,--<-] (0,0.43) to[out=75,in=225] (1.2,2.2);
\draw[fill,pattern=north east lines,draw=none] (-1.5,0) to[out=30,in=150] (1.5,0) to (1.5,-0.3) to[out=150,in=30] (-1.5,-0.3) to (-1.5,0);
\draw[ultra thick] (-1.5,0) to[out=30,in=150] (1.5,0);
\fill (0,0.43) circle (0.12);
\fill[gray] (0,0.43) circle (0.08);
\draw (-1,1.1) node {$N_A$};
\draw (1,1.1) node {$N_B$};
\draw (0,-0.15) node {$p_j$};
\end{tikzpicture}
\caption{Ingoing-Ingoing}
\label{fig:BracketLII}
\end{subfigure}
\caption{Relative position of half edges}
\end{figure}
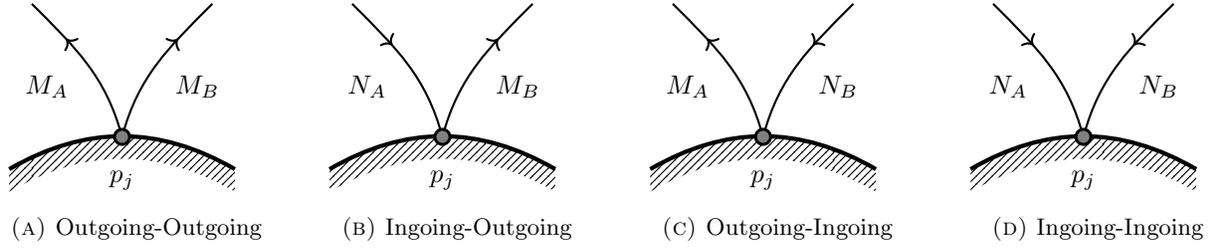
In this case the Fock and Rosly Poisson bracket between matrix elements\footnote{ Here and below by matrix elements  we always assume matrix elements in a finite dimensional representation.} of holonomies along these arcs reads
\begin{subequations}
\begin{align}
\{M_A\otimes M_B \}=(M_A\otimes M_B)\,r.
\label{eq:BracketLOO}
\end{align}

One can easily calculate the Poisson bracket between matrix elements of powers and products of matrices. For example, by Leibnitz identity applied to $M_AM_A^{-1}=\mathrm{Id}=M_BM_B^{-1}$ from (\ref{eq:BracketLOO}) we immediately obtain Poisson brackets for other three cases shown on Figures \ref{fig:BracketLIO}--\ref{fig:BracketLII}
\begin{align}
\{N_A\otimes M_B\}=&-(\mathrm{Id}\otimes M_B)\,r\,(N_A\otimes\mathrm{Id}),
\label{eq:BracketLIO}\\
\{M_A\otimes N_B\}=&-(M_A\otimes\mathrm{Id})\,r\,(\mathrm{Id}\otimes N_B),
\label{eq:BracketLOI}\\
\{N_A\otimes N_B\}=&\,r\,(N_A\otimes N_B),
\label{eq:BracketLII}
\end{align}
\label{eq:BracketsL}
\end{subequations}
where $N_A=M_A^{-1}$ and $N_B=M_B^{-1}$.

Similarly, brackets $\{M_B\otimes M_A\},\dots,\{N_B\otimes N_A\}$ can be obtained by skew-symmetry, and the answer can be expressed in terms of $r_{21}$ as follows
\begin{subequations}
\begin{align}
\{M_B\otimes M_A\}=&-(M_B\otimes M_A)\,r_{21},
\label{eq:BracketGOO}\\
\{M_B\otimes N_A\}=&(M_B\otimes\mathrm{Id})\,r_{21}\,(\mathrm{Id}\otimes N_A),
\label{eq:BracketGIO}\\
\{N_B\otimes M_A\}=&(\mathrm{Id}\otimes M_A)\,r_{21}\,(N_B\otimes\mathrm{Id}),
\label{eq:BracketGOI}\\
\{N_B\otimes N_A\}=&-r_{21}\,(N_B\otimes N_A).
\label{eq:BracketGII}
\end{align}
\label{eq:BracketsG}
\end{subequations}
Here $r_{21}=\sigma(r)$.

Bracket between matrix elements of a general distinct pair of arcs now can be computed by simply adding up to four terms of the form (\ref{eq:BracketsL}),(\ref{eq:BracketsG}) corresponding to brackets between pairs of half edges whenever they are adjacent to the same vertex.

\subsubsection*{Self-brackets}

There are only three cases of relative position of half edges of the same arc as shown on Figure \ref{fig:HalfEdgesSameArc}.
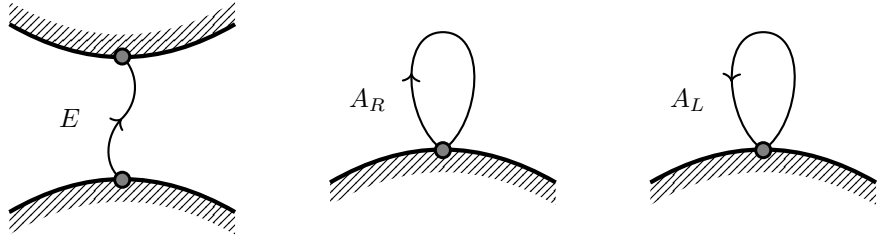
\begin{figure}
\begin{subfigure}{0.21\linewidth}
\begin{tikzpicture}
\draw[thick,->-] (0,0.43) to[out=145,in=225] (0,1.25) to[out=45,in=-45] (0,2.5-0.43);
\draw[fill,pattern=north east lines,draw=none] (-1.5,0) to[out=30,in=150] (1.5,0) to (1.5,-0.3) to[out=150,in=30] (-1.5,-0.3) to (-1.5,0);
\draw[ultra thick] (-1.5,0) to[out=30,in=150] (1.5,0);
\draw[fill,pattern=north east lines,draw=none] (-1.5,2.5) to[out=-30,in=-150] (1.5,2.5) to (1.5,2.8) to[out=-150,in=-30] (-1.5,2.8) to (-1.5,2.5);
\draw[ultra thick] (-1.5,2.5) to[out=-30,in=-150] (1.5,2.5);
\fill (0,0.43) circle (0.12);
\fill[gray] (0,0.43) circle (0.08);
\fill (0,2.5-0.43) circle (0.12);
\fill[gray] (0,2.5-0.43) circle (0.08);
\draw (-0.7,1.25) node {$E$};
\end{tikzpicture}
\end{subfigure}
\qquad
\begin{subfigure}{0.21\linewidth}
\begin{tikzpicture}
\draw[thick,->--] (0,0.43) to[out=145,in=180] (0,2) to[out=0,in=35] (0,0.43);
\draw[fill,pattern=north east lines,draw=none] (-1.5,0) to[out=30,in=150] (1.5,0) to (1.5,-0.3) to[out=150,in=30] (-1.5,-0.3) to (-1.5,0);
\draw[ultra thick] (-1.5,0) to[out=30,in=150] (1.5,0);
\fill (0,0.43) circle (0.12);
\fill[gray] (0,0.43) circle (0.08);
\draw (-1,1.1) node {$A_R$};
\end{tikzpicture}
\end{subfigure}
\qquad
\begin{subfigure}{0.21\linewidth}
\begin{tikzpicture}
\draw[thick,-<--] (0,0.43) to[out=145,in=180] (0,2) to[out=0,in=35] (0,0.43);
\draw[fill,pattern=north east lines,draw=none] (-1.5,0) to[out=30,in=150] (1.5,0) to (1.5,-0.3) to[out=150,in=30] (-1.5,-0.3) to (-1.5,0);
\draw[ultra thick] (-1.5,0) to[out=30,in=150] (1.5,0);
\fill (0,0.43) circle (0.12);
\fill[gray] (0,0.43) circle (0.08);
\draw (-1,1.1) node {$A_L$};
\end{tikzpicture}
\end{subfigure}
\caption{Relative poisition of half edges of the same arc}
\label{fig:HalfEdgesSameArc}
\end{figure}
Corresponding brackets are given by the following formulae
\begin{subequations}
\begin{align}
\{E\otimes E\}=&r_a (E\otimes E)+(E\otimes E) r_a,
\label{eq:SelfBracketOpen}\\
\{A_R\otimes A_R\}=&(A_R\otimes A_R) r_a+r_a (A_R\otimes A_R)+(\mathrm{Id}\otimes A_R)r_{21}(A_R\otimes\mathrm{Id})-(A_R\otimes\mathrm{Id})r_{12}(\mathrm{Id}\otimes A_R),
\label{eq:SelfBracketClosedRight}\\
\{A_L\otimes A_L\}=&r_a (A_L\otimes A_L)+(A_L\otimes A_L) r_a+(A_L\otimes\mathrm{Id})r_{21}(\mathrm{Id}\otimes A_L)-(\mathrm{Id}\otimes A_L)r_{12}(A_L\otimes\mathrm{Id}).
\label{eq:SelfBracketClosedLeft}
\end{align}
\end{subequations}
where $r_a=\frac12(r_{12}-r_{21})$ is the antisymmetric (with respect to the exchange of tensor components) part of r-matrix and we write $r_{12}$ for $r$.

\subsection{Hamiltonian flows generated by central functions} In this section we compute the evolution of matrix element functions on $\mathrm{Hom}(\pi_1(\Sigma,p_1,\dots,p_m),G)$ with respect to the  Hamiltonian flows generated by central functions on  holonomies around simple closed arcs $a_i$ or $y_j$ defined in Section \ref{sec:ChoiceOfGenerators}. We start by considering three basic cases shown on Figure \ref{fig:RelativeOrderOfHalfEdgesForPairs}.
\begin{figure}
\begin{subfigure}{0.21\linewidth}
\begin{tikzpicture}
\draw[thick,style={decoration={
  markings,
  mark=at position .85 with {\arrow{>}}},postaction={decorate}}] (0,0.43) to[out=90,in=215] (1.1,2.3);
\fill[white] (0.38,1.64) circle (0.12);
\draw[thick,--<-] (0,0.43) to[out=20,in=0] (0,2) to[out=180,in=160] (0,0.43);
\draw[fill,pattern=north east lines,draw=none] (-1.5,0) to[out=30,in=150] (1.5,0) to (1.5,-0.3) to[out=150,in=30] (-1.5,-0.3) to (-1.5,0);
\draw[ultra thick] (-1.5,0) to[out=30,in=150] (1.5,0);
\fill (0,0.43) circle (0.12);
\fill[gray] (0,0.43) circle (0.08);
\draw (1.2,1.7) node {$E$};
\draw (-0.8,1.5) node {$A$};
\end{tikzpicture}
\caption{Separating pair}
\end{subfigure}
\qquad\qquad
\begin{subfigure}{0.21\linewidth}
\begin{tikzpicture}
\draw[thick,-->-] (0,0.43) to[out=90,in=150] (1,2) to[out=-30,in=20] (0,0.43);
\draw[thick,-->-] (0,0.43) to[out=125,in=-45] (-1.3,2);
\draw[fill,pattern=north east lines,draw=none] (-1.5,0) to[out=30,in=150] (1.5,0) to (1.5,-0.3) to[out=150,in=30] (-1.5,-0.3) to (-1.5,0);
\draw[ultra thick] (-1.5,0) to[out=30,in=150] (1.5,0);
\fill (0,0.43) circle (0.12);
\fill[gray] (0,0.43) circle (0.08);
\draw (-0.5,1.8) node {$F$};
\draw (1.5,1.1) node {$A$};
\end{tikzpicture}
\caption{Disjoint left}
\end{subfigure}
\qquad\qquad
\begin{subfigure}{0.21\linewidth}
\begin{tikzpicture}
\draw[thick,--<-] (0,0.43) to[out=90,in=30] (-1,2) to[out=210,in=160] (0,0.43);
\draw[thick,-->-] (0,0.43) to[out=55,in=225] (1.3,2);
\draw[fill,pattern=north east lines,draw=none] (-1.5,0) to[out=30,in=150] (1.5,0) to (1.5,-0.3) to[out=150,in=30] (-1.5,-0.3) to (-1.5,0);
\draw[ultra thick] (-1.5,0) to[out=30,in=150] (1.5,0);
\fill (0,0.43) circle (0.12);
\fill[gray] (0,0.43) circle (0.08);
\draw (1.2,1.3) node {$G$};
\draw (-1.4,1) node {$A$};
\end{tikzpicture}
\caption{Disjoint right}
\end{subfigure}
\caption{Relative order of half-edges w.r.t to the closed arc}
\label{fig:RelativeOrderOfHalfEdgesForPairs}
\end{figure}
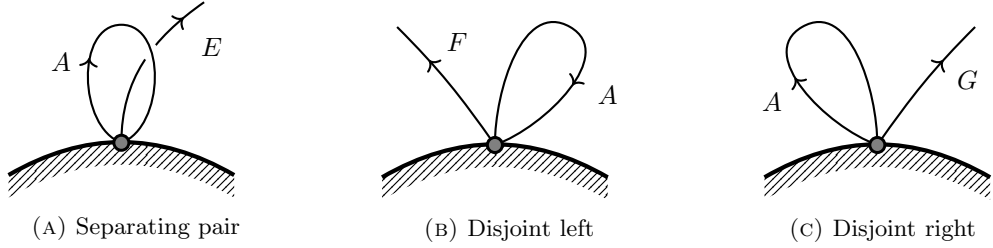

\begin{proposition}
Let $(a,e), (f,a), (a,g)$ be pairs of generators with relative order of half-edges as shown on Figure \ref{fig:RelativeOrderOfHalfEdgesForPairs},  $H: G\to \CC$ be a polynomial central function, $H(ghg^{-1})=H(h)$
and $\pi: G\to End(V)$ be a finite dimensional representation of $G$. Denote by $A,E,F,G\in \mathrm{Hom}(\pi_1(\Sigma,p_1,\dots,p_m),G)$
holonomies along the corresponding oriented arcs. We have
\begin{subequations}
\begin{align}
\{ H(A), \pi_V(E)\}=&\sum_J (e_J, \nabla H(A))\,\pi_V(Ee_J),
\label{eq:SeparatingDynamics}\\
\{H(A), \pi_V(F)\}=&\{H(A), \pi_V(G)\}=\{H(A), \pi_V(A)\}=0.
\label{eq:TrivialDynamicsProp}
\end{align}
\end{subequations}
where $e_J$ is a basis in $\mathfrak g$ orthonormal with respect to the Killing form and
\[
(\nabla H(A), X)=\frac{d}{dt} H(e^{tX}A)|_{t=0}
\]
Here $X\in\mathfrak g$ and $(.,.)$ is the Killing form.
\label{prop:AkActionHalfEdgesRMatrix}
\end{proposition}

\begin{proof}
It is enough to consider $H(A)=\mathrm{Tr}\,A^k$ where trace is taken over a finite dimensional representation (which can be different from $V$).
Poisson bracket between matrix elements of $A$ and $E$ can be calculated as a sum of two terms: first term of the form (\ref{eq:BracketLOO}) corresponds to the ordered pair of half edges $(a,e)$ adjacent to the same vertex; second term of the form (\ref{eq:BracketGOI}) corresponds to the ordered pair of half edges $(e,a^{-1})$. As a result, we have
\begin{align*}
\{A\otimes E\}=&(A\otimes E)\,r_{12}+(\mathrm{Id}\otimes E)\,r_{21}\,(A\otimes\mathrm{Id}).
\end{align*}
Hence, by Leibnitz Identity for all $k\in\mathbb N$ we get
\begin{align*}
\{A^k\otimes E\}=&\sum_{j=0}^{k-1}(A^j\otimes\mathrm{Id})\{A\otimes E\} (A^{k-j-1}\otimes\mathrm{Id})\\
=&\sum_{j=0}^{k-1}\Big((A^{j+1}\otimes E)\, r_{12}\, (A^{k-j-1}\otimes\mathrm{Id})+ (A^{j}\otimes E)\, r_{21}\, (A^{k-j}\otimes\mathrm{Id})\Big).
\end{align*}
Taking partial trace with respect to the first component we obtain (\ref{eq:SeparatingDynamics}) as
\begin{align*}
\{\mathrm{Tr}\,A^k,E\}=&\mathrm{Tr}_1\{A^k\otimes E\}\\
=&\sum_{j=0}^{k-1}\mathrm{Tr}_1\Big( (A^{j+1}\otimes E)\,r_{12}\,(A^{k-j-1}\otimes\mathrm{Id})\Big)+ \sum_{j=0}^{k-1}\mathrm{Tr}_1\Big((A^{j}\otimes E)\,r_{21}\,(A^{k-j}\otimes\mathrm{Id})\Big)\\
=&\sum_{j=0}^{k-1}\mathrm{Tr}_1\Big((A^k\otimes E)\,r_{12}\Big)+\sum_{j=0}^{k-1}\mathrm{Tr}_1\Big((A^k\otimes E)\,r_{21}\Big)=k\,\mathrm{Tr}_1\Big((A^k\otimes E)\,I\Big).
\end{align*}

Poisson bracket between matrix elements of $A$ and $G$ can be calculated as a sum of two terms: first term of the form (\ref{eq:BracketLOO}) corresponds to ordered pair of half edges $(a,g)$; second term of the form (\ref{eq:BracketLIO}) corresponds to ordered pair of half edges $(a^{-1},g)$. Hence, we have
\begin{align*}
\{A\otimes G\}=&(A\otimes G)\,r-(\mathrm{Id}\otimes G)\,r\,(A\otimes\mathrm{Id}).
\end{align*}
By Leibnitz Identity we get that for all $k\in\mathbb N$
\begin{align*}
\{A^k\otimes G\}=\sum_{j=0}^{k-1}\Big((A^{j+1}\otimes G)\,r\,(A^{k-j-1}\otimes\mathrm{Id})-(A^j\otimes G)\,r\,(A^{k-j}\otimes\mathrm{Id})\Big).
\end{align*}
Taking partial trace with respect to the first component we get
\begin{align*}
\{\mathrm{Tr}\,A^k,G\}=&\sum_{j=0}^{k-1}\mathrm{Tr}_1\Big((A^{j+1}\otimes G)\, r\, (A^{k-j-1}\otimes\mathrm{Id})\Big)-\sum_{j=0}^{k-1} \mathrm{Tr}_1\Big((A^{j}\otimes G)\,r\,(A^{k-j}\otimes\mathrm{Id})\Big)\\
=&k\,\mathrm{Tr}_1\Big((A^k\otimes G)\,r\Big)-k\,\mathrm{Tr}_1\Big((A^k\otimes G)\,r\Big)=0.
\end{align*}

Poisson bracket between matrix elements of $A$ and $F$ is determined by (\ref{eq:BracketGOO}) and (\ref{eq:BracketGOI}) and has the following form
\begin{align}
\{A\otimes F\}=&-(A\otimes F)\,r_{21}+(\mathrm{Id}\otimes F)\,r_{21}\,(A\otimes\mathrm{Id})
\label{eq:RbracketAF}
\end{align}
Note that $r_{21}$ is on the same side from $F$ in both terms of (\ref{eq:RbracketAF}), so as in the previous case we obtain for all $k\in\mathbb N$
\begin{align*}
\{\mathrm{Tr}\,A^k,F\}=0.
\end{align*}

Finally, Poisson bracket between matrix elements of $A$ is given by (\ref{eq:SelfBracketClosedRight})
\begin{align*}
\{A\otimes A\}=(A\otimes A)\,r_a+r_a\,(A\otimes A)+(\mathrm{Id}\otimes A)\,r_{21}\,(A\otimes\mathrm{Id})-(A\otimes\mathrm{Id})\,r_{12}\,(\mathrm{Id}\otimes A)
\end{align*}
By Leibnitz Identity we get for all $k\in\mathbb N$
\begin{align*}
\{A^k\otimes A\}=&\sum_{j=0}^{k-1}(A^j\otimes\mathrm{Id})\,\{A\otimes A\}\,(A^{k-j-1}\otimes\mathrm{Id})\\
=&\sum_{j=0}^{k-1}\Big((A^{j+1}\otimes A)\,r_a\,(A^{k-j-1}\otimes\mathrm{Id}) +(A^j\otimes\mathrm{Id})\,r_a\,(A^{k-j}\otimes A)\\
&+(A^j\otimes A)\,r_{21}\,(A^{k-j}\otimes\mathrm{Id})-(A^{j+1}\otimes\mathrm{Id})\,r_{12}\, (A^{k-j-1}\otimes A)\Big)
\end{align*}
Taking the partial trace with respect to the first component we have
\begin{equation}
\begin{aligned}
\{\mathrm{Tr}\,A^k,A\}=&k\,\mathrm{Tr}_1\Big((A^k\otimes A)\,r_a+r_a\,(A^k\otimes A)+(A^k\otimes A)\,r_{21}-r_{12}\,(A^k\otimes A)\Big)\\
=&\frac k2\,\mathrm{Tr}_1\,[A^k\otimes A,I].
\end{aligned}
\label{eq:SelfActionOriginal}
\end{equation}
From the invariance of $I$ and from the cyclic invariance of the trace we get
\[
(\mathrm{Tr}\otimes\mathrm{Id})((A^k\otimes A)I)=(\mathrm{Tr}\otimes\mathrm{Id})((A^{k-1}I(A\otimes A))=
(\mathrm{Tr}\otimes\mathrm{Id})(I(A^k\otimes A))
\]
This implies that the bracket $\{\mathrm{Tr}\,A^k,A\}$ vanishes.
\end{proof}

\begin{figure}
\begin{subfigure}{0.23\linewidth}
\centering
\begin{tikzpicture}
\draw[thick,-->-] (0,0.43) to[out=135,in=170] (0.5,2) to[out=-10,in=20] (0,0.43);
\fill[white] (0,1.75) circle (0.12);
\draw[thick,-->-] (0,0.43) to[out=180-135,in=180-170] (-0.5,2) to[out=180+10,in=180-20] (0,0.43);
\draw[fill,pattern=north east lines,draw=none] (-1.5,0) to[out=30,in=150] (1.5,0) to (1.5,-0.3) to[out=150,in=30] (-1.5,-0.3) to (-1.5,0);
\draw[ultra thick] (-1.5,0) to[out=30,in=150] (1.5,0);
\fill (0,0.43) circle (0.12);
\fill[gray] (0,0.43) circle (0.08);
\draw (-1.3,1.5) node {$X$};
\draw (1.3,1.5) node {$Y$};
\end{tikzpicture}
\caption{Genus one pair}
\label{fig:GenusOnePair}
\end{subfigure}
\qquad\qquad
\begin{subfigure}{0.23\linewidth}
\centering
\begin{tikzpicture}
\draw[thick,-->-] (0,0.43) to[out=80,in=150] (1.2,2) to[out=-30,in=20] (0,0.43);
\draw[thick,--<-] (0,0.43) to[out=180-80,in=180-150] (-1.2,2) to[out=180+30,in=180-20] (0,0.43);
\draw[fill,pattern=north east lines,draw=none] (-1.5,0) to[out=30,in=150] (1.5,0) to (1.5,-0.3) to[out=150,in=30] (-1.5,-0.3) to (-1.5,0);
\draw[ultra thick] (-1.5,0) to[out=30,in=150] (1.5,0);
\fill (0,0.43) circle (0.12);
\fill[gray] (0,0.43) circle (0.08);
\draw (-1.5,1.1) node {$Y$};
\draw (1.5,1.1) node {$A$};
\end{tikzpicture}
\caption{Genus zero pair}
\label{eq:fig:GenusZeroPair}
\end{subfigure}
\caption{Relative order of half edges for pairs of closed arcs.}
\label{fig:ClosedPairs}
\end{figure}
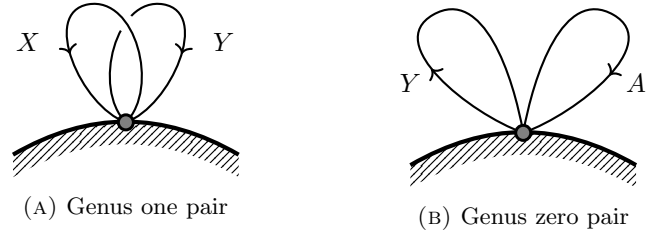

\begin{corollary}
Let $(y,x)$ and $(y,a)$ be pairs of generators with relative order of half edges as shown on Figure \ref{fig:ClosedPairs} and $H$ and $\pi_V$ as in the previous proposition. Denote by $A,X,Y\in \mathrm{Hom}(\pi_1(\Sigma,p_1,\dots,p_m),G)$ the  holonomies along the corresponding closed oriented arcs. We have
\begin{subequations}
\begin{align}
\{H(Y), \pi_V(X)\}=&\sum_J (e_J, \nabla H(A))\,\pi_V(Xe_J),
\label{eq:GenusOneAction}\\
\{H(Y), \pi_V(A)\}=&\{H(Y), \pi_V(Y)\}=0.
\label{eq:GenusZeroAction}
\end{align}
\end{subequations}
\end{corollary}
\begin{proof}
Indeed, the Poisson bracket between $\mathrm{Tr}\,Y^k$ and $X_{ij}$ gets two contributions of the form (\ref{eq:SeparatingDynamics}) and (\ref{eq:TrivialDynamicsProp}) which correspond to the outgoing and the ingoing half edges of $X$ respectively. As a result, we obtain (\ref{eq:GenusOneAction}).

On the other hand, the Poisson bracket between $\mathrm{Tr}\,Y^k$ and $A_{ij}$ is trivial due to (\ref{eq:TrivialDynamicsProp}) because none of the half edges of $A$ appear in between half edges of $Y$. Similar logic applies to the bracket between $\mathrm{Tr}\,A^k$ and $Y_{ij}$.
\end{proof}

Note that the Hamilton flow given by $H(Y)$ preserves matrix element functions of the group commutator $\pi_V(XYX^{-1}Y^{-1})$.

\begin{theorem}
Let $\Sigma$ be an oriented surface with $b>0$ boundary components and $C=C_1\sqcup\dots\sqcup C_k$ be a disjoint union of pairwise nonintersecting simple closed curves in $\Sigma$ as in Theorem \ref{theorem:SuperIntSystemsCycles}, $H: G\to \CC$ be a central function and $\pi_V: G\to End(V)$ be a finite dimensional representation. Fix a choice of marked points $p_1,\dots,p_m\in\partial\Sigma$ and generators of $\pi_1(\Sigma,p_1,\dots,p_m)$ as in Section \ref{sec:ChoiceOfGenerators} and denote by $A_i,E_i,X_l,Y_l,Z_o \in \mathrm{Hom}(\pi_1(\Sigma,p_1,\dots,p_m),G)$ the holonomies along arcs $a_i,e_i,x_l,y_l,z_o$, then
\begin{align*}
\{H(A_i),\pi_V(E_j)\}=&\left\{\begin{array}{ll}
\sum_J (e_J, \nabla H(A_i))\,\pi_V(E_je_J),&i=j,\\[5pt]
0,&i\neq j,
\end{array}\right.\\[5pt]
\{H(A_i),\pi_V(X_l)\}=&\{H(A_i),\pi_V(Y_l)\}=\{H(A_i),\pi_V(Z_o)\}=0,\\[10pt]
\{H(Y_l),\pi_V(X_s)\}=&\left\{\begin{array}{ll}
\sum_J (e_J, \nabla H(Y_l))\,\pi_V(X_se_J),&l=s,\\[5pt]
0,&l\neq s,
\end{array}\right.\\[5pt]
\{H(Y_l),\pi_V(A_i)\}=&\{H(Y_l),\pi_V(E_i)\}=\{H(Y_l),\pi_V(Z_o)\}=0.
\end{align*}
\label{theorem:HamiltonFlows}
\end{theorem}

Particular cases of Theorem \ref{theorem:HamiltonFlows} already appeared in the literature. For example, these brackets were computed in \cite{FockRosly'1999} for $G=SL(N,\mathbb C)$ and a torus with one boundary component.

\begin{corollary}
Fix $i,1\leqslant i\leqslant r$ and a central function $H: G\to \CC$. Then the Hamiltonian flow generated by the function $H$ evaluated on a cycle $C_i\in\Sigma$ is given by the following expressions
\begin{itemize}
\item When $C_i\in\partial\Sigma_a\cap\partial\Sigma_b$ separates a pair of distinct connected components of $\Sigma\backslash C$ we have
    \begin{align*}
    E_j(t)=&\left\{\begin{array}{ll}
    E_j\,\mathrm{exp}\left(t\nabla H(A_i)\right),&i=j,\\[5pt]
    E_j,&i\neq j,
    \end{array}\right.\\[10pt]
    X_l(t)=&X_l,\qquad
    Y_l(t)
    =Y_l,\\[5pt]
    A_j(t)=&A_j,\qquad
    Z_o(t)=Z_o.
    \end{align*}
\item Similarly, when $C_l$ is a nonseparating cycle which belongs to the boundary of a single connected component of $\Sigma\backslash C$ the Hamiltonian flow generated by $H$ evaluated at the holonomy along $C_l$
is given by
    \begin{align*}
    X_s(t)=&\left\{\begin{array}{ll}
    X_s\,\mathrm{exp}\left(t\nabla H(Y_l)\right ),&l=s,\\[5pt]
    X_s,&l\neq s,
    \end{array}\right.\\[10pt]
    A_j(t)=&A_j,\qquad
    E_j(t)=E_j\,\\[5pt]
    Y_l(t)=&Y_l,\qquad
    Z_o(t)=Z_o.
    \end{align*}
\end{itemize}
\end{corollary}

\begin{remark}
Let $\psi_t$ is the evolution on the algebra of functions on $\mathrm{Hom}(\pi_1(\Sigma,p_1,\dots,p_m),G)$
generated by $H$, i.e.
\[
\psi^H_t(F)=F+\sum_{n\geq 1} \frac{t^n}{n!} \{H,\{H, \dots, \{H,F\}\dots\}\}.
\]
If the functions $H_1$ and $H_2$ Poisson commute then $\psi^{H_1}_{t_1}\psi^{H_2}_{t_2}=\psi^{H_2}_{t_2}\psi^{H_1}_{t_1}$. If $H_1,\dots, H_k$ is a complete set of independent functions on holonomies along $C$, their joint flow lines
generate ange variables  on their level surfaces.
\end{remark}

\section{Particular Cases}

Here we will describe Poisson algebras $\mathcal O[\mathcal M_\Sigma^G]$ and corresponding superintegrable systems explicitly in few low dimensional
examples when $\Sigma$ is a torus with one or two boundary components. Writing out examples inevitably brings us to the problem of an explicit description of the invariant ring $\mathcal O[\mathcal M_\Sigma^G]$. For a simple linear algebraic group $G$ such invariant rings are always finitely generated, yet the number of generators grows much faster than the dimension of the corresponding moduli space. Here we describe subalgebras $Z_{\partial\Sigma}\subset B_C\subset J_{\Sigma\backslash C}\subset \mathcal O[\mathcal M_\Sigma^G]$ and their Poisson brackets internally, in terms of generators of the invariant ring $\mathcal O[\mathcal M_\Sigma^G]$ for some simple examples.

\subsection{Torus with one boundary component, $G=SL(2,\mathbb C)$}

Let $\Sigma_{1,1}$ denote the torus with one boundary component. Choose a marked point $p_1\in\partial\Sigma_{1,1}$. The fundamental group $\pi_1(\Sigma_{1,1},p_1)$ is freely generated by a pair of arcs $X,Y$ along the equator and meridian of the torus. Let $\Gamma$ be the ribbon graph associated to such choice of generators. The data of $\Gamma$ is encoded by an ordered set
\begin{align}
S_1=\{X^{-1},Y,X,Y^{-1}\}
\label{eq:TorusHalfEdges}
\end{align}
of half-edges adjacent to $p_1$ as shown on Figure \ref{fig:GenusOnePair}. In (\ref{eq:TorusHalfEdges}) we have labelled half-edges by generators $X,Y$ of $\pi_1(\Sigma_{1,1},p_1)$ and their inverses $X^{-1},Y^{-1}$ according to the convention introduced in the last paragraph of Section \ref{sec:ChoiceOfGenerators}. Namely, we label an outgoing half-edge by the first power of the corresponding generator, while we label an ingoing edge by the inverse of the corresponding generator.

\subsubsection*{Moduli Space} The moduli space of flat $SL(2,\mathbb C)$-connections on a once punctured torus has dimension $n=3$. We will use notation
\[
\tau_A=\mathrm{Tr}(A)
\]
for the trace of a matrix $A$.

The coordinate ring of the moduli space  in this example is a free commutative algebra with three generators
\begin{align*}
    \mathcal O[\mathcal M_{\Sigma_{1,1}}^{SL(2,\mathbb C)}]=\mathbb C[\tau_X,\tau_Y,\tau_{XY}]
\end{align*}
where $X,Y\in SL(2,\CC)$ are elements representing standard $a$ and $b$ cycles.
The  Poisson bracket between these coordinate functions are:
\begin{equation*}
\begin{aligned}
    \{\tau_X,\tau_Y\}=&-\tau_{XY}+\frac12\tau_X\tau_Y,\\
    \{\tau_X,\tau_{XY}\}=&\tau_Y-\frac12\tau_X\tau_{XY},\\
    \{\tau_Y,\tau_{XY}\}=&-\tau_X+\frac12\tau_Y\tau_{XY}.
\end{aligned}
\end{equation*}
This is a rank $2$ Poisson structure. The center of $\mathcal O[\mathcal M_{\Sigma_{1,1}}^{SL(2,\mathbb C)}]$ is generated by a single Casimir function
\begin{align}
    Z=\mathbb C[z],\qquad z=\tau_X^2+\tau_Y^2+\tau_{XY}^2-\tau_X\tau_Y\tau_{XY}=\mathrm{Tr}\,(XYX^{-1}Y^{-1})
    \label{eq:TorusB1SL2SuperIntegrable}
\end{align}

\subsubsection*{Superintegrable system} Choose cycle $C$ defining the system to be $Y$ and the Hamiltonian $H=\tau_Y$. From (\ref{eq:subgroupC}) we obtain a chain of subgroups
\begin{align*}
\pi_1(C)=\langle XYX^{-1}Y^{-1}\rangle\quad\subset\quad \pi_1(\Sigma_{1,1}\backslash C,p)=\langle Y,XYX^{-1}\rangle\quad\subset\quad\pi_1(\Sigma_{1,1},p)=\langle X,Y\rangle.
\end{align*}
Here and below we write $\langle a_1,\dots, a_n\rangle$ for a free group generated by $a_1,\dots, a_n$.

The algebra of first integrals is then generated by three elements, two of which coincide
\begin{align*}
J=\mathcal O[\mathrm{Hom}(\pi_1(\Sigma_{1,1}\backslash C,p))]=\mathbb C[\tau_{XYX^{-1}},\tau_{Y^{-1}},\tau_{XYX^{-1}Y^{-1}}]=\mathbb C[z,H].
\end{align*}
As a result, we obtain the following chain of subalgebras of the coordinate ring
\begin{align}
    Z\quad\subset\quad B=\mathbb C[z,H]\quad\subset\quad J=\mathbb C[z,H]\quad\subset\quad\mathcal O[\mathcal M_{\Sigma_{1,1}}^{SL(2,\mathbb C)}].
\label{eq:SuperIntSystemG1B1SL2}
\end{align}
In this case the algebra of Hamiltonians coincide with teh algebra of first integrals and therefore the system is Liouville  integrable.


\subsubsection*{Mapping Class Group action}

Consider the mapping class group $Mod(\Sigma_{1,1,0})$ (relative to the boundary) of the torus with one boundary component and no punctures. It contains two left Dehn twists along the $X$ and $Y$ cycles which satisfy the braid relation and are acting on our generators as follows:
\begin{align*}
D_X:\left\{\begin{array}{l}
X\mapsto X,\\
Y\mapsto YX^{-1},
\end{array}\right.
\qquad
D_Y:\left\{\begin{array}{l}
X\mapsto XY,\\
Y\mapsto Y,
\end{array}\right.
\qquad
D_XD_YD_X=D_YD_XD_Y=\left\{\begin{array}{l}
X\mapsto XYX^{-1},\\
Y\mapsto X^{-1}.
\end{array}\right.
\end{align*}
This defines a pair of Poisson automorphisms $\mathcal D_X,\mathcal D_Y$ of the coordinate ring $\mathcal O[\mathcal M_{\Sigma_{1,1}}^{SL(2,\mathbb C)}]$
\begin{align}
\mathcal D_X=\left\{\begin{array}{l}
\tau_X\mapsto\tau_X,\\
\tau_Y\mapsto-\tau_{XY}+\tau_X\tau_Y,\\
\tau_{XY}\mapsto\tau_Y,
\end{array}\right.
\qquad
\mathcal D_Y=\left\{\begin{array}{l}
\tau_X\mapsto\tau_{XY},\\
\tau_Y\mapsto\tau_Y,\\
\tau_{XY}\mapsto-\tau_X+\tau_Y\tau_{XY},
\end{array}\right.
\label{eq:DehnTwistsG1P1SL2}
\end{align}

Note that action of $Mod(\Sigma_{1,1,0})$ on the character variety factors through the action of the Mapping Class Group of a torus with one puncture $Mod(\Sigma_{1,0,1})\simeq SL(2,\mathbb Z)$.

Poisson automorphisms (\ref{eq:DehnTwistsG1P1SL2}) define a family of isomorphic integrable systems associated to nonseparating cycles on $\Sigma_{1,1}$.

\subsection{Torus with two boundary components, $G=SL(2,\mathbb C)$.}

The coordinate ring of the moduli space has dimension $n=6$ and the coordinate ring is generated by 7 polynomials subject to the single relation
\begin{align*}
    \mathcal O[\mathcal M_{\Sigma_{1,2}}^{SL(2,\mathbb C)}]=\frac{\mathbb C[\tau_X,\tau_Y,\tau_Z,\tau_{XY},\tau_{XZ},\tau_{YZ},\tau_{XYZ}]}{(\tau_{XYZ}^2 +\mu_1\,\tau_{XYZ}+\mu_0)}.
\end{align*}
Here $\tau_{A}$ stands for $\mathrm{Tr}\,(A)$ and
\begin{align*}
\mu_1=&\tau _Z \tau _{XY}+\tau _Y \tau _{XZ}+\tau _X \tau _{YZ}-\tau _X \tau _Y \tau _Z,\\
\mu_0=&\tau _X \tau _Y \tau _{XY}+\tau _Y \tau _Z \tau _{YZ}+\tau _X \tau _Z \tau _{XZ}-\tau _{XY} \tau _{XZ} \tau _{YZ}-\tau _{XY}^2-\tau _{XZ}^2-\tau _{YZ}^2-\tau _X^2-\tau _Y^2-\tau _Z^2+4.
\end{align*}
\begin{table}
\begin{flushleft}
$
\begin{array}{c|ccc}
&\tau_X&\tau_Y&\tau_{XY}\\
\hline
\tau_X&0&*&*\\
\tau_Y& \frac{\tau _X \tau _Y}{2}-\tau _{XY}&0&* \\
\tau_{XY}& \tau _Y-\frac{1}{2} \tau _X \tau _{XY} & \frac{1}{2} \tau _Y \tau _{XY}-\tau _X&0 \\
\tau_{XZ}& 0 & \tau _{XYZ}-\frac{1}{2} \tau _Y \tau _{XZ} & -\tau _Y \tau _Z+\frac{1}{2} \tau _{XY} \tau _{XZ}+\tau _{YZ} \\
\tau_{YZ}& \tau _X \tau _Y \tau _Z-\tau _{XY} \tau _Z-\tau _Y \tau _{XZ}-\frac{1}{2} \tau _X \tau _{YZ}+\tau _{XYZ} & 0 & \tau _X \tau _Z-\tau _{XZ}-\frac{1}{2} \tau _{XY} \tau _{YZ} \\
\tau_{XYZ}& \tau _Y \tau _Z-\tau _{XY} \tau _{XZ}-\tau _{YZ}+\frac{1}{2} \tau _X \tau _{XYZ} & \frac{1}{2} \tau _Y \tau _{XYZ}-\tau _{XZ} & 0
\end{array}
$
\end{flushleft}
\vspace{0.4cm}

\begin{flushleft}
$
\begin{array}{c|cc}
&\tau_{XZ}&\tau_{YZ}\\
\hline
\tau_{YZ}& \tau _{XY}+\frac{1}{2} \tau _{XZ} \tau _{YZ}-\tau _Z \tau _{XYZ}&0\\
\tau_{XYZ}& \tau _Y-\frac{1}{2} \tau _{XZ} \tau _{XYZ}&-\tau _X+\tau _Y \tau _{XY}-\tau _Y \tau _Z \tau _{XYZ}+\frac{1}{2} \tau _{YZ} \tau _{XYZ}+\tau _Z \tau _{XZ}
\end{array}
$
\end{flushleft}
\vspace{0.4cm}

\caption{Brackets between generators of $\mathcal O[\mathcal M_{\Sigma_{1,2}}^{SL(2)}]$}
\label{tab:G1B2SL2Brackets}
\end{table}

Poisson brackets between generators are summarized in Table \ref{tab:G1B2SL2Brackets}, where we have omitted generator $\tau_Z$ which belongs to the Poisson center of $\mathcal O[\mathcal M_{\Sigma_{1,2}}^{SL(2,\mathbb C)}]$.

The Poisson center $Z$ of the coordinate ring has Krull dimension $\dim Z=2$ and has two algebraically independent Casimir elements corresponding to traces of monodromies around each of the boundary components
\begin{align*}
z_1=&\tau_Z,&z_2=\mathrm{Tr}\,(XYX^{-1}Y^{-1}Z)=-\tau _Y \tau _{XY} \tau _{XZ}+\tau _{XY} \tau _{XYZ}+\tau _X \tau _{XZ}-\tau _Y \tau _{YZ}+\tau _Y^2 \tau _Z-\tau _Z
\end{align*}

\subsubsection{A superintegrable system associated to nonseparating cycle.} Consider a superintegrable system given by a nonseparating cycle homotopic to $Y$, from (\ref{eq:subgroupC}) we get a chain of subgroups
    \begin{align}
    \pi_1(C)=\langle Y\rangle\quad\subset\quad\pi_1(\Sigma_{1,2},p\,|C)=\langle Y,XYX^{-1},Z\rangle\quad\subset\quad\pi_1(\Sigma_{1,2},p)=\langle X,Y,Z\rangle
    \label{eq:SubgroupsG1B2SL2}
    \end{align}
    Coordinate ring of the character variety of the smallest subgroup is generated by a single hamiltonian $H=\Tau_Y$. At the same time, the coordinate ring of the character variety of the middle subgroup in (\ref{eq:SubgroupsG1B2SL2}) has the following presentation\footnote{Finding a generating set together with transcendence basis for the corresponding field of fractions is always straightforward in the case of all rings involved in our paper. However, it is much more complicated to find a complete set of relations outside of the most elementary examples.

    In this particular case, one can recall that $J$ is isomorphic to a quotient ring of $\mathcal O[M_{\Sigma_{0,4}}^{SL(2)}]$ modulo the relation that traces of monodromies around the fixed pair of boundary components are equal to each other. This allows one to resolve one of the generators in $\mathcal O[M_{\Sigma_{0,4}}^{SL(2)}]$ and obtain presentation (\ref{eq:JPresentationG1B2SL2}).}
    \begin{equation}
    \begin{aligned}
    J=&\mathcal O[\mathrm{Hom}(\pi_1(\Sigma_{1,1},p\,|\,C),SL(2))]^{SL(2)}\\[5pt]
    =&\frac{\mathbb C[\Tau_Y,\Tau_Z,\Tau_{YZ},\Tau_{XYX^{-1}Z},\Tau_{YXYX^{-1}}, \Tau_{XYX^{-1}YZ}]}{\left(\Tau_{XYX^{-1}YZ}^2 +\lambda_1\Tau_{XYX^{-1}YZ}+\lambda_0\right)},
    \end{aligned}
    \label{eq:JPresentationG1B2SL2}
    \end{equation}
    where
    \begin{align*}
    \lambda_1=&\Tau _Y \Tau _{XYX^{-1}Z}+\Tau _Z \Tau _{YXYX^{-1}}+\Tau _Y \Tau _{YZ}-\Tau _Y^2 \Tau _Z,\\
    \lambda_0=&\Tau _Y^2 \Tau _{YXYX^{-1}}+\Tau _Y \Tau _Z \Tau _{XYX^{-1}Z}-\Tau _{YZ} \Tau _{YXYX^{-1}} \Tau _{XYX^{-1}Z}+\Tau _Y \Tau _Z \Tau _{YZ}\\
    &-\Tau _{YXYX^{-1}}^2-\Tau _{XYX^{-1}Z}^2-\Tau _{YZ}^2-2 \Tau _Y^2-\Tau _Z^2.
    \end{align*}
    Subalgebra $J\subset\mathcal O[\mathcal M_{\Sigma_{1,2}}^{SL(2)}]$ is closed under the Poisson bracket. Explicitly, all nonzero brackets between generators are given by
    \begin{subequations}
    \begin{align}
    \{\Tau_{XYX^{-1}Z},\Tau_{YXYX^{-1}}\}=&\Tau_Y \Tau_Z-2 \Tau_{YZ}-\Tau_{XYX^{-1}Z} \Tau_{YXYX^{-1}}+\Tau_Y \Tau_{XYX^{-1}YZ}\\
    \{\Tau_{YZ},\Tau_{YXYX^{-1}}\}=&\Tau_Z \Tau_Y^3-\Tau_{YZ} \Tau_Y^2-\Tau_{XYX^{-1}Z} \Tau_Y^2-\Tau_Z \Tau_Y-\Tau_Z \Tau_{YXYX^{-1}} \Tau_Y\label{eq:G1B2SL2JBracketPart2}\\
    &+\Tau_{XYX^{-1}YZ} \Tau_Y+2 \Tau_{XYX^{-1}Z}+\Tau_{YZ} \Tau_{YXYX^{-1}}\nonumber\\
    \{\Tau_{YZ},\Tau_{XYX^{-1}Z}\}=&\Tau_Y^2-\Tau_{YZ} \Tau_{XYX^{-1}Z}-2 \Tau_{YXYX^{-1}}+\Tau_Z \Tau_{XYX^{-1}YZ}\\
    \{\Tau_{XYX^{-1}YZ},\Tau_{YXYX^{-1}}\}=&\Tau_Z \Tau_Y^2+\Tau_{XYX^{-1}YZ} \Tau_Y^2-2 \Tau_{YZ} \Tau_Y-\Tau_{XYX^{-1}Z} \Tau_{YXYX^{-1}} \Tau_Y\\
    \{\Tau_{XYX^{-1}YZ},\Tau_{YZ}\}=&-\Tau_Y^3+\Tau_{YZ} \Tau_{XYX^{-1}Z} \Tau_Y+2 \Tau_{YXYX^{-1}} \Tau_Y-\Tau_Z \Tau_{XYX^{-1}YZ} \Tau_Y
    \end{align}
    \end{subequations}

    The algebra of first integrals (\ref{eq:JPresentationG1B2SL2}) has Krull dimension 5. One can choose its maximal algebraically independent subset to contain two Casimir functions
\begin{align*}
z_1=\Tau_Z,\qquad z_2=\Tau_Y\Tau_{XYX^{-1}Z}-\Tau_{XYX^{-1}YZ},
\end{align*}
one Hamiltonian
\[
H=\Tau_Y,
\]
and  two more first integrals
\begin{align*}
g_1=\Tau_{YZ}=\tau_{YZ},\qquad g_2=\Tau_{XYX^{-1}Y}=\tau_{XY}^2+\tau_X\tau_Y\tau_{XY}-\tau_X^2+2.
\end{align*}

\subsubsection{Separating cycle} Now consider a superintegrable system given by a separating cycle homotopic to $XYX^{-1}Y^{-1}$. As in the previous case we have two Casimir elements $z_1$ and $z_2$. Their generic level sets form a 2-parametric family of 4-dimensional symplectic manifolds.  Choose the Hamiltonian on these phase spaces as
\begin{align*}
H=\mathrm{Tr}\,(XYX^{-1}Y^{-1})=\tau_X^2+\tau_Y^2+\tau_{XY}^2-\tau_X\tau_Y\tau_{XY}-2.
\end{align*}
From (\ref{eq:SeparatingPropBJDef}) we have the following algebra of first integrals
\begin{align*}
J=&J_L\otimes_B J_R,& B=&\mathbb C[H]\simeq \mathcal O[SL(2)]^{SL(2)},\\
J_L=&\mathbb C[\tau_X,\tau_Y,\tau_{XY}]\simeq \mathcal O[\mathcal M_{\Sigma_{1,1}}^{SL(2)}],&
J_R=&\mathbb C[z_1,z_2,H]\simeq \mathcal O[\mathcal M_{\Sigma_{0,3}}^{SL(2)}].
\end{align*}
Here $J_R$ is Poisson commutative, while $J_L$ is equipped with a Poisson bracket of generic rank 2 (see first three rows of Table \ref{tab:G1B2SL2Brackets}). Algebra $J$ has Krull dimension 5. An example of maximal algebraically independent subset is $z_1,z_2, H, g_1=\tau_X,g_2=\tau_Y$, i.e. generators of the Poisson center, the Hamiltonian and two additional first integrals.

\subsection{Torus with one boundary component, $G=SL(3,\mathbb C)$.}
The moduli space of flat $SL(3,\mathbb C)$-con\-nec\-tions on a once punctured torus has dimension $n=8$, the coordinate ring is generated by 9 polynomials subject to a single relation
\begin{align}
    \mathcal O[\mathcal M_{\Sigma_{1,1}}^{SL(3)}]=\frac{\mathbb C[\tau_X,\tau_Y,\tau_{X^{-1}},\tau_{Y^{-1}},\tau_{YX},\tau_{YX^{-1}},\tau_{Y^{-1}X}, \tau_{Y^{-1}X^{-1}},\tau_{XYX^{-1}Y^{-1}}]}{\big(\tau_{XYX^{-1}Y^{-1}}^2 +\zeta_1\,\tau_{XYX^{-1}Y^{-1}}+\zeta_0)},
\label{eq:CoordinateRingG1B1SL3}
\end{align}
where $\zeta_1$ and $\zeta_0$ are polynomials in $\tau_X,\dots\tau_{X^{-1}Y^{-1}}$:
\begin{equation}
\begin{aligned}
\zeta_1=&\mathrm{Tr}\,(YXY^{-1}X^{-1})\\
=&\tau _X \tau _{X^{-1}}+\tau _X \tau _Y \tau _{X^{-1}} \tau _{Y^{-1}}-\tau _{X^{-1}} \tau _{Y^{-1}} \tau _{YX}-\tau _Y \tau _{X^{-1}} \tau _{Y^{-1}X}-\tau _X \tau _{Y^{-1}} \tau _{YX^{-1}}\\
&+\tau _{YX^{-1}} \tau _{Y^{-1}X}-\tau _X \tau _Y \tau _{Y^{-1}X^{-1}}+\tau _{YX} \tau _{Y^{-1}X^{-1}}-\tau _{XYX^{-1}Y^{-1}}+\tau _Y \tau _{Y^{-1}}-3.
\label{eq:G1B1SL3Zeta1}
\end{aligned}
\end{equation}
and
\begin{align*}
\zeta_0=&\tau _Y \tau _{Y^{-1}} \tau _X^3-\tau _X^3-\tau _Y \tau _{X^{-1}}^2 \tau _{Y^{-1}} \tau _X^2-\tau _Y^2 \tau _{YX} \tau _X^2-\tau _{Y^{-1}} \tau _{YX} \tau _X^2+\tau _{X^{-1}} \tau _{Y^{-1}} \tau _{YX^{-1}} \tau _X^2\\
&-\tau _{Y^{-1}}^2 \tau _{Y^{-1}X} \tau _X^2-\tau _Y \tau _{Y^{-1}X} \tau _X^2+\tau _Y \tau _{X^{-1}} \tau _{Y^{-1}X^{-1}} \tau _X^2-\tau _{YX^{-1}} \tau _{Y^{-1}X^{-1}} \tau _X^2+\tau _{X^{-1}} \tau _{Y^{-1}}^3 \tau _X\\
&-\tau _Y^2 \tau _{X^{-1}} \tau _{Y^{-1}}^2 \tau _X+2 \tau _Y \tau _{YX}^2 \tau _X-\tau _{YX} \tau _{YX^{-1}}^2 \tau _X+2 \tau _{Y^{-1}} \tau _{Y^{-1}X}^2 \tau _X-\tau _{Y^{-1}X} \tau _{Y^{-1}X^{-1}}^2 \tau _X\\
&+\tau _Y^3 \tau _{X^{-1}} \tau _X+6 \tau _{X^{-1}} \tau _X-\tau _Y \tau _{X^{-1}} \tau _{Y^{-1}} \tau _X+\tau _{X^{-1}}^2 \tau _{Y^{-1}} \tau _{YX} \tau _X-\tau _Y^2 \tau _{YX^{-1}} \tau _X\\
&+\tau _Y \tau _{Y^{-1}}^2 \tau _{YX^{-1}} \tau _X-3 \tau _{Y^{-1}} \tau _{YX^{-1}} \tau _X+\tau _Y \tau _{X^{-1}} \tau _{YX} \tau _{YX^{-1}} \tau _X+\tau _Y \tau _{X^{-1}}^2 \tau _{Y^{-1}X} \tau _X\\
&+\tau _Y \tau _{Y^{-1}} \tau _{YX} \tau _{Y^{-1}X} \tau _X+3 \tau _{YX} \tau _{Y^{-1}X} \tau _X-\tau _{X^{-1}} \tau _{YX^{-1}} \tau _{Y^{-1}X} \tau _X-\tau _{Y^{-1}}^2 \tau _{Y^{-1}X^{-1}} \tau _X\\
&-3 \tau _Y \tau _{Y^{-1}X^{-1}} \tau _X+\tau _Y^2 \tau _{Y^{-1}} \tau _{Y^{-1}X^{-1}} \tau _X-\tau _{X^{-1}} \tau _{YX} \tau _{Y^{-1}X^{-1}} \tau _X+\tau _{X^{-1}} \tau _{Y^{-1}} \tau _{Y^{-1}X} \tau _{Y^{-1}X^{-1}} \tau _X\\
&-\tau _Y^3-\tau _{X^{-1}}^3-\tau _{Y^{-1}}^3-\tau _{YX}^3-\tau _{YX^{-1}}^3-\tau _{Y^{-1}X}^3-\tau _{Y^{-1}X^{-1}}^3+2 \tau _Y \tau _{X^{-1}} \tau _{YX^{-1}}^2-\tau _Y \tau _{YX} \tau _{Y^{-1}X}^2\\
&+2 \tau _{X^{-1}} \tau _{Y^{-1}} \tau _{Y^{-1}X^{-1}}^2-\tau _Y \tau _{YX^{-1}} \tau _{Y^{-1}X^{-1}}^2+\tau _Y \tau _{X^{-1}}^3 \tau _{Y^{-1}}+6 \tau _Y \tau _{Y^{-1}}+\tau _Y \tau _{X^{-1}} \tau _{Y^{-1}}^2 \tau _{YX}\\
&-\tau _Y^2 \tau _{X^{-1}} \tau _{YX}-3 \tau _{X^{-1}} \tau _{Y^{-1}} \tau _{YX}-\tau _Y^2 \tau _{X^{-1}}^2 \tau _{YX^{-1}}-\tau _{X^{-1}} \tau _{YX}^2 \tau _{YX^{-1}}-\tau _{X^{-1}}^2 \tau _{Y^{-1}} \tau _{YX^{-1}}\\
&-\tau _{Y^{-1}}^2 \tau _{YX} \tau _{YX^{-1}}+3 \tau _Y \tau _{YX} \tau _{YX^{-1}}-\tau _{X^{-1}} \tau _{Y^{-1}}^2 \tau _{Y^{-1}X}-\tau _{Y^{-1}} \tau _{YX}^2 \tau _{Y^{-1}X}-3 \tau _Y \tau _{X^{-1}} \tau _{Y^{-1}X}\\
&+\tau _Y^2 \tau _{X^{-1}} \tau _{Y^{-1}} \tau _{Y^{-1}X}-\tau _{X^{-1}}^2 \tau _{YX} \tau _{Y^{-1}X}-\tau _Y \tau _{Y^{-1}} \tau _{YX^{-1}} \tau _{Y^{-1}X}+6 \tau _{YX^{-1}} \tau _{Y^{-1}X}\\
&-\tau _Y \tau _{X^{-1}}^2 \tau _{Y^{-1}X^{-1}}-\tau _{X^{-1}}^2 \tau _{Y^{-1}}^2 \tau _{Y^{-1}X^{-1}}-\tau _{Y^{-1}} \tau _{YX^{-1}}^2 \tau _{Y^{-1}X^{-1}}-\tau _{X^{-1}} \tau _{Y^{-1}X}^2 \tau _{Y^{-1}X^{-1}}\\
&-\tau _Y \tau _{Y^{-1}} \tau _{YX} \tau _{Y^{-1}X^{-1}}+6 \tau _{YX} \tau _{Y^{-1}X^{-1}}+3 \tau _{X^{-1}} \tau _{YX^{-1}} \tau _{Y^{-1}X^{-1}}+\tau _Y \tau _{X^{-1}} \tau _{Y^{-1}} \tau _{YX^{-1}} \tau _{Y^{-1}X^{-1}}\\
&-\tau _Y^2 \tau _{Y^{-1}X} \tau _{Y^{-1}X^{-1}}+3 \tau _{Y^{-1}} \tau _{Y^{-1}X} \tau _{Y^{-1}X^{-1}}-\tau _{YX} \tau _{YX^{-1}} \tau _{Y^{-1}X} \tau _{Y^{-1}X^{-1}}-9
\end{align*}
For a proof, see for example Lemma 5 in \cite{Lawton'2007}.

Poisson brackets between generators are summarized in Table \ref{Tab:G1B1SL3}, where we have omitted generator $\tau_{XYX^{-1}Y^{-1}}$ because it belongs to the Poisson center of $\mathcal O[\mathcal M_{\Sigma_{1,1}}^{SL(3,\mathbb C)}]$, see \cite{Lawton'2009}.
\begin{table}
\begin{adjustbox}{width=\columnwidth,center}
$
\begin{array}{c|ccc}
&\tau_X&\tau_Y&\tau_{X^{-1}}\\
\hline
\tau_X& 0&*&*\\
\tau_Y& \frac{\tau _X \tau _Y}{3}-\tau _{YX}&0&* \\
\tau_{X^{-1}}& 0 & \frac{1}{3} \tau _Y \tau _{X^{-1}}-\tau _{YX^{-1}}&0 \\
\tau_{Y^{-1}}& \tau _{Y^{-1}X}-\frac{1}{3} \tau _X \tau _{Y^{-1}} & 0 & \frac{1}{3} \tau _{X^{-1}} \tau _{Y^{-1}}-\tau _{Y^{-1}X^{-1}} \\
\tau_{YX}& \tau _Y \tau _{X^{-1}}-\frac{2}{3} \tau _X \tau _{YX}-\tau _{YX^{-1}} & -\tau _X \tau _{Y^{-1}}+\frac{2}{3} \tau _Y \tau _{YX}+\tau _{Y^{-1}X}& \tau _Y-\frac{1}{3} \tau _{X^{-1}} \tau _{YX} \\
\tau_{YX^{-1}}& \frac{1}{3} \tau _X \tau _{YX^{-1}}-\tau _Y & \tau _{X^{-1}} \tau _{Y^{-1}}-\frac{2}{3} \tau _Y \tau _{YX^{-1}}-\tau _{Y^{-1}X^{-1}} & -\tau _X \tau _Y+\tau _{YX}+\frac{2}{3} \tau _{X^{-1}} \tau _{YX^{-1}} \\
\tau_{Y^{-1}X}& -\tau _{X^{-1}} \tau _{Y^{-1}}+\frac{2}{3} \tau _X \tau _{Y^{-1}X}+\tau _{Y^{-1}X^{-1}} & \tau _X-\frac{1}{3} \tau _Y \tau _{Y^{-1}X} & \frac{1}{3} \tau _{X^{-1}} \tau _{Y^{-1}X}-\tau _{Y^{-1}} \\
\tau_{Y^{-1}X^{-1}}& \tau _{Y^{-1}}-\frac{1}{3} \tau _X \tau _{Y^{-1}X^{-1}} & \frac{1}{3} \tau _Y \tau _{Y^{-1}X^{-1}}-\tau _{X^{-1}} & \tau _X \tau _{Y^{-1}}-\tau _{Y^{-1}X}-\frac{2}{3} \tau _{X^{-1}} \tau _{Y^{-1}X^{-1}}
\end{array}
$
\end{adjustbox}
\vspace{0.4cm}

\begin{adjustbox}{width=\columnwidth,center}
$
\begin{array}{c|cc}
&\tau_{Y^{-1}}&\tau_{YX}\\
\hline
\tau_{YX} & \frac{1}{3} \tau _{Y^{-1}} \tau _{YX}-\tau _X&0\\
\tau_{YX^{-1}}& \tau _{X^{-1}}-\frac{1}{3} \tau _{Y^{-1}} \tau _{YX^{-1}}&-\tau _Y^2+\tau _X \tau _{X^{-1}} \tau _{Y^{-1}}+\tau _{Y^{-1}}-\frac{1}{3} \tau _{YX} \tau _{YX^{-1}}-\tau _{X^{-1}} \tau _{Y^{-1}X}-\tau _X \tau _{Y^{-1}X^{-1}}\\
\tau_{Y^{-1}X}& \tau _X \tau _Y-\tau _{YX}-\frac{2}{3} \tau _{Y^{-1}} \tau _{Y^{-1}X} &-\tau _{X^{-1}}-\tau _Y \tau _{X^{-1}} \tau _{Y^{-1}}+\tau _{Y^{-1}} \tau _{YX^{-1}}+\frac{1}{3} \tau _{YX} \tau _{Y^{-1}X}+\tau _Y \tau _{Y^{-1}X^{-1}}+\tau _X^2\\
\tau_{Y^{-1}X^{-1}}& -\tau _Y \tau _{X^{-1}}+\tau _{YX^{-1}}+\frac{2}{3} \tau _{Y^{-1}} \tau _{Y^{-1}X^{-1}}&0
\end{array}
$
\end{adjustbox}
\vspace{0.4cm}

\begin{flushleft}
$
\begin{array}{c|ccc}
&\tau_{YX^{-1}}&\tau_{Y^{-1}X}&\tau_{Y^{-1}X^{-1}}\\
\hline
\tau_{YX^{-1}}&0&0&\tau _{X^{-1}}^2-\tau _X \tau _Y \tau _{Y^{-1}}+\tau _{Y^{-1}} \tau _{YX}+\tau _Y \tau _{Y^{-1}X}+\frac{1}{3} \tau _{YX^{-1}} \tau _{Y^{-1}X^{-1}}-\tau _X\\
\tau_{Y^{-1}X}&*&0&\tau _X \tau _Y \tau _{X^{-1}}-\tau _{X^{-1}} \tau _{YX}-\tau _X \tau _{YX^{-1}}-\frac{1}{3} \tau _{Y^{-1}X} \tau _{Y^{-1}X^{-1}}-\tau _{Y^{-1}}^2+\tau _Y\\
\tau_{Y^{-1}X^{-1}}&*&*&0
\end{array}
$
\end{flushleft}
\vspace{0.4cm}

\caption{Brackets between generators of $\mathcal O[\mathcal M^{SL(3)}_{\Sigma_{1,1}}]$}
\label{Tab:G1B1SL3}
\end{table}

Coordinate ring $\mathcal O[\mathcal M_{\Sigma_{1,1}}^{SL(3)}]$ has a Poisson center $Z$ of Krull dimension $\dim Z=2$. Two algebraically idependent Casimir elements: $z_1=\tau_{YXY^{-1}X^{-1}}$ and $z_2=\tau_{XYX^{-1}Y^{-1}}$.
The formula expressing $z_1$ in terms of generators is given in (\ref{eq:G1B1SL3Zeta1})

\subsubsection{A superintegrable system} Here we consider a superintegrable system defined by a single cycle homotopic to $X$. By (\ref{eq:subgroupC}) we have a chain of inclusions of subgroups of the fundamental group
\begin{align}
\pi_1(C)=\langle X\rangle\quad\subset\quad \pi_1(\Sigma_{1,1},p\,|\,C)=\langle X,YXY^{-1}\rangle\quad\subset\quad \pi_1(\Sigma_{1,1},p)=\langle X,Y\rangle.
\label{eq:SubgroupsG1B1SL3}
\end{align}

Choose two independent  Hamiltonians in $\mathcal O[\mathrm{Hom}(\pi_1(C), SL(3)]^{SL(3)}$ as
\begin{align*}
H_1=\tau_X,\qquad H_2=\tau_{X^{-1}}
\end{align*}
These two Poisson commuting Hamiltonians define a superintegrable system with the
Poisson algebra of first integrals $J$ defined above. We can choose a maximal algebraically idependent subset in $J$ to be $z_1,z_2,H_1,H_2$  with two more integrals
\begin{align*}
g_1=&\mathrm{Tr}\,(YXY^{-1}X)=\tau _{X^{-1}}-\tau _Y \tau _{X^{-1}} \tau _{Y^{-1}}+\tau _{Y^{-1}} \tau _{YX^{-1}}+\tau _{YX} \tau _{Y^{-1}X}+\tau _Y \tau _{Y^{-1}X^{-1}},\\
g_2=&\mathrm{Tr}\,(YX^{-1}Y^{-1}X^{-1})=\tau _X-\tau _X \tau _Y \tau _{Y^{-1}}+\tau _{Y^{-1}} \tau _{YX}+\tau _Y \tau _{Y^{-1}X}+\tau _{YX^{-1}} \tau _{Y^{-1}X^{-1}}.
\end{align*}


\section{Conclusion}

\subsection{The algebra of chord diagrams on a surface and "universal superintegrable systems"}

The algebra of chord diagrams is a universal model for the Poisson algebra of functions on
moduli spaces of flat connections \cite{AndersenMattesReshetikhin'1996}, see also \cite{Turaev'1991} for the case of loop algebras on a surface, i.e. skein modules. Chord diagrams first appeared in the setting of finite type invariants of knots \cite{Vassiliev'1990} and
in perturbative Chern-Simons invariants \cite{Bar-Natan'1991}\cite{Kontsevich'1993}. They became an important tool in the related theory of Vassiliev invariants.

Let us recall some basic definitions. A {\it chord diagram} is a homotopy class of a graph on a surface which consists of solid lines and chords. Solid lines are oriented, cords are not oriented. They satisfy the analogue of Reidemeister relations (see \cite{AndersenMattesReshetikhin'1996} for details) together with the so-called 4T-relation shown on Figure \ref{fig:4TRelation}.
\begin{figure}
\begin{tikzpicture}
\draw (-1.75,0) node {$-$};
\draw (0,0) [dashed] circle (1);
\draw[thick,->] (180+45:1) to (45:1);
\draw[thick,->] (-45:1) to (90+45:1);
\draw[thick,->] (-70:1) to (70:1);
\fill (0,0) circle (0.07);
\fill (45:0.48) circle (0.07);
\begin{scope}[shift={(3,0)}]
\draw (-1.5,0) node {$+$};
\draw (0,0) [dashed] circle (1);
\draw[thick,->] (180+45:1) to (45:1);
\draw[thick,->] (-45:1) to (90+45:1);
\draw[thick,->] (180+70:1) to (180-70:1);
\fill (0,0) circle (0.07);
\fill (180+45:0.48) circle (0.07);
\end{scope}
\begin{scope}[shift={(6,0)}]
\draw (-1.5,0) node {$+$};
\draw (0,0) [dashed] circle (1);
\draw[thick,->] (180+45:1) to (45:1);
\draw[thick,->] (-45:1) to (90+45:1);
\draw[thick,->] (-70:1) to (70:1);
\fill (0,0) circle (0.07);
\fill (-45:0.48) circle (0.07);
\end{scope}
\begin{scope}[shift={(9,0)}]
\draw (-1.5,0) node {$-$};
\draw (0,0) [dashed] circle (1);
\draw[thick,->] (180+45:1) to (45:1);
\draw[thick,->] (-45:1) to (90+45:1);
\draw[thick,->] (180+70:1) to (180-70:1);
\fill (0,0) circle (0.07);
\fill (180-45:0.48) circle (0.07);
\draw (1.7,0) node {$=0$};
\end{scope}
\begin{scope}[shift={(3,3)}]
\draw (0,0) [dashed] circle (1);
\draw[thick,->] (180+45:1) to (45:1);
\draw[thick,->] (-45:1) to (90+45:1);
\fill (0,0) circle (0.07);
\end{scope}
\begin{scope}[shift={(6,3)}]
\draw (-1.5,0) node {$:=$};
\draw (0,0) [dashed] circle (1);
\draw[thick,->] (180+45:1) to (45:1);
\draw[thick,->] (-45:1) to (90+45:1);
\draw[thick,decorate,decoration={snake,amplitude=.4mm,segment length=2mm,post length=0mm}] (180+45:0.67) to (-45:0.67);
\end{scope}
\end{tikzpicture}
\caption{$4T$ relation in chord diagrams}
\label{fig:4TRelation}
\end{figure}
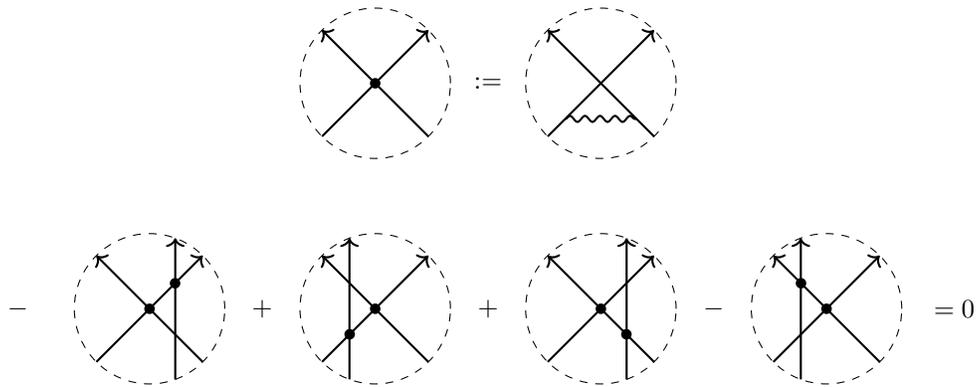

The {\it algebra of chord diagrams} $Ch(\Sigma)$ on a surface $\Sigma$ is a vector space spanned by $\ZZ$-linear combinations of chord diagrams with
\begin{itemize}

\item the commutative multiplication of $[D_1]$ and $[D_2]$ is defined as $[D_1][D_2]=[D_1\sqcup D_2]$, i.e. as the homotopy class of their disjoint union.
T
\item The Poisson bracket is a "universal" version of the Atiyah and Bott Poisson
bracket.  For two chord diagrams $D_1$ and $D_2$ their Poisson bracket is
\[
\{[D_1], [D_2]\}=\sum_{p\in D_1\cap D_2} \sum_{p\in D_1\cap D_2} \epsilon(p) [D_1\#_p D_2]
\]
where the sum is taken over the intersection points of representatives $D_1$ and $D_2$ and $D_1\#_p D_2$
is the result of the insertion of a chord in a small vicinity of point $p$. It is clear that the right side
does not depend on the choice of representatives and on where the extra cord in inserted.

\end{itemize}

The Poisson center of $Z(\Sigma)\subset Ch(\Sigma)$ is, conjecturally, generated by chord diagrams which are
contractible to the boundary of $\Sigma$.

To define a conjecturally superintegrable system on $Ch(\Sigma)$ choose a system of simple
curves on the surface. Define the subalgebra of Hamiltonians $B(C)$ and the subalgebra of chord diagrams
which are contractible to $C$ and the subalgebra of first integrals $J(C)$ and the subalgebra of chord diagrams
which can be separated from $C$. We have a natural inclusion of Poisson subalgebras:
\[
Z(\Sigma)\subset B(C)\otimes Z(\Sigma)\subset J(C)\subset Ch(\Sigma)
\]
We conjecture that this chain of inclusions define a superintegrable systems.
In this infinite dimensional setting the superintegrability means that the space of all nontrivial Poisson derivations $Der_{J(C)}(Ch(\Sigma))$
of $Ch(\Sigma)$ which act trivially on $J(C)$ is generated as an $Ch(\Sigma)$-module by Hamilton derivations corresponding to elements from $B(C)$.

\subsection{Quantization and the algebra of links in a cylinder} The natural question about any integrable system is
how to quantize it. In the case of superintegrable systems on chord diagrams there is a natural quantization.

There is a natural associative algebra that can be naturally associated to a surface.
The his is the algebra of links, or linked graphs, in $\Sigma\times I$ where $I$ is an interval.
This algebra, which we will denote $A(\Sigma)$ is the space of $\ZZ$-linear combinations of
links in $\Sigma\times I$. The natural associative multiplication on this space is the "placing of one link
on the top of the other":
\[
[L_1]\ast [L_2]=[L_1\sqcup L2]
\]
where in the right side of this equation we assume  that $L_1\subset \Sigma[1,1/2]$ and $L_2\subset \Sigma\times [1/2,0]$.
here $[L]$ is the topological link which the isotopy class of a geometrical link $L: {S^1}^{\times n}\to \Sigma\times I$.

The algebra $A(\Sigma)$ has a natural filtration with $Ch(\Sigma)$ being isomorphic to its associated graded
algebra. In this sense $A(\Sigma)$ is the quantization of $Ch(\Sigma)$.

Finite dimensional representations of the algebra $A(\Sigma)$ can be easily constructed using \cite{RT}.
Such representations use representation theory of $U_q({\mathfrak g})$ at roots of unity.
We will present details in a separate publication.

For generic $q$, the quantization can be done by quantizing Fock and Rosly brackets by $R$-matrices, see for example \cite{AlekseevGrosseSchomerus'1995,AlekseevSchomerus'1996}. It is worth noting that a closely related question for the extended moduli space with the choice of flags at marked points on the boundary can be studied by means of Cluster Algebras \cite{FockGoncharov'2006}, see for example \cite{SchraderShapiro'2017,GekhtmanShapiroVainshtein'2018,ChekhovMazzoccoRubtsov'2017} and references therein.


\begin{thebibliography}{GHJW97}

\bibitem[AGS95]{AlekseevGrosseSchomerus'1995}
A.~Alekseev, H.~Grosse, and V.~Schomerus
\newblock Combinatorial quantization of the Hamiltonian Chern-Simons theory I.
\newblock {\em Communications in Mathematical Physics}, 172(2), 317-358, 1995.

\bibitem[AK00]{AlekseevKosmann-Schwarzbach'2000}
A.~Alekseev, Y.~Kosmann-Schwarzbach.
\newblock {M}anin pairs and moment maps.
\newblock {\em {J}ournal of Differential geometry}, 56(1): 133-165, 2000.

\bibitem[AKM02]{AlekseevKosmann-SchwarzbachMeinrenken'2002}
A.~Alekseev, Y.~Kosmann-Schwarzbach, E.~Meinrenken.
\newblock {Q}uasi-{P}oisson manifolds.
\newblock {\em {C}anadian journal of mathematics}, 54(1), 3-29, 2002.

\bibitem[AMM98]{AlekseevMalkinMeinrenken'1998}
A.~Alekseev, A.~Malkin, and E.~Meinrenken.
\newblock {L}ie group valued moment maps.
\newblock {\em {J}ournal of {D}ifferential Geometry}, 48(3):445-495, 1998.

\bibitem[AR]{ArthamonovRoubtsov}
S.~Arthamonov, V.~Roubtsov.
\newblock {\em to appear}.

\bibitem[AS96]{AlekseevSchomerus'1996}
A.~Alekseev, V.~Schomerus
\newblock Representation theory of {C}hern-{S}imons observables.
\newblock {\em Duke Math. J.} 85(2), 447-510, 1996.

\bibitem[AO19]{AO} G. ~Arutyunov, E. ~Olivucci, Hyperbolic spin Ruijsenaars-Schneider model from Poisson reduction, arXiv:1906.02619.

\bibitem[AB83]{AtiyahBott'1983}
M.~F. Atiyah and R.~Bott.
\newblock The {Y}ang-{M}ills equations over {R}iemann surfaces.
\newblock {\em Phil. Trans. R. Soc. Lond. A}, 308(1505):523--615, 1983.

\bibitem[AMR96]{AndersenMattesReshetikhin'1996}
J.~E. Andersen, J.~Mattes, and N.~Reshetikhin.
\newblock The {P}oisson structure on the moduli space of flat connections and
  chord diagrams.
\newblock {\em Topology}, 35(4):1069--1083, 1996.

\bibitem[BN91]{Bar-Natan'1991}
D.~Bar-Natan,
\newblock Perurbative aspects of the Chern-Simons topological quantum field theory.
\newblock {\em PhD thesis}, Princeton University, Princeton, 1991.

\bibitem[BD82]{BelavinDrinfeld'1982}
A.~Belavin, and V.~Drinfel'd.
\newblock Solutions of the classical {Y}ang-{B}axter equation for simple {L}ie algebras. \newblock {\em Functional Analysis and Its Applications}, 16(3), 159-180, 1982.

\bibitem[BG93]{BiswasGuruprasad'1993}
I~Biswas and K~Guruprasad.
\newblock Principal bundles on open surfaces and invariant functions on lie
  groups.
\newblock {\em International Journal of Mathematics}, 4(04):535--544, 1993.



\bibitem[CF18]{CF} O. ~Chalykh, M. ~Fairon, On the Hamiltonian formulation of the trigonometric spin Ruijsenaars-Schneider system, arXiv:1811.08727.
    
\bibitem[CMR17]{ChekhovMazzoccoRubtsov'2017}
L.~Chekhov, M.~Mazzocco, and V.~Rubtsov, V,
\newblock Algebras of quantum monodromy data and decorated character varieties. 
\newblock {\em arXiv preprint} arXiv:1705.01447, 2017.
    
\bibitem[FK11]{FeherKlimcik'2011}
L.~Feh\'er, C.~Klim\v{c}\'ik Self-duality of the compactified Ruijsenaars-Schneider system from quasi-Hamiltonian reduction. Nuclear Physics B, 860(3), 464-515. (2012)

\bibitem[FK13]{FK} L.~Feh\'er, C.~Klim\v{c}\'ik, The Ruijsenaars self-duality map as a mapping class symplectomorphism,
 In: Lie Theory and Its Applications in Physics, IX International Workshop, ed. V. Dobrev, pp. 423-437, Springer, 2013,
 arXiv:1203.3300.
 
\bibitem[FG06]{FockGoncharov'2006}
V.~Fock, A.~Goncharov
\newblock Moduli spaces of local systems and higher Teichm\"uller theory. 
\newblock {\em Publications Math\'ematiques de l'IH\'ES,} 103, 1-211, 2006.

\bibitem[FGNR00]{FGNR}  V. ~Fock, A. ~Gorsky, N. ~Nekrasov, V. ~Rubtsov, Duality in Integrable Systems and Gauge Theories, JHEP 0007 (2000) 028, hep-th/9906235 .

\bibitem[F35] {F} V. ~Fock, Zur Theorie Des Wasserstoffatoms, Z. Physik 98, 145 (1935)

\bibitem[FR93]{FockRosly'1993}
V.~Fock and A.~Rosly
\newblock Flat connections and polyubles.
\newblock {\em Theoretical and Mathematical Physics,} 95(2), 526-534 (1993)

\bibitem[FR99]{FockRosly'1999}
V.~Fock and A.~Rosly.
\newblock {P}oisson structure on moduli of flat connections on {R}iemann
  surfaces and $r$-matrix.
\newblock In A.~Morozov and M.~Olshanetsky, editors, {\em {M}oscow seminar in
  mathematical physics}, pages 67--86. American Mathematical Soc., 1999.

\bibitem[FMSW65]{FMSW} J.~Frish, V.~Mandrosov,Y.A.~Smorodinsky, M.~Uhlir and P.
~Winternitz.
\newblock On higher symmetries in quantum mechanics
\newblock {\em Physics Letters} 16:354-356 (1965).

\bibitem[GSV18]{GekhtmanShapiroVainshtein'2018}
M.~Gekhtman, M.~Shapiro, and A.~Vainshtein,
\newblock Drinfeld double of $GL_n$ and generalized cluster structures. Proceedings of the London Mathematical Society, 116(3), 429-484. (2018)

\bibitem[GN94]{GN1} A. ~Gorsky and N. ~Nekrasov, Hamiltonian systems of Calogero type and two dimensional Yang-Mills theory, Nucl.Phys. B414 (1994) 213-238, hep-th/9304047.

\bibitem[GN95]{GN2} A. ~Gorsky and N. ~Nekrasov, Relativistic Calogero-Moser Model as a gauged WZW model, Nucl. Phys. B436, (1995) 582, hep-th/9401017.


\bibitem[GHJW97]{GuruprasadHuebschmannJeffreyWeinstein'1997}
K. ~Guruprasad, J. ~Huebschmann, L. ~Jeffrey, and A. ~Weinstein.
\newblock Group systems, groupoids, and moduli spaces of parabolic bundles.
\newblock {\em Duke Mathematical Journal}, 89(2):377--412, 1997.

\bibitem[Gol86]{Goldman'1986}
W.~Goldman.
\newblock Invariant functions on {L}ie groups and {H}amiltonian flows of
  surface group representations.
\newblock {\em Inventiones mathematicae}, 85(2):263--302, 1986.

\bibitem[H87]{H} N. ~Hitchin, "Stable bundles and integrable systems", Duke Mathematical Journal, 54 (1): 91–114, 1987.

\bibitem[KRWY]{KashaevReshetikhinWebsterYakimov}
R.~Kashaev, N.~Reshetikhin, B.~Webster, and M.~Yakimov.
\newblock {M}oduli {S}paces of {F}lat {$G$}-connections and {S}ymmetrically
  {F}actorizable {G}roups.

\bibitem[K93]{Kontsevich'1993}
M.~Kontsevich
\newblock Vassiliev's knot invariants.
\newblock {\em In I.M. Gelfand Seminar, Adv. Soviet Math.} v. 16, American Math Soc , Providence, 1993, 137-150.

\bibitem[Law07]{Lawton'2007}
S.~Lawton.
\newblock Generators, relations and symmetries in pairs of 3x3 unimodular matrices.
\newblock {\em Journal of Algebra}, 313(2), 782-801, 2007.

\bibitem[Law09]{Lawton'2009}
S.~Lawton.
\newblock Poisson geometry of $SL(3,\mathbb C)$-character varieties relative to a surface with boundary.
\newblock {\em Transactions of the American Mathematical Society}, 361(5), 2397-2429, 2009.

\bibitem[Meu06]{Meusberger'2006}
C.~Meusburger.
\newblock Dual generators of the fundamental group and the moduli space of flat connections. \newblock {\em Journal of Physics A: Mathematical and General}, 39(47), 14781, 2006.

\bibitem[MF78]{MishchenkoFomenko'1978}
A.~S. Mishchenko and A.~T. Fomenko.
\newblock {G}eneralized {L}iouville method of integration of {H}amiltonian
  systems.
\newblock {\em {F}unctional {A}nalysis and {I}ts {A}pplications},
  12(2):113--121, Apr 1978.


\bibitem[Nek72]{Nekhoroshev'1972}
N.~N. Nekhoroshev.
\newblock Action-angle variables, and their generalizations.
\newblock {\em Trans. Moscow Math. Soc.}, 26:181--198, 1972.

\bibitem[Nek96]{Nik} N. ~Nekrasov, Holomorphyc bundles and many-body systems, CMP, v. 180 (1996), 587-604.

\bibitem[Pa26]{Pa} W. ~Pauli, On the hydrogen spectrum from the standpoint of the new quantum mechanics, Zeitschrift fur Physik, 36, 336-363 (1926)

\bibitem[Res03a]{R3} N. ~Reshetikhin, Integrability of Characteristic Hamiltonian Systems on Simple Lie Groups with Standard Poisson Lie Structure, 2003, Volume 242, Issue 1--2, pp 1--29

\bibitem[Res03b]{Reshetikhin'2003}
N. ~Reshetikhin.
\newblock Degenerate integrability of the spin {C}alogero--{M}oser systems and
  the duality with the spin {R}uijsenaars systems.
\newblock {\em Letters in Mathematical Physics}, 63(1):55--71, 2003.

\bibitem[Res16]{Reshetikhin'2016}
N.~Reshetikhin.
\newblock Degenerately integrable systems.
\newblock {\em Journal of Mathematical Sciences}, 213.5 (2016): 769-785.

\bibitem[Res18]{R2} N.~Reshetikhin
\newblock Spin Calogero-Moser models on symmetric spaces.
\newblock {\em arXiv:1903.03685}.

\bibitem[RS19]{SR'2019}
N. ~Reshetikhin,J. ~Stokman.
\newblock Classical spin Calogero-Moser-Sutherland
type systems and boundary Knizhnik-Zamolodchikov equation.
\newblock {\em to appear}

\bibitem[RT90]{RT}
N.~Reshetikhin and V.~Turaev
\newblock Ribbon graphs and their invaraints derived from quantum groups.
\newblock {\em Communications in Mathematical Physics}, 127(1), 1-26 (1990)

\bibitem[SS17]{SchraderShapiro'2017}
G.~Schrader and A.~Shapiro
\newblock Continuous tensor categories from quantum groups I: algebraic aspects.
\newblock {\em arXiv preprint} arXiv:1708.08107 (2017)


\bibitem[SInt]{SInt} Superintegrability in Classical and Quantum Systems, Edited by: P. ~Tempesta, P. ~Winternitz, J. ~Harnad, W. ~Miller, Jr., G. ~Pogosyan, M. ~Rodriguez, CRM Proceedings and Lecture Notes, Volume: 37, 2004.



\bibitem[T91]{Turaev'1991}
V.~G.~Turaev,
\newblock Skein quantization of Poisson algebras of loops on surfaces.
\newblock In {\em Annales scientifiques de l'Ecole normale sup\'erieure} Vol. 24, No. 6, pp. 635-704 (1991).

\bibitem[V90]{Vassiliev'1990}
V.~A.~Vassiliev,
\newblock Cohomology of knot spaces
\newblock in {\em Theory of Singularities and Applications},
Adv. Soviet Math. v. 1, American Math Soc , Providence, 1990, 23-69.

\end{thebibliography}
\end{document}